\documentclass{lmcs}
\pdfoutput=1

\usepackage{lastpage}
\lmcsdoi{20}{3}{13}
\lmcsheading{}{\pageref{LastPage}}{}{}%
{Jun.~12,~2019}{Aug.~05,~2024}{}

\usepackage[utf8]{inputenc}

\usepackage{url}
\usepackage{comment}

\usepackage{stmaryrd}
\usepackage{amssymb}
\usepackage{xypic}
\usepackage{xspace}

\usepackage{xcolor}

\usepackage{mathpartir}
\usepackage{amsmath}
\usepackage{longtable}


\theoremstyle{plain} 

\usepackage{xspace}
\usepackage{ifthen}
\usepackage{amssymb}

\newif\ifdraft
 \drafttrue

\newcommand\DEF{\ \ ::=\ \ }
\newcommand\OR{\ \ |\ \ }

\newcommand{\set}[1]{\{#1\}}
\newcommand{\setbar}{\mathrel{|}}

\newcommand{\bnfis}{\ensuremath{::=}}
\newcommand{\bnfbar}{\ensuremath{\ \ |\ \ }}

\newcommand{\tree}[4][] {
	\def\arraystretch{1.2}
	\begin{array}{l}
		\ifthenelse { \equal {#1} {} }
		{}
		{
		#1
		\\
		}
		{\textstyle \frac{
			\begin{array}{c}
			#2
			\end{array}
		}{
			\def\arraystretch{1.2}
			\begin{array}{l}
				#3
			\end{array}
		}
		}
		\ifthenelse { \equal {#4} {} }
		{}
		{
		\ #4
		}
	\end{array}
}

\newcommand\andalso{\quad}

\newcommand{\btree}[4][]{
	\boxed{
		\tree[#1]{
			#2
		}{
			#3
		}{#4}
	}
}

\newcommand{\Proc}{\ensuremath{\mathcal{P}}\xspace}
\newcommand{\Buf}{\ensuremath{\mathcal{B}}\xspace}
\newcommand{\Net}{\ensuremath{\mathcal{N}}\xspace}

\newcommand\aggr[1]{\breve{#1}}
\newcommand\dual[1]{\overline{#1}}
\newcommand\pout[2]{#1{}_!\langle #2\rangle.}
\newcommand\pin[2]{{#1}{}_?(#2).}
\newcommand\pdin[3]{{#1}{}_?(#2)\langle#3\rangle.}

\newcommand{\default}{\ensuremath{\mathtt{df}}\xspace}

\newcommand\pll{\,|\,}
\newcommand{\Sum}{\ensuremath{+}\xspace}
\newcommand\sel[2]{#1\triangleleft #2.}
\newcommand\branchOn[2]{#1\,{\triangleright}\,#2}
\newcommand\branch[2]{#1\,{\triangleright}\,\{#2\}}
\newcommand\branchI[3]{\branch{#1}{{#2}_i: {#3}_i}_{i\in I}}

\newcommand\branchDef[3]{#1\,{\triangleright}\,\{#2, \default: #3\}}

\newcommand{\ndet}[2]{#1 \Sum #2}
\newcommand\rec[1]{\mu #1.}
\newcommand\recvar[1]{X}

\newcommand{\defin}{\ensuremath{\mathsf{def}}\xspace}
\newcommand{\Def}[2]{\ensuremath{\defin\ #1\ \mathsf{in}\ #2} \xspace}
\newcommand{\appl}[2]{\ensuremath{#1\langle#2\rangle}\xspace}
\newcommand{\abs}[2]{\ensuremath{#1(#2)}\xspace}

\newcommand{\defeq}{\ensuremath{\stackrel{\mathsf{def}}{=}}\xspace}

\newcommand{\If}{\ensuremath{\mathsf{if}}\xspace}
\newcommand{\Then}{\ensuremath{\mathsf{then}}\xspace}
\newcommand{\Else}{\ensuremath{\mathsf{else}}\xspace}
\newcommand\cond[3]{\If\,#1\,\Then\,#2\,\Else\,#3}
\newcommand\request[2]{#1_!(#2).}
\newcommand\accept[2]{#1_?(#2).}
\newcommand\nil{\mathbf{0}}

\newcommand{\newn}[1]{(\nu\,#1)}
\newcommand{\newnp}[2]{\newn{#1}(#2)}

\newcommand{\net}[1]{[\,#1\,]}
\newcommand{\npll}{\,|\!|\,}

\newcommand\condtrue{\Vdash}

\newcommand{\econd}{\ensuremath{\varphi}\xspace}

\newcommand{\Par}[3]{\prod_{#1 \in #2} #3}

\newcommand{\dom}[1]{\mathrm{dom}(#1)}

\newcommand{\eval}[1]{\ensuremath{#1\downarrow}\xspace}
\newcommand{\mult}{\ensuremath{\odot}\xspace}
\newcommand{\unit}{\ensuremath{\mathbf{1}}}
\newcommand{\exception}{\ensuremath{\mathbf{exc}}}

\newcommand{\true}{\ensuremath{\mathtt{true}}\xspace}
\newcommand{\false}{\ensuremath{\mathtt{false}}\xspace}

\newcommand\subst[2]{\{{#1}/{#2}\}}

\newcommand\reduces{\mathrel{\longrightarrow}}

\newcommand{\gthr}[2]{\ensuremath{\mathsf{V}(#1, #2)}\xspace}
\newcommand{\newbuffer}[2]{\ensuremath{\mathsf{B}(#1, #2)}\xspace}

\newcommand{\ebuffer}{\ensuremath{\varepsilon}\xspace}

\newcommand{\squeue}[3]{\ensuremath{#1[#2, #3]}}
\newcommand{\pol}[1]{\ensuremath{\tilde{#1}}}
\newcommand{\cat}{\cdot}
\newcommand{\equeue}{\ensuremath{\varepsilon}\xspace}

\newcommand\tout[1]{{}_!#1.}
\newcommand\tin[1]{{}_?#1.}
\newcommand\tsel[1]{\oplus\set{#1}}
\newcommand\tselI[2]{\tsel{{#1}_i: {#2}_i}_{i\in I}}
\newcommand\tbranchOn[1]{\&\{#1\}}
\newcommand\tbranch[2]{\&\{#1\mathrel:#2\}}
\newcommand\tbranchI[2]{\&\{{#1}_i\mathrel:{#2}_i\}_{i\in I}}
\newcommand\tend{\mathsf{end}}
\newcommand\tvar[1]{#1}
\newcommand\trec[1]{\mu\,\tvar#1.}
\newcommand\tdual[1]{\overline{#1}}
\newcommand\tadvance{\mathrel{\shortrightarrow}}

\newcommand{\tempty}{\ensuremath{\varepsilon}\xspace}

\newcommand{\econtext}{\ensuremath{\emptyset}\xspace}

\newcommand{\ctype}[2]{\ensuremath{(#1, #2)}}

\newcommand\msel[1]{{}_{\oplus}#1}
\newcommand{\mtout}[1]{{}_{!}#1}
\newcommand{\sessionop}{\ensuremath{\circ}\xspace}

\newcommand{\synchronise}[2]{\ensuremath{#1 \mathrel{\hookrightarrow} #2}}

\newcommand{\fn}[1]{\mathrm{fn}(#1)}
\newcommand{\bn}[1]{\mathrm{bn}(#1)}
\newcommand{\nam}[1]{\mathrm{n}(#1)}
\newcommand{\fs}[1]{\mathrm{fs}(#1)}
\newcommand{\fv}[1]{\mathrm{fv}(#1)}

\newcommand\seq[1]{\widetilde{#1}}
\newcommand\types{\mathrel{\vdash}}
\newcommand\etypes{\mathrel{\succ}}

\newcommand\Figure[1]{Fig.~\ref{fig:#1}}

\newcommand\Definition[1]{Definition~\ref{#1}}
\newcommand\Section[1]{Section~\ref{sec:#1}}
\newcommand\IF{\mathrel{\mbox{if}}}
\newcommand\AND{\mathrel{\mbox{and}}}

\newcommand{\bool}{\mathsf{bool}}

\def\infrulestyle{\displaystyle}
\def\makeinfrulesmall{\def\infrulestyle{\textstyle}}

\newcommand\infrule[2]{{\infrulestyle\frac{#1}{#2}}}

\newcommand\rname[1]{\ensuremath{\scriptstyle \textsc{{[#1]}}}\xspace}
\newcommand\rulename[1]{\rname{#1}}

\newcommand\sname[1]{\ensuremath{\scriptstyle \textsf{{(#1)}}}\xspace}

\newcommand{\SharedC}{\sname{ShCh}}
\newcommand{\SessionC}{\sname{SCh}}
\newcommand{\Labels}{\sname{Labels}}
\newcommand{\Counter}{\sname{Cnt}}
\newcommand{\Variables}{\sname{Var}}
\newcommand{\Expressions}{\sname{Exp}}

\newcommand{\Endpoint}{\sname{Endp}}
\newcommand{\Session}{\sname{Sess}}
\newcommand{\Identifier}{\sname{Identifier}}

\newcommand{\Message}{\sname{Msg}}
\newcommand{\TMessage}{\sname{TMsg}}

\newcommand{\Request}{\sname{Req}}
\newcommand{\Accept}{\sname{Acc}}
\newcommand{\Send}{\sname{Snd}}
\newcommand{\Rcv}{\sname{Rcv}}

\newcommand{\Selection}{\sname{Sel}}
\newcommand{\Branching}{\sname{Bra}}
\newcommand{\Conditional}{\sname{Cond}}
\newcommand{\NDet}{\sname{Sum}}
\newcommand{\Recursion}{\sname{Def}}
\newcommand{\PVar}{\sname{PVar}}
\newcommand{\Inact}{\sname{Inact}}

\newcommand{\Node}{\sname{Node}}
\newcommand{\Parallel}{\sname{Par}}
\newcommand{\Restriction}{\sname{Restr}}

\newcommand{\Buffer}{\sname{Buffer}}

\newcommand{\SRecover}{\sname{Recov}}

\newcommand{\Connect}{\rname{Conn}}
\newcommand{\Broadcast}{\rname{Bcast}}
\newcommand{\Unicast}{\rname{Ucast}}
\newcommand{\Receive}{\rname{Rcv}}

\newcommand{\Gather}{\rname{Gthr}}
\newcommand{\Select}{\rname{Sel}}
\newcommand{\Branch}{\rname{Bra}}
\newcommand{\Loss}{\rname{Loss}}
\newcommand{\Recover}{\rname{Rec}}
\newcommand{\BRecover}{\rname{BRec}}

\newcommand{\True}{\rname{True}}
\newcommand{\False}{\rname{False}}
\newcommand{\NonDet}{\rname{NDet}}
\newcommand{\RDef}{\rname{Def}}
\newcommand{\RPar}{\rname{RPar}}
\newcommand{\RCong}{\rname{RCong}}
\newcommand{\RRes}{\rname{RRes}}

\newcommand{\Cond}{\rname{Cond}}

\newcommand{\SWk}{\rname{SWk}}

\newcommand{\BEmp}{\rname{BEmp}}
\newcommand{\SEmp}{\rname{SEmp}}
\newcommand{\SExp}{\rname{SExp}}
\newcommand{\SLab}{\rname{SLab}}
\newcommand{\LExp}{\rname{LExp}}
\newcommand{\BPar}{\rname{BPar}}

\newcommand{\TReq}{\rname{TReq}}
\newcommand{\TAcc}{\rname{TAcc}}
\newcommand{\TSnd}{\rname{TSnd}}
\newcommand{\TRcv}{\rname{TRcv}}

\newcommand{\TSel}{\rname{TSel}}
\newcommand{\TBr}{\rname{TBr}}

\newcommand{\TCond}{\rname{TCond}}

\newcommand{\TInact}{\rname{TInact}}
\newcommand{\TVar}{\rname{TVar}}
\newcommand{\TRec}{\rname{TRec}}
\newcommand{\TSum}{\rname{TSum}}

\newcommand{\TNode}{\rname{TNode}}
\newcommand{\TSynch}{\rname{TSynch}}
\newcommand{\TPar}{\rname{TPar}}
\newcommand{\TCRes}{\rname{TCRes}}
\newcommand{\TSRes}{\rname{TSRes}}

\def\nname#1{\textsc{#1}\xspace}
\newcommand{\Nbrc}{\nname{Brc}}
\newcommand{\Nuni}{\nname{Uni}}
\newcommand{\Ngth}{\nname{Gth}}
\newcommand{\Nrcv}{\nname{Rcv}}
\newcommand{\Nsel}{\nname{Sel}}
\newcommand{\Nbra}{\nname{Bra}}

\newcommand{\mytype}[1]{\ensuremath{\mathtt{#1}}\xspace}
\newcommand{\nat}{\mytype{nat}}

\newcommand{\myprocess}[1]{\ensuremath{\mathsf{#1}}\xspace}

\newcommand{\Sender}{\myprocess{Sender}}
\newcommand{\Medium}{\myprocess{Medium}}
\newcommand{\Receiver}{\myprocess{Receiver}}

\newcommand{\map}[1]{[\![#1]\!]}
\newcommand{\pmap}[2]{[\![#1]\!]^{#2}}

\newcommand\recover{\,\diamond\,}

\newcommand{\Heartbeat}{\myprocess{Heartbeat}}
\newcommand{\heartbeat}{\ensuremath{\mathsf{hbt}}\xspace}

\newcommand{\theartbeat}{\mytype{b}}

\newcommand{\Network}{\myprocess{Network}}

\newcommand{\Recursive}{\myprocess{Recursive}}

\newcommand{\PreparePh}{\ensuremath{\mathsf{Prepare}}\xspace}
\newcommand{\AcceptPh}{\ensuremath{\mathsf{Accept}}\xspace}

\newcommand{\Proposer}{\myprocess{Proposer}}
\newcommand{\ProposerDef}[1]{\abs{\Proposer}{#1}}
\newcommand{\ProposerVar}[1]{\appl{\Proposer}{#1}}

\newcommand{\Acceptor}{\myprocess{Acceptor}}
\newcommand{\AcceptorDef}[1]{\abs{\Acceptor}{#1}}
\newcommand{\AcceptorVar}[1]{\appl{\Acceptor}{#1}}

\newcommand{\Acc}{\myprocess{Acc}}
\newcommand{\AccDef}[1]{\abs{\Acc}{#1}}
\newcommand{\AccVar}[1]{\appl{\Acc}{#1}}

\newcommand{\AccPh}{\myprocess{AcceptPhase}}
\newcommand{\AccPhDef}[1]{\abs{\AccPh}{#1}}
\newcommand{\AccPhVar}[1]{\appl{\AccPh}{#1}}

\newcommand{\Paxos}{\ensuremath{\mathsf{Paxos}}\xspace}
\newcommand{\PaxosDef}[1]{\abs{\Paxos}{#1}}
\newcommand{\PaxosVar}[1]{\appl{\Paxos}{#1}}
\newcommand{\PaxosNode}{\ensuremath{\mathsf{PaxosNode}}\xspace}

\newcommand{\restart}{\ensuremath{\mathsf{restart}}\xspace}
\newcommand{\acceptLabel}{\ensuremath{\mathsf{accept}}\xspace}

\newcommand{\chooseVal}[1]{\ensuremath{\mathsf{choose}(#1)}\xspace}

\newcommand{\messageTuple}{(\rnumber, \pvalue)}
\newcommand{\prep}{\mytype{prepare}}
\newcommand{\prom}{\mytype{promise}}

\newcommand{\rnumber}{\mytype{round}}
\newcommand{\pvalue}{\mytype{value}}
\newcommand{\id}{\ensuremath{\mathsf{id}}\xspace}
\newcommand{\emptyval}{\ensuremath{\epsilon}\xspace}

\newcommand{\PaxosType}{\ensuremath{\mathsf{PaxosType}}\xspace}

\newcommand{\UBSC}{\textsf{UBSC}\xspace}

\title[Session Types for Asynchronous Unreliable Broadcast Communication]{A Session Type System for Asynchronous Unreliable Broadcast Communication}

\author[D. Kouzapas]{Dimitrios Kouzapas \lmcsorcid{0000-0001-9300-0146}}[a]
\author[R. F. Gutkovas]{Ram\=unas Forsberg Gutkovas \lmcsorcid{0009-0003-8712-825X}}[b]
\author[A. L. Voinea]{\texorpdfstring{\linebreak}{}A. Laura Voinea \lmcsorcid{0000-0003-4482-205X}}[c]
\author[S. J. Gay]{Simon J. Gay \lmcsorcid{0000-0003-3033-9091}}[c]

\address{KIOS Centre of Excellence, University of Cyprus, Cyprus}
\address{Ericsson AB, Gothenburg, Sweden}
\address{School of Computing Science, University of Glasgow, UK}

\begin{document}

	\maketitle

	\begin{abstract}

	Session types are formal specifications of communication
	protocols, allowing protocol implementations to be
	verified by typechecking.
	Up to now, session type disciplines have assumed that the
	communication medium is reliable, with no loss of messages.
	However, unreliable broadcast communication is common in a
	wide class of distributed systems such as ad-hoc and wireless
	sensor networks.
	Often such systems have structured communication patterns
	that should be amenable to analysis by means of session types,
	but the necessary theory has not previously been developed.
	%
	We introduce the Unreliable
	Broadcast Session Calculus,
	a process calculus with unreliable  broadcast communication,
	and equip it with a session type system that we show is sound.
	We capture two common operations, broadcast
	and gather, inhabiting dual session types.
	Message loss may lead to non-synchronised session endpoints.
	To further account for unreliability we provide with an
	autonomous recovery mechanism
	that does not require
	acknowledgements from session participants.
	Our type system ensures soundness, safety, and progress
	between the synchronised endpoints within a session.
	%
	We demonstrate the expressiveness of our framework
	by implementing Paxos, the textbook protocol for
	reaching consensus in an unreliable, asynchronous network.

	\end{abstract}


\section{Introduction}

Networks that use a shared, stateless communication medium, such
as wireless sensor networks or ad-hoc networks, are widely used.
They are complex in nature, and their design and verification is
a challenging topic. In the presence of unreliability, these
networks feature structured communication, which lends itself to
being formalised. The use of a formalism in a first step can
help in understanding and implementing network protocols, and
then move on to proving their correctness.

One such formalisation is session types~\cite{Honda1998}.
Session types specify structured communication (protocols) in
concurrent, distributed systems. They allow implementations of
such protocols to be verified by type-checking, for properties
such as session fidelity and deadlock freedom. A prolific topic of
research, session types have found their way into many programming
languages and paradigms~\cite{BETTYWG3}, with several session type
technologies developed~\cite{BETTY:tools}. An assumption of
session types so far has been the reliability of communication.
That is, that messages are never lost and are always delivered to the
receiver. However, it is not realistic to assume reliability in
ad-hoc and wireless sensor networks. Such networks use a shared,
stateless communication medium and broadcast to deliver messages.
Broadcast collisions and the stateless nature of the communication
medium will occasionally lead to message loss.

In this paper we introduce 
{\em the Unreliable Broadcast Session Calculus} (\UBSC)
accompanied by a safe and sound session type system.
The semantics of the calculus are inspired by the practice
of ad-hoc and sensor networks. More concretely, the semantics
describe asynchronous broadcast and gather operations between
a set of network nodes operating in an unreliable context.
Unreliable communication allows for link failures and lost messages,
that may lead to situations where a network node is not synchronised
with the rest of the computation. To cope with these cases, the
semantics of the calculus offer flexible recovery mechanisms that
follow the practice of ad-hoc and sensor networks.

We develop a session type system based on binary session
types~\cite{Honda1998} for the \UBSC.
%
%
The syntax and the duality operator for the \UBSC
session types are identical to the syntax and duality operator
for standard binary session types.
%
However, a non-trivial type system 
describes the interaction between multiple nodes,
in contrast to binary session types that describe
interaction only between two processes.
Moreover, the type system, ensures session safety and
soundness of the communication
interaction and of the recovery mechanisms.
The main idea for the type system is to interpret
the case of failure and message loss
as a node
that is not synchronised with the overall session type protocol.
%
The session
type system ensures that the conditions for recovery are
adequate to support safe session interaction.
The type system is proved to be sound via a type preservation
theorem and safe via a type safety theorem stating that a 
well-typed network will never reduce to an undesirable/error
state.
The type system also ensures progress within a session, under
standard conditions of non-session interleaving.
An additional progress result states that every non-synchronised
node within a typed process may eventually recover.

Notably, our system does not introduce the description
of failure and recovery at the type level as
in~\cite{10.1007/978-3-319-60225-7_1,CarboneHY08,CGY2014}.
Our approach assumes that every interaction described at the
type level can fail and can lead to non-synchronised session endpoints.
Moreover, all non-synchronised endpoints can eventually recover. 
In~\cite{10.1007/978-3-319-60225-7_1,CarboneHY08,CGY2014},
non-synchronised endpoints (or multiparty roles)
recover at runtime using safe recovery as permitted by the 
session type description. Thus, the description of recovery
actions at each communication interaction would require long and
tedious description of protocols.
As shown in Section~\ref{subsec:os}, a second advantage of our approach
is the lack of global, and often complicated and unnatural, synchronisation
between failed session endpoints (or multiparty roles) in order for the
session to safely recover.


\subsection{Unreliable Broadcast Session Communication}

The communication assumptions of the \UBSC are based on
the practice of networks such as ad-hoc and wireless networks,
and Ethernet networks.
Specifically, the semantics of the calculus are justified by
the following assumptions:
\begin{itemize}[align=left]
	\item[A1.]	The network nodes operate within a shared,
			stateless communication medium.
			The term broadcast in these networks is understood
			when a node transmits a message within the shared
			communication medium, which is then ``instantly''
			received by a set of nodes that share the
			communication medium.

	\item[A2.]	The communication medium is considered unreliable,
			therefore transmitted messages may fail to be delivered
			to some or to all the receivers.
			In practice, unreliability arises for various reasons,
			e.g.~weak transmission signal or message collisions.

	\item[A3.]	A network node transmits a message only once;
			we assume that a message is received at most once,
			i.e.~there is no duplication of successfully
			received messages.

	\item[A4.]	There is no mechanism to acknowledge a successful
			message reception,
			i.e.~upon reception a receiver network node does
			not reply with an acknowledgement message.
			Lack of acknowledgement messages, implies that a
			transmitting network node cannot possibly know the
			subset of network nodes that successfully
			received the transmitted message.

	\item[A5.]	Messages are either lost, or delivered without being corrupted.
			This can be achieved by assuming an underlying mechanism that
			detects, e.g.~through parity bit comparison, and rejects 
			all corrupted messages. Rejected messages are considered lost.

	\item[A6.]	The network nodes operate at an arbitrary speed. Each network
			node decides to perform an interaction (either transmission, or
			internal processing) at an arbitrary moment.
\end{itemize}


The semantics of the \UBSC
deploy the mechanisms that support safe session interaction
and, at the same time, respect assumptions A1-A6. Specifically, the
semantics of the calculus deploy mechanisms that respect
the following session type principles:
%
\begin{itemize}[align=left]
	\item[S1.]	
			We assume binary interaction within a session.
			The interaction within a session name $s$ is
			defined between 
			a $\aggr s$-endpoint, uniquely used by a single network
			node, and a $s$-endpoint shared by an arbitrary
			number of network nodes.

			The one to many correspondence between endpoints
			gives rise to a broadcasting operation, as
			imposed by assumption A1,
			where the $\aggr s$-endpoint
			broadcasts a value towards the $s$-endpoints, and
			a gather operation, where the $\aggr s$-endpoint
			gathers messages sent from the $s$-endpoints.

	\item[S2.]
			Safe session interaction requires that messages are
			delivered in the order they are transmitted.
			Ordered message delivery is a direct consequence of
			assumptions A1, A3, and A5.

	\item[S3.]	Session interaction is subject to unreliability.
			Assumptions A2 and A6 imply that session endpoints
			may become non-synchronised with the overall session
			interaction.

	\item[S4.]
			Safe session interaction requires
			that two session endpoints can communicate
			whenever they are synchronised.

	\item[S5.]
			To ensure progress, endpoints that cannot progress, e.g.~because
			the endpoint is not synchronised or because the opposing endpoint
			is deadlocked, need to be able to autonomously recover.
			Autonomous recovery is a consequence of assumptions A4 and A6
			that imply that a network node does not have global information
			whether an endpoint can progress or not, a behaviour typical
			for ad-hoc and wireless sensor networks.

\end{itemize}

Below, we use a simple example to introduce the basic
assumptions of the \UBSC.
Prior to the example, consider the following
informal presentation of the syntax of the calculus.
The calculus defines the syntax for {\em network nodes}.
Specifically, a network node,
\[
	\textstyle
	N = \net{P \pll \prod_{i \in I} \squeue{s_i}{c}{\tilde{m}_i}}
\]
composes a binary session $\pi$-calculus~\cite{Honda1998} process $P$
together with a finite parallel composition of session buffer terms,
$\prod_{i \in I} \squeue{s_i}{c}{\tilde{m}_i}$.

A session $\pi$-calculus process of the form $\pout{s}{v} P$ denotes a process
that is ready to send value $v$ via channel (or session) endpoint $s$ and proceed
as process $P$. Dually, a process of the form $\pin{s}{x} P$ denotes a process
that is ready to receive a value on channel endpoint $s$ and substitute it on variable
$x$ within process $P$.
A buffer term, $\squeue{s}{c}{\tilde{m}}$, represents
a first-in first-out message buffer that interacts on session endpoint $s$.
The buffer stores messages $\tilde{m}$ and keeps track of the session
endpoint state using integer counter $c$.
Multiple network nodes can be composed in parallel
$N_1 \npll \dots \npll N_n$ to form a network. We often
write term,
$\prod_{j \in J} \net{P_j \pll \prod_{i \in I_j} \squeue{s_{i}}{c}{\tilde{m}_{i}}}$
to represent a parallel composition of network nodes.

The next example demonstrates the semantics and basic typing
ideas for the \UBSC.
Variations of the example will be used as a running example
throughout the paper.
\begin{exa}[A simple Heartbeat Protocol]
	\label{ex:intro_heartbeat}
	The heartbeat protocol, is a simple sensor
	network protocol where one, or more, network nodes
	periodically broadcast a {\em heartbeat} message
	to signal that they are alive.
	In the \UBSC, a simple heartbeat
	interaction can be specified with the following network.
	\[
	\begin{array}{c}
		\Heartbeat =
				\net{\pout{\aggr s}{\heartbeat} P_0 \pll \squeue{\aggr s}{0}{\equeue}}
				\,\,\npll\,\,
				\prod_{i \in I}
				\net{\pin{s}{x} P_i \pll \squeue{s}{0}{\equeue}}
	\end{array}
	\]
%
	The network implements the requirements of S1 on session channel $s$;
	the $\aggr s$-endpoint is uniquely used by network node
	$\net{\pout{\aggr s}{\heartbeat} P_0 \pll \squeue{\aggr s}{0}{\equeue}}$,
	and the $s$-endpoint is shared among an arbitrary number of network nodes,
	represented by network
	$\prod_{i \in I} \net{\pin{s}{x} P_i \pll \squeue{s}{0}{\equeue}}$.

	Under the requirements of S2 for ordered delivery of messages,
	it is typical for nodes in ad-hoc and wireless networks,
	to deploy a received message buffer that follows a first-in
	first-out policy.
	The use of a message buffer leads to asynchronous communication
	semantics, where we first observe a message stored in a message buffer
	and then extracted from the buffer for processing.
	At each buffer there is also a counter which is increased
	with every endpoint interaction and keeps track of the
	session endpoint state.
	The counter is used by the semantics to maintain interaction
	between synchronised endpoints,	as required by S4.
	%
	In the example above,  each session endpoint is associated with a
	corresponding empty message buffer, $\squeue{s}{0}{\equeue}$,
	(similarly $\squeue{\aggr s}{0}{\equeue}$ for the $\aggr s$-endpoint).
	The state counter designates that all endpoints are in the $0$ state.
	%
	%

	A heartbeat (\heartbeat) message is broadcast on the
	$\aggr s$-endpoint and it is received by the $s$-endpoints,
	as expected by the requirements of S1.
	%
	As required by S3, the broadcast interaction is subject to
	unreliability and is captured by the reduction:
	\[
	\begin{array}{rcl}
		\Heartbeat \reduces  \Heartbeat' = &&
				\net{P_0 \pll \squeue{\aggr s}{1}{\equeue}}
		\\		&\npll&
				\prod_{j \in J}
				\net{\pin{s}{x} P_j \pll \squeue{s}{1}{\heartbeat}}
				\,\,\npll\,\, 
				\prod_{k \in K}
				\net{\pin{s}{x} P_k \pll \squeue{s}{0}{\equeue}}
	\end{array}
	\]
	where $I = J \cup K$ and $J, K$ are disjoint.
	The broadcast message is received instantly during the transmission
	time.
	Failure of reception 
	is modelled by the fact that
	only an arbitrary subset of the network nodes, indexed by set $J$,
	accept the heartbeat message.
	Each successful receiver uses the session queue to buffer the
	heartbeat message and model asynchronous communication.

	Following the requirements of S4, the reduction semantics,
	defined in Section~\ref{subsec:os}, allow for interaction only between
	the $\aggr s$-endpoint and the $s$-endpoints that are in
	the same state.
	%
	To maintain synchronisation, a  successful interaction increases
	the session state counter on each buffer by $1$. The network nodes
	that failed to receive the heartbeat messages, indexed by set $K$,
	do not update their session counter, therefore they are considered
	non-synchronised and cannot continue to safely interact with the
	$\aggr s$-endpoint.

	In practice, state counting implies an underlying mechanism
	where each transmitted message includes a header that tags
	the message with the corresponding session and state counter.
	A network node can only accept a message if it implements the
	corresponding session endpoint and the session endpoint is
	synchronised with the state of the message.
	In a different case the message is dropped
	and considered lost.

	Each receiver will then interact locally with its own queue
	and extract the message for processing. For example,
	network node, $\net{\pin{s}{x} P_q \pll \squeue{s}{1}{\heartbeat}}$
	will interact with its local buffer using the reduction:
	\[
	\begin{array}{rcl}
		\Heartbeat' \reduces &&
				\net{P_0 \pll \squeue{\aggr s}{1}{\equeue}}
				\,\,\npll\,\,  
				\net{P_q \subst{\heartbeat}{x} \pll \squeue{s}{1}{\equeue}}
		\\
				&\npll&  
				\prod_{j \in J\backslash\set{q}}
				\net{\pin{s}{x} P_j \pll \squeue{s}{1}{\heartbeat}}
				\,\,\npll\,\,
				\prod_{k \in K}
				\net{\pin{s}{x} P_k \pll \squeue{s}{0}{\equeue}}
	\end{array}
	\]
	In network $\Heartbeat'$, the nodes described by set $K$
	appear to be stuck because the receiving $s$-endpoints are
	non-synchronised with the $\aggr s$-endpoint.
	The requirements of S5 impose the development of recovery
	semantics, in order to observe session progress.
	For example, 
	recovery on network $\Heartbeat'$ is described by the interaction.
	\[
	\begin{array}{rcl}
		\Heartbeat' \reduces &&
				\net{P_0 \pll \squeue{\aggr s}{1}{\equeue}}
				\,\,\npll\,\, 
				\net{P_q \subst{\unit}{x} \pll \squeue{s}{1}{\equeue}}
		\\
				&\npll&  
				\prod_{j \in J}
				\net{\pin{s}{x} P_j \pll \squeue{s}{1}{\heartbeat}}
				\,\,\npll\,\,
				\prod_{k \in K\backslash\set{q}}
				\net{\pin{s}{x} P_k \pll \squeue{s}{0}{\equeue}}
	\end{array}
	\]
	Whenever a session endpoint is possibly unable to progress,
	the reduction semantics use a default expression, in this
	case the unit (\unit) expression, to substitute the receive
	variable.

	However, due to assumptions A4 and A6, a network node does
	not have knowledge whether an endpoint can progress or not.
	This situation also exists in the practice of ad-hoc and wireless
	sensor networks that approximate lack of progress using
	necessary but not sufficient conditions,
	combined with internal mechanisms, e.g.~timeout signals, to achieve
	recovery and progress. Here the conditions that are necessary but
	not sufficient to detect lack of progress
	are the receiving prefix and the empty $s$-buffer.

	The typing system ensures the conditions for a safe recovery.
	The full semantics of the calculus offer the programmer multiple
	and flexible recovery choices to deploy whenever the session safety
	conditions are met.
\end{exa}

\subsection{Related Work}
The \UBSC follows the principles of~\cite{Kouzapas2014}, where
the authors provided semantics to a syntax for an unreliable broadcasting session
calculus through an encoding into the psi-calculus~\cite{Borgstrom2011}.
The work also provides with a sound session type system for the
aforementioned syntax.
In comparison to~\cite{Kouzapas2014},
in this work we define the syntax and semantics of a more sophisticated
calculus that includes distinction between processes and nodes,
asynchronous semantics through buffers, and more advanced recovery mechanisms.
The type system is also non-trivial, despite the fact that the session types
syntax remains the same. In contrast to~\cite{Kouzapas2014}, we also prove type
safety and progress results.

Asynchronous session semantics and their corresponding session typing systems are studied,
both for binary~\cite{DBLP:conf/ecoop/HuKPYH10,DBLP:journals/mscs/KouzapasYHH16}
and multiparty~\cite{DBLP:conf/popl/HondaYC08,DBLP:journals/mscs/CoppoDYP16}
session types.
However, our calculus is the first calculus that
supports asynchronous broadcast semantics, in contrast to
point-to-point communication
proposed by state-of-the-art session type calculi.

Session types for reliable gather semantics were proposed
in~\cite{DBLP:journals/corr/abs-1203-0780} but
a corresponding session calculus, was never proposed.
%
Structured multiparty session interaction for a parametrised (arbitrary)
number of participants is studied
in~\cite{DBLP:journals/corr/abs-1208-6483,DBLP:journals/soca/NgY15},
where the authors propose a framework for describing the behaviour
of an arbitrary number of participants that implement different
interacting roles. A parametrised multiparty session type protocol
can be instantiated to a specific number of participants.
Our framework can handle an arbitrary number of agents within a session
as a result of the shared unreliable communication medium. Agents
might arbitrarily lose messages and become non-synchronised and
arbitrarily recover and re-synchronise.

Session type systems in an unreliable context are studied in~\cite{10.1007/978-3-319-60225-7_1,CarboneHY08,CGY2014}. 
Adameit et al.~\cite{10.1007/978-3-319-60225-7_1} extend
multiparty session types with optional blocks that cover a limited class of
link failures. Specifically, the authors extend the typing
syntax with constructs that allow the description of a default value each role
needs to receive whenever there is a possibility of link of failure. 
In contrast, our framework does not require to describe failures at the type level,
rather than it enforces a number of conditions for a session safe recovery whenever
a failure occurs. Furthermore, it assumes that every interaction is subject
to failure and proposes several recovery mechanisms,
inspired by the practice of wireless sensor-networks. Our calculus, also,
supports broadcast semantics in contrast to point-to-point communication
supported by~\cite{10.1007/978-3-319-60225-7_1}.

Recovery after failure, was introduced in the shape of exception handling in
the binary~\cite{CarboneHY08} and multipary~\cite{CGY2014} session types.
Both works present a procedure where upon an exception during communication
the session endpoints are informed and safely handle the exception.
Similarly to~\cite{10.1007/978-3-319-60225-7_1}, the type discipline 
proposes a complicated syntax for exception description at the type level.
Moreover, exception handling within session types assumes strong global
synchronisation requirements among the session endpoints.
In contrast, our recovery semantics are implemented locally, thus being
more natural and general; each network node can autonomously and safely
recover from a communication failure, by choosing between
several recovery mechanisms.

The report in~\cite{ftmpst} presents fault-tolerant multiparty session types;
an extension to multiparty session types
that handles failures such as unreliable communication and process crashes,
applied to a calculus that supports failure patterns. Moreover, this work
demonstrates fault tolerant multiparty session types with an application
on the rotating coordination algorithm.
An additional work that describes recovery patterns is found in~\cite{magpi},
which develops an asynchronous multiparty framework that accommodates
non-Byzantine faults in unreliable settings. 
The work in~\cite{ftmpst} has similar recovery semantics as our work adapted
in the context of multiparty session types. Nevertheless,
we are the first to define patterns of broadcast and gather that are found
in more dynamic systems with share communication medium, such as
wireless sensor networks. Moreover, our work is the first to 
demonstrate an implementation of the Paxos protocol defined
as sessions of structured interaction.

\subsection{Overview}

In Section~\ref{sec:broadcast-calculus} we present
the Unreliable Broadcast Session Calculus --- \UBSC.
\UBSC is accompanied by the first type system for structured
communication in an unreliable broadcast setting.
The calculus develops all the necessary mechanisms
to satisfy the requirements S1-S5 under the assumptions
A1-A6.

Section~\ref{sec:session-types} presents the \UBSC session
type system.
Interestingly, we do not introduce the
description of unreliability at the type level, keeping the
session type syntax identical to standard binary session
type syntax.
The type system makes use of the synchronisation notion
to cope with non-synchronised session endpoints and
to ensure session duality.

In Section~\ref{sec:soundness}, the type system is proved
to be sound via a type preservation theorem and safe
via a type safety theorem. A process is safe whenever
it can never reduce to an error process. In turn, 
error processes are a class of processes that do not
respect the session types principles.
Section~\ref{sec:soundness} also includes a set of
progress results that ensure safe progress
within a session and safe session recovery.



%
%


We demonstrate the expressiveness of
our framework in \Section{paxos} with a session
specification of the Paxos consensus algorithm,
which is the standard consensus algorithm in distributed 
systems. 
We argue that our framework can provide support for the
implementation of such protocols and their extensions.

Finally, section~\ref{sec:conclusion} discusses the possibility for
future work and concludes the article.
	

\section{Asynchronous Unreliable Broadcast Session Calculus}
\label{sec:broadcast-calculus}

In this Section, we define the syntax and the semantics
for the Unreliable Broadcast Session Calculus.
The semantics are extensively demonstrated using several
examples.

\subsection{Syntax.}
\label{subsec:syntax}

Assume the following disjoint sets of names/variables:
$\mathcal C$ is a countable set of \emph{shared channels}
ranged over by $a,b, \dots$;
$\mathcal S$ is a countable set of \emph{session channels}
ranged over by $s,s',\dots$, where each session channel has two
distinct endpoints $s$ and $\aggr s$
(we write $\kappa$ to denote either $s$ or $\aggr s$);
$\mathcal V$ is a countable set of \emph{variables} ranged
over by $x,y,z,\dots$;
and $\mathsf{Lab}$ is a countable set of labels ranged over
by $\ell, \ell', \dots$.
We let $n$ range over
shared channels or sessions.
We write $k$ to denote either $\kappa$ or $x$
or $\aggr x$, where $\aggr x$ is used to
distinguish a variable used as a $\aggr s$-endpoint.

Let $c$ to range over natural numbers, $\mathbb{N}$ and let
\true and \false be the boolean values.
Let $\mathcal E$ be a non-empty set of {\em expressions}
ranged over by $e,e',\dots$. Elements of $\mathcal E$
contain natural numbers and boolean values, and may contain variables.
Function $\fv{e}$ returns the variables in expression $e$. 
Expressions that do not contain variables, i.e.~$\fv{e} = \emptyset$, are called closed.
%
Assume a binary operation
$\mult$ on $\mathcal E$ called {\em aggregation} operator,
and an element $\unit$ of $\mathcal E$ called {\em unit}.
We define an evaluation operator $\eval{}: \mathcal E \to \mathcal E$
from closed expressions to single value expressions (natural numbers, boolean values, \unit value, etc.). 
Let $\mathcal F \subseteq \mathcal E$ be a non-empty set
of {\em conditions} ranged over by \econd. Closed conditions
are evaluated to boolean values $\eval{\econd} = \true$ or $\eval{\econd} = \false$.
%
%
%
We use metavariable $u$ to denote either
shared names, session names, expressions or variables.

\begin{figure}
	\[
		\arraycolsep=3pt
		\begin{array}{rcll@{\hskip 25pt}rcll}
			a, b, \dots &\in& \mathcal{C}	& \SharedC &		s, s', \dots &\in& \mathcal{S} & \SessionC
			\\
			\ell, \ell', \dots &\in& \mathsf{Lab} & \Labels &		x, y, \dots &\in& \mathcal{V}	&	\Variables
			\\
			e, e', \dots &\in& \mathcal{E} & \Expressions &		c, c', \dots &\in& \mathbb{N}	&	\Counter
			\\[2mm]
			\kappa	&\bnfis& \aggr s \bnfbar s &	\Endpoint 	&k 	&\bnfis&	\kappa \bnfbar x & \Session
			\\
			v &\bnfis& a \bnfbar \kappa \bnfbar e	& \Identifier
			\\[2mm]
			P, Q, R	&\bnfis& \nil & \Inact 				&	&\bnfbar&	\sel{k}{\ell} P & \Selection
			\\	&\bnfbar& \request{a}{\aggr x} P & \Request	&	&\bnfbar&	\branchDef{k}{\ell_i: P_i}{R}_{i \in I} & \Branching
			\\	&\bnfbar& \accept{a}{x}{P} & \Accept		&	&\bnfbar&	\ndet{P}{P} & \NDet
			\\	&\bnfbar& \pout{k}{e} P & \Send			&	&\bnfbar&	\appl{D}{\tilde{v}} & \PVar
			\\	&\bnfbar& \pdin{k}{x}{e} P & \Rcv		&	&\bnfbar&	\cond{\econd}{P}{P} & \Conditional
			\\	&\bnfbar& \multicolumn{5}{l}{\Def{\set{\abs{D_i}{\tilde{x}_i} \defeq P_i}_{i \in I}}{P}} & \Recursion
			\\[4mm]
			N, M
				&\bnfis&	\net{P \pll B} &\Node 		&m	&\bnfis&	e \bnfbar \ell	& \Message
			\\	&\bnfbar& 	N \npll M &\Parallel 		&h	&\bnfis&	(c, e)	& \TMessage
			\\	&\bnfbar&	\newn{n} N& \Restriction &
			B	&\bnfis& 	\ebuffer \bnfbar B \pll \squeue{s}{c}{\pol{m}} \bnfbar B \pll \squeue{\aggr s}{c}{\pol{h}} & \Buffer
		\end{array}
	\]
	\caption{Syntax of Processes, Buffers, and Networks \label{fig:syntax}}
\end{figure}

The syntax of processes, $P, Q, R \in \Proc$, buffers, $B \in \Buf$, 
and networks, $N \in \Net$, is then defined in Figure~\ref{fig:syntax}.
Functions returning the set of free names, $\fn{P}$,
bound names, $\bn{P}$, and free and bound names, $\nam{P}$
are defined in the expected way.
Terms \Request and \Accept bind $\aggr x$ and $x$ in $P$ respectively,
and term \Rcv binds $x$ in $P$.
In term \Recursion, $\Def{\set{\abs{D_i}{\tilde{x}_i} \defeq P_i}_{i \in I}}{P}$,
$\abs{D_i}{\tilde{x}_i}$ binds $\tilde{x}_i$ in $P_i$.
Moreover, terms $\set{D_i}_{i \in I}$ 
are bound in $P$.
%
Term \Inact is the inactive term.
Terms \Request and \Accept express the processes that
are ready to initiate
a fresh session on a shared channel $a$ via a request/accept
interaction, respectively.
%
Term \Send defines a prefix ready to send an expression on session $k$.
Term \Rcv defines a prefix ready to receive a message on session
$k$ and substitute it on $x$. The prefix also provides with an
expression $e$, called default expression, that will be
substituted on $x$ in the case of recovery. 
We often write $\pin{k}{x} P$ for $\pdin{k}{x}{\unit} P$.
The select prefix, defined by term \Selection, is ready to send a
label $\ell$ over session $k$.
Dually, the branch prefix, defined by term \Branching,
is ready to receive a label from a predefined set of label
$\set{\ell_i}_{i \in I}$ on session $k$. Moreover,
the branch prefix defines a default label, $\default$, with
a process $R$ used for in the case of recovery.

In a highly dynamic unreliable environment, it is convenient
to consider non-deterministic choice --- term \NDet.
We write $\sum_{1 \leq i \leq n} P_i$ for
process $P_1 \Sum \dots \Sum P_n$.
Term \Conditional is a standard conditional term.
Finally, terms \Recursion and \PVar
express a named recursive process definition with parameters,
cf.~\cite{Honda1998}.
We assume a standard variable substitution over processes $P \subst{e}{x}$,
inductively defined to include a standard variable substitution over expressions, $e'\subst{e}{x}$.
Moreover, the process variable substitution $P \subst{\abs{P'}{\tilde{x}}}{D}$
is defined inductively with
$\appl{D}{\tilde{v}} \subst{\abs{P}{\tilde{x}}}{D} = P \subst{\tilde{v}}{\tilde{x}}$
as the basic definition case.
%
Concurrency is introduced at the network level,
rather than at process level.

Term $B$, is a parallel composition of session
buffers, that are used to store messages and keep track
of the session state via natural number $c$.
Buffer terms are used to model a form of asynchrony that
preserves the order of received messages, as required by S2.
The purpose of counter $c$ is to keep track of the
session state in the presence of communication failure,
and to synchronise the interaction between session prefixes,
as required by S4.
Message loss in an unreliable setting leads to
session endpoints that are not synchronised with the
overall protocol, as expected by S3.
To ensure correctness, many frameworks and algorithms
that operate in an unreliable setting use message tagging
or state counting; for example, the TCP/IP protocol tags
packets with unique sequential numbers to maintain consistency
in the case of packet loss.
In our setting, state counting is necessary
to maintain the correct semantics within a session.
The type system in Section~\ref{sec:session-types}
provides with static guarantees for a session despite 
the dynamic nature of session reduction.

Buffer terms on $s$-endpoints store {\em messages} $m$ that range
over expressions $e$ and labels $\ell$. 
Buffer terms on $\aggr s$-endpoints store messages,
$h$, which are expressions tagged with a session counter, $h = (c, e)$.
The session counter in $h$ distinguishes the session
state, at which the expression $e$ needs to be received.

Network \Node consists of a process $P$,
and the necessary buffer
terms, $B$, used for asynchronous session communication.
A process may participate in several sessions, and therefore,
more than one buffer term may be present in a node.
The type system ensures that there is no more than one
buffer term on the same session in each network node.
We write $\net{P}$ for node
$\net{P \pll \ebuffer}$.


A network is a parallel composition of nodes --- term \Parallel.
We write $\Par{i}{I}{N_i}$ for the parallel composition of
$N_1 \npll \cdots \npll N_n$
for (possibly empty) $I = \set{1, \dots, n}$. 
Network \Restriction binds both session and shared channels.
We write $\newn{\seq n} N$ for the network $\newn{n_1} \dots \newn{n_m} N$,
where the sequence $\seq n$ may be empty.
We also extend the $\fn{\cdot}$ function to networks.

\subsection{Operational Semantics.}
\label{subsec:os}
The operational semantics are defined as a reduction relation
on networks with the use of a standard structural congruence relation.

\begin{figure}
	\[
	\begin{array}{c}
		P_1 \Sum P_2 \equiv P_2 \Sum P_1
		\andalso \andalso
		(P_1 \Sum P_2) \Sum P_3 \equiv P_1 \Sum (P_2 \Sum P_3)
		\andalso \andalso
		P \equiv_\alpha P'
		\\
		\Def{\set{\abs{D_i}{\tilde{x}_i} \defeq P_i}_{i \in I}}{P} \equiv \Def{\set{\abs{D_i}{\tilde{x}_i} \defeq P_i}_{i \in I}}{P \subst{\abs{P_k}{\tilde{x}_k}}{D_k}} \quad k \in I
		\\[6mm]
		B \pll B' \equiv B' \pll B
		\andalso \andalso
		(B \pll B') \pll B'' \equiv B \pll (B' \pll B'')
		\\[6mm]
		N_1 \npll N_2 \equiv N_2 \npll N_1
		\andalso \andalso
		(N_1 \npll N_2) \npll N_3 \equiv N_1 \npll (N_2 \npll N_3)
		\andalso \andalso
		N \equiv \net{\nil} \npll N
		\\[2mm]
		\newn{n} \newn{m} N \equiv \newn{m} \newn{n} N
		\andalso \andalso
		\newn{n} N \npll M \equiv \newnp{n}{N \npll M} \quad \IF n \notin \fn{M}
		\\[2mm]
		\net{P \pll B} \equiv \net{P' \pll B'} \quad \IF P \equiv P' \AND B \equiv B'
		\andalso \andalso
		N \equiv_\alpha N'
	\end{array}
	\]
	\caption{Structural Congruence for Processes, Buffers, and Networks \label{fig:structural-congruence}}
\end{figure}

\subsubsection*{Structural Congruence}
The structural congruence on processes, resp.~buffers and networks,
is defined to be the least congruence relation satisfying the 
rules in Figure~\ref{fig:structural-congruence}.
Structural congruence on processes considers commutativity
and associativity of the \Sum operator, and includes the
unfolding of definitions and alpha-conversion.
Named definition substitution is defined up-to structural
congruence.
%
%
The parallel composition is commutative and associative for
buffer terms and for network terms, with $\net{\nil}$ as the
unit for network terms.
Name restriction order is irrelevant, and moreover,
the scope of restricted channels can be extruded.
The clause for network nodes simply bridges the buffer and process
congruences with the structural congruence for network. Finally,
structural congruence allows alpha renaming for networks.

The operational semantics is defined as the least relation on networks,
$N \reduces N'$, satisfying the rules given in  Figure~\ref{fig:reduction-networks} (Network Semantics),
Figure~\ref{fig:reduction-relation}
(Process Communication Semantics), and Figure~\ref{fig:reduction-recovery} (Recovery Semantics).


\begin{figure}
	\[
	\def\arraystretch{1.2}
	\begin{array}{c}
		\tree {
			\net{P_1 \pll B} \npll N
			\reduces
			\net{P' \pll B'} \npll N'
		}{
			\net{P_1 \Sum P_2 \pll B} \npll N
			\reduces
			\net{P' \pll B'} \npll N'
		}{\NonDet}
		\\[8mm]
		\tree[\RDef] {
			\net{P \pll B} \npll N
			\reduces
			\net{P' \pll B'} \npll N'
		}{
			\net{\Def{\set{\abs{D_i}{\tilde{x}_i} \defeq P_i}_{i \in I}}{P} \pll B} \,\,\npll\,\, N
			\reduces
			\net{\Def{\set{\abs{D_i}{\tilde{x}_i} \defeq P_i}_{i \in I}}{P'} \pll B'} \,\,\npll\,\, N'
		}{}
		\\[12mm]
		\tree[\RPar] {
			N \reduces N'
		}{
			N \npll M \reduces N' \npll M
		}{}
		\andalso
		\tree[\RCong] {
			N \equiv N_1
			\andalso
			N_1 \reduces N_2
			\andalso
			N_2 \equiv N'
		}{
			N \reduces N'
		}{}
		\andalso
		\tree[\RRes] {
			N \reduces N'
		}{
			\newn{n} N
			\reduces
			\newn{n} N'
		}{}
	\end{array}
	\]
        \caption{Reduction rules for networks.}
        \label{fig:reduction-networks}
\end{figure}

\subsubsection{Operational Semantics for Networks}
In Figure~\ref{fig:reduction-networks} rules
\NonDet, \RDef, \RPar, \RRes, and \RCong are standard
congruence rules for operator \Sum, named definition,
parallel composition, name restriction,
and structural congruence, respectively.


\begin{figure}
	\[
	\def\arraystretch{1.2}
	\begin{array}{cl}
		\tree {
			s \text{ fresh}
		}{
			\net{\request{a}{\aggr x} P \pll B} \,\,\npll\,\,
			\Par{i}{I}{\net{\accept{a}{x} P_i \pll B_i}}
			\\
			\qquad\qquad
			\reduces
			\newnp{s}{\net{P \subst{\aggr s}{\aggr x} \pll B \pll \squeue{\aggr s}{0}{\equeue} } \,\,\npll\,\,
			\Par{i}{I}{\net{P_i \subst{s}{x} \pll B_i \pll \squeue{s}{0}{\equeue} }}}
		}{}
		& \Connect
		\\[11mm]
		\begin{array}{l}
			\tree {
				\eval{e} = e'
			}{
			\net{\pout{\aggr s}{e} P \pll B \pll \squeue{\aggr s}{c}{\pol{m}} } \,\,\npll\,\,
			\Par{i}{I}{\net{P_i \pll B_i \pll \squeue{s}{c}{\pol{m}_i}}}
			\\
			\qquad\qquad 
			\reduces
			\net{P \pll B \pll \squeue{\aggr s}{c + 1}{\pol{m}}} \,\,\npll\,\,
			\Par{i}{I}{\net{P_i \pll B_i \pll \squeue{s}{c + 1}{\pol{m}_i \cat e'} } }
			}{}
		\end{array}
		& \Broadcast
		\\[11mm]
		\tree {
			c_1 \geq c_2
			\andalso
			\eval{e} = e'
		}{
			\net{\pout{s}{e} P_1 \pll B_1 \pll \squeue{s}{c_1}{\pol{m}}} \,\,\npll\,\,
			\net{P_2 \pll B_2 \pll \squeue{\aggr s}{c_2}{\pol{h}} }
			\\
			\qquad\qquad
			\reduces
			\net{P_1 \pll B_1 \pll \squeue{s}{c_1 + 1}{\pol{m}} } \,\,\npll\,\,
			\net{P_2 \pll \squeue{\aggr s}{c_2}{\pol{h} \cat (c_1, e')} }
		}{}
		& \Unicast
		\\[11mm]
		\net{\pdin{s}{x}{e'} P \pll B \pll \squeue{s}{c}{e \cat \pol{m}}}
		\reduces
		\net{P \subst{e}{x} \pll B \pll \squeue{s}{c}{\pol{m}}}
		& \Receive
		\\[5mm]
		\tree {
			\pol{h}' = \newbuffer{\pol{h}}{c}
			\andalso
			e = \gthr{\pol{h}}{c}
		}{
			\net{\pin{\aggr s}{x} P \pll B \pll \squeue{\aggr s}{c}{ \pol{h} }}
			\reduces
			\net{P \subst{e}{x} \pll B \pll \squeue{\aggr s}{c + 1}{ \pol{h}' }}
		}{}
		& \Gather
		\\[9mm]
		\begin{array}{l}
			\net{\sel{\aggr s}{\ell} P \pll B \pll \squeue{\aggr s}{c}{\pol{m}}} \,\,\npll\,\,
			\Par{i}{I}{\net{P_i \pll B_i \pll \squeue{s}{c}{\pol{m}_i}}}
			\\
			\qquad\qquad 
			\reduces
			\net{P \pll B \pll \squeue{\aggr s}{c + 1}{\pol{m}}} \,\,\npll\,\,
			\Par{i}{I}{\net{P_i \pll B_i \pll \squeue{s}{c + 1}{\pol{m}_i \cat \ell} }}
		\end{array}
		& \Select
		\\[7mm]
		\tree {
			k \in I
		}{
			\net{\branchDef{s}{\ell_i: P_i}{R}_{i \in I} \pll B \pll \squeue{s}{c}{\ell_k \cat \pol{m}}}
			\reduces
			\net{P_k \pll B \pll \squeue{s}{c}{\pol{m}}}
		}{}
		& \Branch
	\end{array}
	\]
        \caption{Reduction rules for Broadcast and Gather Operations.}
        \label{fig:reduction-relation}
\end{figure}

\subsubsection{Operational Semantics for Broadcast and Gather Operations}

Figure~\ref{fig:reduction-relation} defines the reduction semantics for
the broadcast and gather operations.
Rule \Connect establishes a session between a request network node and
several accept network nodes. It is a broadcast (one-to-many) communication
between the request network node and an arbitrary (possibly empty) set
of accepting network nodes described by $I$. Unreliability is
achieved because set $I$ is chosen arbitrarily from a set of parallel
nodes and then using the reduction rule \RPar to close the reduction of
the entire network.
The interaction creates a fresh session $s$ with a unique $\aggr s$-endpoint
in the request side and a shared $s$-endpoint on the accept sides.
A corresponding session buffer term is created in all connected network nodes.

There are two kind of session communication distinguished by the interaction
between the $\aggr s$-endpoint and $s$-endpoint: broadcast communication; and
unicast communication.
Rule \Broadcast defines the asynchronous broadcasting semantics;
the $\aggr s$-endpoint broadcasts evaluated message $e'$, with the
message being enqueued to the buffers of the $s$-endpoint,
thus modelling asynchrony.
The crucial condition for a broadcast interaction is that
all participating nodes are synchronised in the same session
state $c$. After the interaction all
participating nodes update their state to $c+1$.
Unreliability is modelled by the fact that the, possibly empty,
set $I$ is chosen arbitrarily from a set of parallel nodes,
and then using reduction rule \RPar to close the reduction
for the entire network.

Rule \Unicast defines the case where an $s$-endpoint enqueues an
evaluated message $e'$ in the buffer of the unique $\aggr s$-endpoint.
The message $e'$, when enqueued, is tagged with the sender's session
state $c$, as $h = (c, e')$, since all the messages on the same
state will be gathered by the $\aggr s$-endpoint (rule \Gather).
The session state counter of the $s$-endpoint will be updated,
in contrast to the state counter of the $\aggr s$-endpoint.
The $\aggr s$-endpoint remains in the same state since
it needs to continue the \Unicast interactions with other
network nodes.
Increasing the state counter does not disallow the $s$-endpoint
to be involved in a subsequent \Unicast interaction, i.e~the
$s$-endpoint is not considered non-synchronised.
This is captured by condition $c_1 \geq c_2$ that ensures that a
$s$-endpoint exists in the same or in a later session state from
the $\aggr s$-endpoint, prior to interacting.
The \Unicast semantics follow the practice of ad-hoc and sensor
networks, where a sender node cannot locally know whether the
receiving node has performed a gather prior to sending
a subsequent message.
For further intuition on the \Unicast rule, see Example~\ref{ex:reduction_heartbeat}.

Rule \Receive defines the interaction of a process with its $s$-endpoint
buffer; a \Rcv process on endpoint $s$ receives and substitutes on variable $x$
the next available expression, $e$, from the $s$-endpoint buffer. Session state $c$ is not updated,
because it was updated by the operation that stored expression $e$ in the buffer. 
Expression $e'$ is not used in this rule, since it is a default value used
for recovery in the case where the session endpoint becomes non-synchronised.
The case for recovery is described by recovery rule \Recover defined in Section~\ref{subsubsec:recovery}.

Rule \Gather defined the interaction between a process and its $\aggr s$-endpoint
buffer, where it gathers, via the \mult operator, all the expression messages that
are tagged with state, $c$, of the state of the $\aggr s$-endpoint. After
the reduction the state of the $\aggr s$-endpoint increases by $1$.
The rule uses auxiliary operations \gthr{\pol{h}}{c} and \newbuffer{\pol{h}}{c}:
\[
\arraycolsep=3pt
\begin{array}{rclcrclcrclr}
	\gthr{(c, e) \cat \pol{h}}{c} &=& e \mult \gthr{\pol{h}}{c}
	& &
	\gthr{\equeue}{c} &=& \unit
	&&
	\gthr{(c', e) \cat \pol{h}}{c} &=& \gthr{\pol{h}}{c} & \IF c' \not= c
	\\[1mm]
	\newbuffer{(c, e) \cat \pol{h}}{c} &=& \newbuffer{\pol{h}}{c}
	&\ &
	\newbuffer{\equeue}{c} &=& \equeue
	& &
	\newbuffer{(c', e) \cat \pol{h}}{c} &=& (c', e) \cat \newbuffer{\pol{h}}{c} & \IF c' \not= c
\end{array}
\]
Operation \gthr{\pol{h}}{c} goes through the messages $\pol{h}$ and returns,
up to operator \mult, all expressions $e$ tagged with session state $c$, $(c, e)$.
Operation \newbuffer{\pol{h}}{c} returns a new $\pol{h}'$ by removing
all messages tagged with session state $c$, $(c, e)$.
The \Gather semantics also describes the case where the $\aggr s$-endpoint
gathers no messages, capturing the case where all message where either lost
or not delivered yet; operator \gthr{\pol{h}}{\pol{s}} will
return unit, \unit, if there are no messages tagged with state $c$ in \pol{h}.
The gather pattern is common in ad-hoc networks; see, for example, the RIME
communication stack~\cite{Dunkels2007} for wireless sensor networks.

Rules \Select and \Branch are similar to rule \Broadcast and \Receive;
the $\aggr s$-endpoint selects and broadcasts a label to the corresponding
$s$-endpoint buffer terms, and dually the $s$-endpoints receive a label 
from its session buffer and proceeds accordingly.
The dual case where multiple $s$-endpoint select a label
is not defined, since gathering (i.e.~branching on) multiple
labels on the $\aggr s$-endpoint makes no sense in session
types semantics.
Rule \Branch also defines a default label, \default,
with a process $R$, used when there is a need to
recover. Recovery on process $R$ is described by recover rule \BRecover defined in Section~\ref{subsubsec:recovery}.

An implementation of the above semantics in real systems,
e.g.~wireless ad-hoc networks, requires that sent messages are
tagged with the the session state of the sender.
This allows a prospective receiver to check the session state
conditions of the interaction and act accordingly,
e.g.~a receiver will check its session state against
the message tag and only then will store a broadcast message.


\begin{figure}
	\[
	\def\arraystretch{1.2}
	\begin{array}{cl}
		\tree{
			\eval{e}=e'
		}{
			\net{\pdin{s}{x}{e} P \pll B \pll \squeue{s}{c}{\equeue}}
			\reduces
			\net{P \subst{e'}{x} \pll B \pll \squeue{s}{c+1}{\equeue}}
		}{}
		& \Recover
		\\[8mm]
		\tree {
			\fs{B} = \fs{R}
		}{
			\net{\branchDef{s}{\ell_i: P_i}{R}_{i \in I} \pll B \pll B' \pll \squeue{s}{c}{\equeue}}
			\reduces
			\net{R \pll B}
		}{}
		& \BRecover
		\\[8mm]
%
%
%
%
%
		\net{\pout{s}{e} P \pll B \pll \squeue{s}{c}{\pol{m}}}
		\reduces
		\net{P \pll B \pll \squeue{s}{c + 1}{\pol{m}}}
		& \Loss
		\\[5mm]
		\tree {
			\eval{\econd} = \true
			\andalso
			\fs{B} = \fs{P_1}
		}{
			\net{\cond{\econd}{P_1}{P_2} \pll B \pll B'}
			\reduces
			\net{P_1 \pll B}
		}{}
		& \True
		\\[7mm]
		\tree {
			\eval{\econd} = \false
			\andalso
			\fs{B} = \fs{P_2}
		}{
			\net{\cond{\econd}{P_1}{P_2} \pll B \pll B'}
			\reduces
			\net{P_2 \pll B}
		}{}
		& \False
		\\[7mm]
	\end{array}
	\]
        \caption{Reduction rules for recovery.}
        \label{fig:reduction-recovery}
\end{figure}

\subsubsection{Operational Semantics for Recovery}
\label{subsubsec:recovery}
The semantics of the calculus offer flexibility, and recovery
in a context of asynchrony and unreliability.
Figure~\ref{fig:reduction-recovery} defines the semantics
for these mechanisms. Crucially, the recovery semantics depend
on conditions that are checked locally within a network node, 
thus following the typical behaviour of ad-hoc and sensor
networks.
%

Rule \Recover describes the recovery conditions for a
$s$-endpoint input prefix.
Lack of messages in the corresponding $s$-buffer may trigger
recovery, where the evaluation of the default expression is substituted
in the continuation of the process. The interaction results
in an increase of the $s$-endpoint session state counter, aligning
the $s$-endpoint with the overall session interaction.

Rule \BRecover describes the conditions for recovering from a
branch prefix. Lack of messages in the corresponding $s$-buffer
may trigger a recovery, by continuing with the default label 
process $R$.
A branch recover is hard; the reduction drops the corresponding
$s$-endpoint and $s$-buffer. 
%
A process $R$ cannot implement a behaviour for the $s$-endpoint
because it cannot possibly know the corresponding choice from the
$\aggr s$-endpoint, in the case recovery results in a synchronised
$s$-endpoint.



Reduction rules \Recover and \BRecover describe
recovery from a situation where a $s$-endpoint might not be
able to progress.
For example, it may be the case that the $s$-endpoint is non-synchronised,
or the network node using the $\aggr s$-endpoint is deadlocked due to
the interleaving with other behaviour.
However, due to assumptions A2, A4 and A6, a network node cannot
possibly have the knowledge to decide whether a local session
endpoint can progress or not.
Assumption A4 implies that a node cannot know whether a send
message was received. Assumption A2 implies that a network node
cannot know whether an expected message was lost and, moreover,
due to assumption A6, if an expected message was not send yet.

Therefore, reduction rules \Recover and \BRecover allow nodes to act
autonomously, which means that network nodes can recover even
if their recovery is not necessary to achieve progress.
For example, the following reduction on rule \Recover
\[
	\net{\pout{\aggr s}{e} \nil \pll \squeue{\aggr s}{0}{\equeue}}
	\npll
	\net{\pin{ s}{x} \nil \pll \squeue{\aggr s}{0}{\equeue}}
	\reduces
	\net{\pout{\aggr s}{e} \nil \pll \squeue{\aggr s}{0}{\equeue}}
	\npll
	\net{\nil \pll \squeue{\aggr s}{1}{\equeue}}
\]
is not necessary to achieve progress, because the $\aggr s$-endpoints
and $s$-endpoint are in a state that can eventually interact.

In practice communicating nodes require necessary but not
sufficient conditions together with mechanisms, such as time-outs,
to approximate lack of progress.
This means that, in practice, there is always the possibility to
recover even if endpoint interaction is eventually possible.
We could define, for example, semantics that use global information
to recover but then we would not respect assumptions A4 and A6.
For example, rule \Recover might be defined as:
\[
	\tree {
		c < c'
	}{
		\net{\pdin{s}{x}{e} P \pll B \pll \squeue{s}{c}{\equeue}} \npll \net{P' \pll B' \pll \squeue{\aggr s}{c'}{\tilde{h}}}
		\reduces
		\net{P \subst{e}{x} B \pll \squeue{s}{c+1}{\equeue}} \npll \net{P' \pll B' \pll \squeue{\aggr s}{c'}{\tilde{h}}}
	}{}
\]
The rule uses global information between network nodes.
Condition $c < c'$ checks that the $s$-endpoint exists
in an earlier state than the $\aggr s$-endpoint. The
recovery takes place only after ensuring that the
$s$-endpoint is non-synchronised.
%
%
Rule \BRecover can also be accommodated in a similar fashion.
%

Rule \Loss allows for message loss on the $s$-endpoint;
a process simply proceeds by dropping its sending prefix.
Similarly with rule \Recover, the session state is updated/increased,
allowing the $s$-endpoint to re-synchronise with the $\aggr s$-endpoint.
%
Rules \Loss and \Unicast realistically model unreliable unicast
communication since messages can always be lost and, moreover, due
to assumption A4
a node cannot possibly know if a message was passed to the
receiver and should always update its state.
The application of rule \Loss implies two cases:
i) the $s$-endpoint exists in a later state than the state
of the $\aggr s$-endpoint, in which case the $s$-endpoint
is considered synchronised, similarly to rule \Unicast;
ii) the $s$-endpoint prior to \Loss exists in an earlier
state than the $\aggr s$-endpoint, in which case the application
of the \Loss rule is the only interaction that can be done
towards recovery and synchronisation of the $s$-endpoint. 

Rules \True and \False offer the programmer additional flexibility
when handling shared session endpoints. The conditional process gives
the ability to discontinue using some sessions in its branches. Thus, a
session can be dropped based on the truth evaluation of condition $\econd$.
When the process proceeds to a branch it must drop the buffers not used
by the continuation (condition $\fs{B'} = \fn{P_i} \setminus \fn{P_j}, i \not= j$).
The case where no sessions are dropped, corresponds to standard conditional
semantics.
Explicitly dropping a communication channel is a common pattern in the context of
unreliable communication.

\subsection{Encoding Node Failure and Recovery Process}
The \UBSC semantics are powerful enough to express
sophisticated patterns of network node failure and
recovery.
We provide the semantics for a recovery
pattern, typically found in ad-hoc and sensor
network, through an encoding in the \UBSC.
The encoding is useful to avoid long and tedious
descriptions of recovery interaction.

The recovery pattern takes advantage of the 
interplay between rules \Recover, \True, and \False,
to describe a mechanism that recovers by proceeding to a
recovery process.
We defined the \UBSC with recovery terms
by extending the terms of \UBSC in Figure~\ref{fig:syntax} as:
\[
	P \bnfis \dots \bnfbar P \recover R \ \SRecover
\]
Additionally,
in the case of a branch prefixed process we write 
$\branch{k}{l_i: P_i}_{i \in I} \recover R$, instead of
$\branchDef{k}{l_i: P_i}{R'}_{i \in I} \recover R$.

The semantics of the processes of the form $P \recover R$
are then defined through a syntactic encoding into
the terms of \UBSC (syntax of Figure~\ref{fig:syntax}):
%
\[
\begin{array}{rcl}
	\map{\nil \recover R} &\defeq& \nil
	\\
	\map{\appl{D}{\tilde{v}} \recover R} &\defeq& \appl{D}{\tilde{v}}
	\\
	\map{\pin{k}{x} P \recover R} &\defeq& \pdin{k}{\exception}{x} \cond{x \not= \exception}{\map{P \recover R}}{R}
	\\
	\map{\branch{k} {l_i: P_i}_{i \in I} \recover R} &\defeq& \branchDef{k}{l_i: (\map{P_i \recover R})}{R}_{i \in I}
	\\
	\map{(\Def{\set{\abs{D_i}{\tilde{x}_i} \defeq P_i}_{i \in I}}{P}) \recover R} &\defeq&
	\Def{\set{\abs{D_i}{\tilde{x}_i} \defeq \map{P_i\recover R}}_{i \in I}}{\map{P\recover R}} 
\end{array}
\]
and homomorphic for the rest of the syntax of the \UBSC with recovery processes.
The encoding assumes the existence of an {\em exception} value, \exception,
which 
is used internally for the purpose
of detecting recovery.
Encoding for terms \Inact and \PVar
is the identity.
The recovery behaviour of the encoding 
for terms \Rcv and \Branching is based on
reduction rules
\Recover, \BRecover, \True, and \False;
whenever a process exists in a $s$-endpoint input prefix (prefixes
receive and branch) and the recovery conditions for rules
\Recover and \BRecover hold, the network node can recover
by continuing to a recovery process $R$. Note that sessions can be
dropped following the semantics of rules \BRecover, \True, and \False.
Finally, term \Recursion is defined inductively both on the named
definitions and on the process body.

\subsection{Reduction Semantics Examples}
The next two examples demonstrate some rules and
basic intuition of the operational semantics.
The first example demonstrates the semantics for
rule \Connect, 
and the  interplay between the rules \Unicast, \Loss, and
\Gather.

\begin{exa}[A Heartbeat Protocol]
	\label{ex:reduction_heartbeat}
	Consider a variant of the Heartbeat protocol,
	introduced in Example~\ref{ex:intro_heartbeat},
	where a node periodically gathers
	heartbeat messages from nodes within the network.
	\[
	\begin{array}{rcl}
		\Heartbeat = &&	\net{\request{a}{\aggr{y}} \pin{\aggr y}{x_1} \pin{\aggr y}{x_2} P_0}
			\\	&\npll&
				\net{\accept{a}{y} \pout{y}{\heartbeat_1} \pout{y}{\heartbeat_1} P_1}
				\,\,\npll\,\,
				\net{\accept{a}{y} \pout{y}{\heartbeat_2} \pout{y}{\heartbeat_2} P_2}
	\end{array}
	\]
	Node $\net{\request{a}{\aggr{y}} \pin{\aggr y}{x} \pin{\aggr y}{x} P_0}$
	requests a new session on shared name $a$.
	After the establishment of the new session each accepting
	node will periodically send a heartbeat message to the requestor node.
	Consider now the interaction:
	\[
	\begin{array}{rccl}
		\Heartbeat &\reduces &&
				\newnp{s}{
					\net{\pin{\aggr s}{x_1} \pin{\aggr s}{x_2} P_0 \pll \squeue{\aggr s}{0}{\equeue}}
				\\	&&\npll&
					\net{\pout{s}{\heartbeat_1} \pout{s}{\heartbeat_1} P_1 \pll \squeue{s}{0}{\equeue}}
					\,\,\npll\,\,
					\net{\pout{s}{\heartbeat_2} \pout{s}{\heartbeat_2} P_2 \pll \squeue{s}{0}{\equeue}}
				}
		\\[2mm]
		&\reduces &&
				\newnp{s}{
					\net{\pin{\aggr s}{x_1} \pin{\aggr s}{x_2} P_0 \pll \squeue{\aggr s}{0}{(0, \heartbeat_2)}}
				\\	&&\npll&
					\net{\pout{s}{\heartbeat_1} \pout{s}{\heartbeat_1} P_1 \pll \squeue{s}{0}{\equeue}}
					\,\,\npll\,\,
					\net{\pout{s}{\heartbeat_2} P_2 \pll \squeue{s}{1}{\equeue}}
				}
		\\[2mm]
		&\reduces &&
				\newnp{s}{
					\net{\pin{\aggr s}{x_1} \pin{\aggr s}{x_2} P_0 \pll \squeue{\aggr s}{0}{(0, \heartbeat_2), (1, \heartbeat_2)}}
				\\	&&\npll&
					\net{\pout{s}{\heartbeat_1} \pout{s}{\heartbeat_1} P_1 \pll \squeue{s}{0}{\equeue}}
					\,\,\npll\,\,
					\net{P_2 \pll \squeue{s}{2}{\equeue}}
				}
		\\[2mm]
		&\reduces &&
				\newnp{s}{
					\net{\pin{\aggr s}{x_1} \pin{\aggr s}{x_2} P_0 \pll \squeue{\aggr s}{0}{(0, \heartbeat_2), (1, \heartbeat_2), (0, \heartbeat_1)}}
				\\	&&\npll&
					\net{\pout{s}{\heartbeat_1} P_1 \pll \squeue{s}{1}{\equeue}}
					\,\,\npll\,\,
					\net{P_2 \pll \squeue{s}{2}{\equeue}}
				}
		\\[2mm] &=&& \Heartbeat_1
	\end{array}
	\]
	The first interaction is an instance of rule \Connect 
	and has lead to the establishment of new session
	involving all the network nodes.
	The second interaction uses an instance of the rule \Unicast
	where network node
	$\net{\pout{s}{\heartbeat_2} \pout{s}{\heartbeat_2} P_2 \pll \squeue{s}{0}{\equeue}}$
	sends a heartbeat message to the $\aggr s$-endpoint network node.
	Observe that the session state of the sender node has increased
	by one.
	The third interaction is also an instance of rule \Unicast
	and yet another heartbeat message from the same network node
	is unicast to the $\aggr s$-endpoint network node. The node can unicast
	the heartbeat message even if it is in a later session state.
	This is because messages are gathered at each session state.
	In the last interaction, network node
	$\net{\pout{s}{\heartbeat_1} \pout{s}{\heartbeat_1} P_1 \pll \squeue{s}{0}{\equeue}}$
	also unicasts a heartbeat message.
	Consider then that the $\aggr s$-endpoint gathers the heartbeat messages:

	\[
	\begin{array}{rccl}
		\Heartbeat_1 
			&\reduces &&
					\newnp{s}{
						\net{\pin{\aggr s}{x_2} P_0 \subst{\heartbeat_1 \mult \heartbeat_2}{x_1} \pll \squeue{\aggr s}{1}{(1, \heartbeat_2)}}
					\\	&&\npll&
						\net{\pout{s}{\heartbeat_1} P_1 \pll \squeue{s}{1}{\equeue}}
						\,\,\npll\,\,
						\net{P_2 \pll \squeue{s}{2}{\equeue}}
					}
		\\[2mm]
			&\reduces &&
					\newnp{s}{
						\net{P_0 \subst{\heartbeat_1 \mult \heartbeat_2}{x_1} \subst{\heartbeat_2}{x_2} \pll \squeue{\aggr s}{2}{\equeue}}
					\\	&&\npll&
						\net{\pout{s}{\heartbeat_1} P_1 \pll \squeue{s}{1}{\equeue}}
						\,\,\npll\,\,
						\net{P_2 \pll \squeue{s}{2}{\equeue}}
					}
		\\[2mm]	&=&& \Heartbeat_2
	\end{array}
	\]
	Two instances of the rule \Gather allow for the
	$\aggr s$-endpoint to consume all the messages
	from the $\aggr s$-buffer. The first interaction
	gathers all heartbeat messages tagged with state $0$, 
	whereas the second interaction gathers all the
	heartbeat messages tagged with state $1$.
	Each interaction updates the state counter
	of the $\aggr s$-endpoint.
	After the last two interactions the $s$-endpoint
	$\net{\pout{s}{\heartbeat_1} P_1 \pll \squeue{s}{1}{\equeue}}$
	is found in a non-synchronised state.
	However an instance of rule \Loss can be observed
	that will re-synchronise the non-synchronised $s$-endpoint:
	\[
		\begin{array}{rccl}
			\Heartbeat_2
				&\reduces &&
						\newnp{s}{
							\net{P_0 \subst{\heartbeat_1 \mult \heartbeat_2}{x_1} \subst{\heartbeat_2}{x_2} \pll \squeue{\aggr s}{2}{\equeue}}
						\\	&&\npll&
							\net{P_1 \pll \squeue{s}{2}{\equeue}}
							\,\,\npll\,\,
							\net{P_2 \pll \squeue{s}{2}{\equeue}}
						}
		\end{array}
	\]
\qed
\end{exa}

A second example demonstrates a typical pattern of interaction
found ad-hoc and sensor networks when running consensus protocols.
In this example, rules \True/\False are used to drop connections.
\begin{exa}[Dropping Connections]
	It is typical in consensus algorithms
	for a node to establish multiple connections
	with other nodes but maintaining active
	only the one with the highest id number,
	e.g.~the Paxos consensus algorithm~\cite{lamport1998paxos,lamport2001paxos}
	(also see Section~\ref{sec:paxos}).
	For example, consider network:
	\[
	\begin{array}{rcll}
		\Network &=&& \net{\request{a}{w} \pout{\aggr w}{id_2} P_0}
		\,\,\npll\,\,
		\newnp{s}{\net{\pout{\aggr s}{id_1} P_1 \pll \squeue{\aggr s}{0}{\equeue}}
	\\	&&\npll&
		\net{\pin{s}{x} (P \Sum \accept{a}{w} \pin{w}{y} \cond{x > y}{P}{P'}) \pll \squeue{s}{0}{\equeue}}}
	\end{array}
	\]
	\noindent
	with $s \notin \fn{P'}$. 
	Node
	$\net{\pin{s}{x} (P \Sum \accept{a}{w} \pin{w}{y} \cond{x > y}{P}{P'}) \pll \squeue{s}{0}{\equeue}}$
	has already established a connection on channel $s$ and awaits for
	a connection id number; a \Broadcast and a \Receive
	interaction will result in:
	\[
	\begin{array}{rcll}
		\Network &\reduces \reduces&& \net{\request{a}{w} \pout{\aggr w}{id_2} P_0}
		\,\,\npll\,\,
		\newnp{s}{\net{P_1 \pll \squeue{\aggr s}{1}{\equeue}}
	\\	&&\npll&
		\net{P \Sum \accept{a}{w} \pin{w}{y} \cond{id_1 > y}{P}{P'} \pll \squeue{s}{1}{\equeue}}}
		\\[2mm]
		&&=& \Network_1			 
	\end{array}
	\]
	The resulting network has the option to either
	continue interaction on the $s$-endpoint via
	process $P$, or establish a new connection
	on shared channel $a$. 
	Assume that the latter
	interaction takes place:
	\[
	\begin{array}{rcll}
		\Network_1 &\reduces&& \newnp{s, s'}{\net{\pout{\aggr s'}{id_2} P_0 \pll \squeue{\aggr s}{0}{\equeue}}
		\,\,\npll\,\,
		\net{P_1 \pll \squeue{\aggr s}{1}{\equeue}}
	\\	&&\npll&
		\net{\pin{s'}{y} \cond{id_1 > y}{P}{P'} \pll \squeue{s}{1}{\equeue} \pll \squeue{s'}{0}{\equeue}}}
		\\[2mm]
		&&=& \Network_2
	\end{array}
	\]
	The $s'$-endpoint also receives a connection id number, via reduction rules
	\Broadcast and \Receive: 
	\[
	\begin{array}{rcll}
		\Network_2 &\reduces \reduces&& \newnp{s, s'}{\net{P_0 \pll \squeue{\aggr s}{1}{\equeue}}
		\,\,\npll\,\,
		\net{P_1 \pll \squeue{\aggr s}{1}{\equeue}}
	\\	&&\npll&
		\net{\cond{id_1 > id_2}{P}{P'} \pll \squeue{s}{1}{\equeue} \pll \squeue{s'}{1}{\equeue}}}
		\\[2mm]
		&&=& \Network_3
	\end{array}
	\]
	The receiving node will then compare the two connection
	id numbers and decide to drop the session with the
	smallest corresponding id number and continue with the
	corresponding process; it will continue the interaction
	on the $s$-endpoint in process $P$ 
	in the case where the newest connection
	is dropped, or, otherwise, it will proceed with the interaction
	on the $s'$-endpoint via process $P'$.
	For example, if $id_2 > id_1$ the network will reduce
	using an instance of rule \False and result as in:
	\begin{align*}
		\Network_3 &\reduces&& \newnp{s}{\net{P_1 \pll \squeue{\aggr s}{1}{\equeue}}}
		\,\,\npll\,\, \newnp{s'}{\net{P_0 \pll \squeue{\aggr s}{1}{\equeue}}
		\,\,\npll\,\,
		\net{P' \pll \squeue{s'}{1}{\equeue}}} \tag*{\qed}
	\end{align*}
\end{exa}

The next example demonstrates a simple recursive behaviour,
where a $\aggr s$-endpoint recurses or terminates an session interaction
based on the acknowledgements it receives by the corresponding $s$-endpoints.
\begin{exa}[Recursive Interaction]
\label{ex:recursive_interaction}
	The following presents an example where
	the $\aggr s$-endpoint sends a message
	and then requires from
	the corresponding $s$-endpoints to 
	reply with an acknowledgement.
	By inspecting the acknowledgements,
	the $\aggr s$-endpoint decides to terminate
	the interaction or to reiterate.
	Consider the network:
	
	\[
	\begin{array}{rcll}
		P			&&=&		\Def{
					\\		&&& \quad
								\begin{array}{rcl}
									\abs{Q}{w}
									&\defeq&
										\pout{\aggr w}{v}
										\pin{\aggr w}{\set{x_i}_{i \in I}}
									\\
									&&	\If\ \mathsf{cond}(\set{x_i}_{i \in I})\ \Then \ 
										\sel{\aggr w}{\acceptLabel} \nil
									\\
									&&	\Else \ 
										\sel{\aggr w}{\restart} \appl{P}{w}
									\\[2mm]
								\end{array}
						\\		&&& }{ \quad \appl{Q}{s} }
						
		\\[2mm]
		P_k			&&=&		\Def{
								\abs{R_k}{w}
								\defeq
									\pin{w}{x} \pout{w}{\mathsf{ack_k}} \branch{ w } { \acceptLabel: \nil, \restart: \appl{Q}{w}, \default: \nil }
					\\		&&& }{ \appl{R_k}{s} }
		\\[2mm]
		
		\Recursive &&=& \net{P \pll \squeue{\aggr s}{0}{\equeue}}
		\,\,\npll\,\,
		\prod_{J \in J} \net{P_j \pll \squeue{s}{0}{\equeue}}
	\end{array}
	\]
	where operation \mult aggregates all
	acknowledgements as a set $\set{\mathsf{ack}_i}_{i \in I}$.
	\noindent
	A \Broadcast followed by a series of \Receive and \Unicast interactions,
	and finally a \Gather operation, results in
	(the recursive process definition is ommitted):
	\[
	\begin{array}{rcl}
		\Recursive &\reduces^*& \Recursive_1
		\\ &=&
		\net{ \If\ \mathsf{cond}(\set{x_i}_{i \in I})\ \Then \ \sel{\aggr s}{\acceptLabel} \nil
										\ \Else \ \sel{\aggr s}{\restart} \appl{Q}{s} \pll \squeue{\aggr s}{2}{\equeue} }
		\\
		&\npll&
		\prod_{k \in K_1} \net{ \branch{ s } { \acceptLabel: \nil, \restart: \appl{R_k}{s}, \default: \nil } \pll \squeue{s}{2}{\equeue}}
		\\ &\npll&
		\prod_{k \in K_2} \net{ P_k \pll \squeue{s}{0}{\equeue}}
	\end{array}
	\]
	with $K_1 \cup K_2 = J$ and $K_1 \cap K_2 = \emptyset$.
	If the condition on $\set{\mathsf{ack}}_{i \in I}$ is true
	then after a \Select on label \acceptLabel
	followed by a series of \Branch interactions the network may result in:
	\[
		\begin{array}{rcl}
			\Recursive_1 &\reduces^*&
			\net{\nil \pll \squeue{\aggr s}{3}{\equeue}} \,\,\npll\,\, \prod_{k \in K_1} \net{  \nil \pll \squeue{s}{3}{\equeue}} \,\,\npll\,\,
		\prod_{k \in K_2} \net{ Q_k \pll \squeue{s}{0}{\equeue}}
		\end{array}
	\]
	Note that each node in network $\prod_{k \in K_2} \net{ Q \pll \squeue{s}{0}{\equeue}}$
	can terminate following a sequence of \Recover, \Loss and \BRecover interactions.
	If the condition on $\set{\mathsf{ack}}_{i \in I}$ is false
	then after a \Select on label \restart
	followed by a series of \Branch interactions the network may result in:
	\[
		\begin{array}{rcl}
			\Recursive_1 &\reduces^*&
			\net{P \pll \squeue{\aggr s}{3}{\equeue}} \,\,\npll\,\, \prod_{k \in K_1} \net{  Q \pll \squeue{s}{3}{\equeue}} \,\,\npll\,\,
		\prod_{k \in K_2} \net{ Q_k \pll \squeue{s}{0}{\equeue}}
		\end{array}
	\]
	where the network reiterates. Note that
	subnetwork $\prod_{k \in K_2} \net{ Q \pll \squeue{s}{0}{\equeue}}$
	will not participate in the next iteration, since
	it can only terminate its interaction by recovering following a sequence of \Recover, \Loss and \BRecover interactions.
	\qed
\end{exa}


\section{Session Types}
\label{sec:session-types}

We now introduce the session type system for the \UBSC.
The type system
combines ideas from~\cite{DBLP:conf/ecoop/HuKPYH10,dkphdthesis}
to type buffer terms.
A novel notion is the notion of
endpoint synchronisation used to cope with 
the presence of non-synchronised session endpoints 
and with recovery semantics.

\subsection{Type Syntax}

%
%

The session types syntax follows the standard binary session types syntax
as introduced by Honda et al.~\cite{Honda1998}.  
However, we have no session delegation since channel delegation is
not found in systems with unreliable communication such as sensor
networks.
\begin{defi}[Session Type]
	Let $\mathcal B$ be a set of \emph{base types} ranged over by~$\beta$.
	Session types are inductively defined by the following grammar:
	\[
		T	\DEF	\tout \beta T
			\OR		\tin \beta T
			\OR		\tselI{\ell}{T}
			\OR		\tbranchI{\ell}{T}
			\OR		\tend
			\OR		\tvar t
			\OR		\trec t T
	\]
\end{defi}
Type $\tout \beta T$ describes the sending of a value of type $\beta$
and then proceeding with type $T$. Type $\tin \beta T$ describes the reception
of a value with type $\beta$ and then proceeding with type $T$. Type
$\tselI{\ell}{T}$ selects a label
from the set of labels $\set{\ell_i}_{i \in I}$ and then proceeds 
with the corresponding type $\set{T_i}_{i \in I}$. Dually, type
$\tbranchI{\ell}{T}$ branches on the set of labels $\set{\ell_i}_{i \in I}$
and then proceeds with the corresponding type $\set{T_i}_{i \in I}$.
Type $\tend$ is the inactive type, whereas $\tvar t$ is the recursive
variable. Finally, type $\trec t T$ is the recursive type, which
binds free occurrences of $\tvar t$ in $T$.
We define a capture avoiding substitution on types $T \subst{T'}{t}$ in
the usual way. We assume equi-recursive types, $\trec t T = T \subst{\trec t T}{t}$.

Our calculus does not incorporate session delegation. Session delegation is rather
unnatural in an unreliable setting supporting broadcasting and gather semantics,
and sharing of channel resources.
Moreover, session delegation is valid only for $\aggr s$-endpoints, which are linear,
and it would require additional syntax and reduction semantics beyond broadcast and
gather operations.

The duality operator also follows the standard binary session type
definition~\cite{Honda1998}.
A simple inductive definition is enough to capture session endpoint duality
because of the lack of delegation~\cite{BDGK14}.
\begin{defi}[Type Duality]
\[
	\begin{array}{rcl c rcl c rcl c rcl c rcl}
		\tdual\tend	&=&	\tend
		&&
		\tdual{\tout\beta T}	&=&	\tin\beta \tdual T
		&&
		\tdual{\tin\beta T}	&=&	\tout\beta \tdual T
		&&
		\tdual{\tvar t}			&=&	\tvar t 
		&&
		\tdual{\trec t T}		&=&	\trec t \tdual T
		\\[2mm]
		\multicolumn{9}{c}{\tdual{\tselI{\ell}{T}} =\tbranchI{\ell}{\tdual{T}}}
		&&
		\multicolumn{9}{c}{\tdual{\tbranchI{\ell}{T}} = \tselI{\ell}{\tdual{T}}}
	\end{array}
\]
\end{defi}
Two types $T_1$ and $T_2$ are dual if $\tdual T_1 = T_2$.
Note $\tdual{\tdual{T}} = T$ for any $T$.
Next, we
define a buffer type syntax used
to type message buffers (cf.~\cite{DBLP:conf/ecoop/HuKPYH10,dkphdthesis}).
\begin{defi}[Buffer Types]
	Buffer types are inductively defined by the following grammar:
	\[
		M \bnfis \tempty \bnfbar \mtout \beta. M \bnfbar \msel \ell. M
	\]
\end{defi}
Buffer types are used to type buffer terms; they describe the types
of the values, $\mtout \beta. M$, or labels, $\msel \ell. M$, in a session buffer term.

The \sessionop operator
is used to combine session types and buffer types
(cf.~\cite{DBLP:conf/ecoop/HuKPYH10,dkphdthesis}).
\begin{defi}[Operator \sessionop]
	\label{def:combine_context}
	\[
		\begin{array}{l} 
			
			\tin \beta T \sessionop \mtout \beta. M = T \sessionop M
			\qquad \qquad
			\tbranchOn{l_i: T_i}_{i \in I} \sessionop \msel \ell_k. M = T_k \sessionop M
			\\[1mm]
			\trec t T \sessionop M = T \subst{\trec t T}{\tvar t} \sessionop M
			\qquad
			T \sessionop \tempty = T
		\end{array}
	\]
\end{defi}

The \sessionop operator combines a session type with a buffer type 
and returns a session type (cf.~\cite{DBLP:conf/ecoop/HuKPYH10,dkphdthesis}).
It works inductively by removing
the prefix from the input session when the prefix of the buffer
type is dual.
%
The intuition for the \sessionop operator considers that messages
within a session buffer have already been received from a node
(see typing rule \TNode in Figure~\ref{fig:network_typing}), so
the operator consumes a buffer type against a session type.

%
%

\subsection{Typing System}

We define the typing contexts used by the typing system.
\begin{defi}[Typing Context]
	We define $\Gamma$, $\Delta$, and $\Theta$ typing contexts:
	\[
		\begin{array}{rclcrcl}
			\Gamma &\bnfis& \econtext \bnfbar \Gamma, x: \beta \bnfbar \Gamma, a: T \bnfbar \Gamma, \abs{D}{\tilde{x}}: (\Gamma; \Delta)
		\\[1mm]
			\Delta &\bnfis& \econtext \bnfbar \Delta, k : T \bnfbar \kappa: \ctype{c}{T}
		\\[1mm]
			\Theta &\bnfis& \econtext \bnfbar \Theta, \kappa: \ctype{c}{M}
		\end{array}
	\]
\end{defi}
Context $\Gamma$ is called shared context and maps variables to ground types,
shared names to session types,
and recursive variables to typing contexts $(\Gamma;\Delta)$.
Context $\Delta$ is called linear context and maps session names and session variables
to session types, and session names to tuples, $\ctype{c}{T}$, that
combine an integer value, $c$, that corresponds to session state together with a session type, $T$.
Similarly, context $\Theta$ is called buffer context and maps session names
to tuples, $\ctype{c}{M}$, that
combine an integer value, $c$, that corresponds to session state together with a buffer type, $M$.

Contexts are treated as sets. We write $\Gamma, \Gamma'$ to denote the union of contexts $\Gamma$ and
$\Gamma'$. We also write $\Delta,\Delta'$, resp.~$\Theta, \Theta'$,
for the disjoint union of contexts $\Delta$ and $\Delta'$, resp.~$\Theta$ and $\Theta'$.
We define the domain of $\dom{\cdot}$ 
of contexts $\Gamma$, $\Delta$, and $\Theta$ in the expected way.

The \sessionop operator is lifted to combine
contexts $\Delta$ and $\Theta$:
\[
	\Delta \sessionop \Theta = \set{\kappa: \ctype{c}{T \sessionop M} \setbar \kappa : T \in \Delta \land \kappa : \ctype{c}{M} \in \Theta}
\]
The result of the \sessionop operator on contexts is a new linear context.

\begin{figure}
	\[
	\begin{array}{c}
		\Gamma; \econtext \types \ebuffer \quad \BEmp
		\andalso \andalso
%
			\Gamma; \Theta, \kappa: \ctype{c}{\tempty} \types \squeue{\kappa}{c}{\equeue} \quad \SEmp
%
		\andalso \andalso
		\tree {
			\set{\Gamma; \Theta_i \types B_i}_{i \in \set{1, 2}}
		}{
			\Gamma; \Theta_1, \Theta_2 \types B_1 \pll B_2
		}{\BPar}
		\\[8mm]
		\tree {
			\Gamma; \Theta, \ctype{c}{M} \types \squeue{s}{c}{\pol{m}}
			\andalso
			\Gamma \types e : \beta
		}{
			\Gamma; \Theta, s: \ctype{c}{M. \mtout \beta} \types \squeue{s}{c}{\pol{m} \cat e}
		}{\SExp}
		\andalso \andalso
		\tree {
			\Gamma; \Theta, \ctype{c}{M} \types \squeue{s}{c}{\pol{m}}
		}{
			\Gamma; \Theta, s: \ctype{c}{M. \msel \ell} \types \squeue{s}{c}{\pol{m} \cat \ell}
		}{\SLab}
		\\[8mm]
		\tree {
			\pol{h}' = \newbuffer{\pol{h}}{c}
			\andalso
			\Gamma \types \gthr{\pol{h}}{c} : \beta
			\andalso
			\Gamma; \Theta, \aggr s: \ctype{c}{M} \types \squeue{\aggr s}{c'}{\pol{h}'}
		}{
			\Gamma; \Theta, \aggr s: \ctype{c+1}{M. \mtout \beta} \types \squeue{\aggr s}{c'}{ \pol{h} }
		}{\LExp}
	\end{array}
	\]
	\caption{Typing rules for session buffers \label{fig:buffer_typing}}
\end{figure}

The typing judgement for expressions is
defined as the least relation $\Gamma \types e: \beta$
that respects: i) the types of value and unit expressions (e.g., $\nat, \bool \in B$);
and
ii) the types of variables in expressions,
e.g. whenever $x \in \fv{e}$, then $\Gamma \types e$
is defined if $x: \beta \in \Gamma$ for some $\beta$. 

The typing judgement for session buffers, $\Gamma; \Theta \types B$,
is defined as the  least relation that satisfies the rules
in \Figure{buffer_typing}.
Rule \BEmp types the empty buffer term with an empty $\Theta$ context.
Rule \SEmp types an empty $\kappa$-endpoint buffer. It maps, 
in context $\Theta$,
the $\kappa$-endpoint to the state of the $\kappa$-buffer together
with the empty buffer type.
Rule \BPar types a parallel composition of session buffer terms, 
with the disjoint union of the two respective $\Theta$ environments.
Buffers for $s$-endpoints are typed with rules \SExp and \SLab;
the type of an expression (respectively, label) is appended
at the end of the buffer type of $s$ within the $\Theta$ context.

Rule \LExp types the $\aggr s$-endpoints buffers.
The rule works by reconstructing the messages within the
$\aggr s$-buffer using operation $\pol{h}' = \newbuffer{\pol{h}}{c}$,
with $h'$ appearing in the premise, and $h$ appearing in
the conclusion.
Rule \LExp appends to the type of $\aggr s$ within context $\Theta$,
the type of value $\gthr{\pol{h}}{c}$, which
is \mult-composition of all expressions associated with
session state $c$, that appears at the type of the premise. The
session state is updated at the conclusion.

\begin{figure}
	\[
	\begin{array}{c}
		\Gamma; \econtext \types \nil\ \TInact
		\andalso \andalso
		\tree {
			\Gamma; \Delta \types P
		}{
			\Gamma; \Delta, k : \tend \types P
		}{\SWk}
%
%
%
		\andalso \andalso
		\tree {
			\Gamma, a : T; \Delta, \aggr x: \tdual T \types P
		}{
			\Gamma, a : T; \Delta \types \request{a}{\aggr x} P
		}{\TReq}
		\\[6mm]
		\tree {
			\Gamma, a : T;\Delta, x: T \types P
		}{
			\Gamma, a : T; \Delta \types \accept a{x}{P}
		}{\TAcc}
		\andalso \andalso \andalso
		\tree {
			\Gamma;\Delta, k: T \types P
			\andalso
			\Gamma \types e : \beta
		}{
			\Gamma;\Delta, k : \tout \beta T \types \pout{k}{e} P
		}{\TSnd}
		\\[8mm]
%
%
%
		\tree {
			\Gamma \types e: \beta \andalso \Gamma, x : \beta; \Delta, k: T \types P
		}{
			\Gamma;\Delta, k : \tin \beta T \types \pdin{k}{x}{e} P
		}{\TRcv}
		\andalso \andalso
		\tree[] {
			\Gamma; \Delta \types P_i
			\andalso
			i \in \set{1, 2}
		}{
			\Gamma; \Delta \types P_1 \Sum P_2
		}{\TSum}
		\\[8mm]
		\tree[] {
			\Gamma;\Delta, k: T_j \types P
			\andalso
			j \in I
			\andalso
			(\exists s \in \mathcal{S}, k = \aggr s) \lor (\exists x \in \mathcal{V}, k = \aggr x)  
		}{
			\Gamma;\Delta, k: \tselI{\ell}{T} \types \sel{k}{\ell_j} P 
		}{\TSel}
		\\[8mm]
		\tree[] {
			\forall i \in I, \Gamma; \Delta, \Delta', k : T_i \types P_i
			\andalso
			(\exists s \in \mathcal{S}, k = s) \lor (\exists x \in \mathcal{V}, k = x)  
			\andalso
			\Gamma; \Delta' \types R
			\\
			\forall s \in \mathcal{S}, \aggr s \notin \dom{\Delta}
		}{
		    \Gamma; \Delta, \Delta', k : \tbranchOn{\ell_i:T_i}_{i \in I} \types \branchDef{k}{\ell_i:P_i}{R}_{i \in I}
		}{\TBr}
		\\[11mm]
		\tree {
			\set{\Gamma; \Delta, \Delta_i \types P_i
			\andalso
			\forall s \in \mathcal{S}, \aggr{s} \notin \dom{\Delta_i}}_{i \in \set{1,2}} 
		}{
			\Gamma; \Delta, \Delta_1, \Delta_2 \types \cond{\econd}{P_1}{P_2}
		}{\TCond}
		\\[8mm]
		\tree {
			\Gamma' \types u_i: \beta_i \text{ iff } u_i \in \tilde{v}
			\andalso
			\tilde{y} = \dom{\Delta'}
			\andalso
			k_i: T_i \in \Delta \text{ iff } y_i: T_i \in \Delta'

		}{
			\Gamma, \abs{D}{\tilde{x}, \tilde{y}}: (\Gamma';\Delta'); \Delta \types \appl{D}{\tilde{v}, \tilde{k}}
		}{\TVar}
		\\[8mm]
		\tree {
			\forall i \in I, \abs{D_i}{\tilde{x}_i}: (\Gamma_i, \Delta_i) \in \Gamma \land \Gamma, \Gamma_i; \Delta_i \types P_i
			\andalso
			\Gamma; \Delta \types P
		}{
			\Gamma; \Delta \types \Def{\set{\abs{D_i}{\tilde{x}_i} \defeq P_i}_{i \in I}}{P}
		}{\TRec}
	\end{array}
	\]
	\caption{Typing rules for processes \label{fig:process_typing}}
\end{figure}

The typing judgement for processes, $\Gamma; \Delta \types P$,
is defined as the least
relation that satisfies the rules in \Figure{process_typing}.
Rule \TInact is standard for typing
the inactive process.
Rule \SWk defines standard session type weakening for 
context $\Delta$.
The next three rules are standard session type rules (cf.~\cite{Honda1998})
for typing session request, rule \TReq, session accept, rule \TAcc, and session
send, rule \TSnd, respectively. 

Rule \TRcv is used to type the \Rcv process, which requires for
the default expression $e$ to have the same type as variable $x$.
Rule \TSum types non-deterministic behaviour. The rule
expects the same type for each process combined by the \TSum operator.

The next two rules are standard rules for typing session select, rule \TSel,
and session branch, rule \TBr, prefixes, respectively.
The extra conditions on rules \TSel and \TBr ensure
correctness of the selection/branching interaction
as an interaction between $\aggr s$-endpoint and $s$-endpoint, respectively;
only $\aggr s$-endpoints and $\aggr x$ variables are typed
by rule \TSel, and similarly only $s$-endpoints and $x$ variables
are being typed by rule \TBr.
The recovery process in the branch prefixed process
is used to drop session endpoints including the $s$-endpoint,
thus typing rule \TBr ensures that the recovery process $R$
is typed with a subset of the linear context, excluding
the $s$-endpoint. Moreover, the rule ensures that only $s$-endpoints
are be dropped.

A novel rule is \TCond, that types the conditional process;
since the conditional process is used to drop session endpoints,
the typing rule splits the contexts $\Delta_i$ of the two branches
into common and non-common, i.e.~dropped, session endpoints.
The rule also ensures that only $s$-endpoints can be dropped.

Rules \TVar, and \TRec are standard. 
Rule \TVar requires that recursive variables are mapped in the
$\Gamma$ context and checks whether the arguments application
has the correct type with respect to context $\Gamma$.
Similarly rule \TRec requires to checks context $\Gamma$
and the recursive definition for type consistency.

\begin{figure}
	\[
		\begin{array}{c}
		\tree[] {
			\Gamma; \Theta \types B
			\andalso
			\Gamma;\Delta \types P
			\andalso
			\dom{\Delta} = \dom{\Theta}
			\andalso
		}{
			\Gamma; \Delta \sessionop \Theta \types \net{P \pll B}
		}{\TNode}
		\\[8mm]
%
%
		\tree[] {
			\Gamma, a : T; \Delta \types N
		}{
			\Gamma;\Delta \types \newn{a} N
		}{\TCRes}
		\tree {
			\set{\Gamma; \Delta, \Delta_i \types N_i}_{i \in \set{1, 2}}
			\andalso
			\forall s \in \mathcal{S}, \aggr s \notin \dom{\Delta}
		}{
			\Gamma;\Delta, \Delta_1, \Delta_2 \types N_1 \npll N_2
		}{\TPar}
		\\[8mm]
		\tree[] {
			\Gamma; \Delta, \Delta', \aggr s: \ctype{c}{\dual{T}} \types N
			\andalso
			\Delta' = s: \ctype{c}{T}
			\lor
			\Delta' = \econtext
			\andalso
			s \notin \dom{\Delta}
		}{
			\Gamma;\Delta \types \newn{s} N
		}{\TSRes}
		\end{array}
	\]
	\caption{Typing rules for networks \label{fig:network_typing}}
\end{figure}

The typing judgement for networks, $\Gamma; \Delta \types N$,
is defined as the least relation that satisfies the rules
in \Figure{network_typing}.
Rule \TNode type a network node. The rule combines the contexts
of the process, $P$, and the session buffer terms, $B$ using operator \sessionop.
Rule \TNode requires that $\dom{\Delta} = \dom{\Theta}$ implying that:
i) all the session names that appear
in the network process have their corresponding session buffer present
in the network node; and ii) process $P$ does not contain
free session variables, i.e.~$x: T \notin \Delta$.
%

Rule \TPar requires that parallel components of a network should
share the same type for common $s$-endpoints.
The next rule is \TCRes, which restricts shared name $a$
by removing it from $\Gamma$.
Finally, rule \TSRes captures the interaction intuition between
the $\aggr s$-endpoint and the $s$-endpoints.
A session name $s$ can be restricted whenever its two endpoints
are dual when in the same state $c$, or if only the
$\aggr s$-endpoint is present in the typing context.
The latter condition captures the case where the $s$-endpoints
were not created due to message loss or lost due to recovery.

The following example gives the typing for the \Recursive network in Example~\ref{ex:recursive_interaction}. 

\begin{exa}[Typing for Example~\ref{ex:recursive_interaction}]
	Given $\Gamma$
	such that $\Gamma \types v: \beta$ and
	for all $i \in I, \Gamma \types \mathsf{ack_i}: \mathsf{ackType}$,
	and moreover
	$
		T = \rec{t} \tout{\beta} \tin{\mathsf{ackType}} \tsel{\acceptLabel: \tend, \restart: \tvar{t}}
	$,
	the reader can verify that
	$
		\Gamma; \aggr s: T , s: \dual{T} \types \Recursive
	$.
	\qed
\end{exa}

\subsubsection*{Typing System for Runtime Processes}
The typing system presented so far can only type networks with
$s$-endpoints that are in the same state with the corresponding
$\aggr s$-endpoint. This is demonstrated with the following
example.
\begin{exa}[Heartbeat at Runtime]
	\label{ex:typing_synch}
	Consider an instance of Example~\ref{ex:intro_heartbeat}
	with two $s$-endpoints:
	\[
		\Heartbeat = \net{\pout{\aggr s}{\heartbeat} P_0 \pll \squeue{\aggr s}{0}{\equeue}}
		\,\,\npll\,\,
		\net{\pin{s}{x} P_1 \pll \squeue{s}{0}{\equeue}}
		\,\,\npll\,\,
		\net{\pin{s}{x} P_2 \pll \squeue{s}{0}{\equeue}}
	\]
	All the endpoints in the above network are in the same state.
	Assuming that value \heartbeat has type \theartbeat,
	the \Heartbeat network can by typed as
	\[
		\heartbeat: \theartbeat; \aggr{s}: (0, \tout{\theartbeat} \tend), s: (0, \tin{\theartbeat} \tend) \types \Heartbeat
	\]
	Following reduction rule \Broadcast we can observe an unreliable broadcast operation that results in
	\begin{eqnarray*}
		\Heartbeat \reduces \Heartbeat'
		= \net{P_0 \pll \squeue{\aggr s}{1}{\equeue}}
		\,\,\npll\,\,
		\net{\pin{s}{x} P_1 \pll \squeue{s}{1}{\heartbeat}}
		\,\,\npll\,\,
		\net{\pin{s}{x} P_2 \pll \squeue{s}{0}{\equeue}}
	\end{eqnarray*}
	The typing system presented so far
	cannot type the networks that may arise at runtime
	where some $s$-endpoints are not synchronised with
	the corresponding $\aggr s$-endpoint.
	In network $\Heartbeat'$, the second network node
	is typed as
	$\heartbeat: \theartbeat; s: (1, \tend) \types \net{\pin{s}{x} P_1 \pll \squeue{s}{1}{\heartbeat}}$
	and the third network node is typed as  
	$\heartbeat: \theartbeat; s: (0, \tin{\theartbeat} \tend) \types \net{\pin{s}{x} P_2 \pll \squeue{s}{0}{\equeue}}$.
	Thus, we cannot apply rule \TPar that requires the
	same type for the two $s$-endpoints.
\end{exa}

We can achieve typing by using the session state information
to construct the type information that was lost due to
unreliable communication.
The key is to use a typing rule to synchronise the
the $s$-endpoints that are not synchronised with the
$\aggr s$-endpoint.

We define type advancement as a transition relation on types.
The type advancement relation is used to define the notion
of endpoint synchronisation.
%
\begin{defi}[Type Advancement]
	\label{def:type-advancement}
	Relation $T \tadvance^n T'$ is defined as:
	\[
	\begin{array}{c}
		\tin \beta T \tadvance^1 T
		\qquad \qquad
		\tout \beta T \tadvance^1 T
		\qquad \qquad
		\infrule{k \in I}{\tselI{\ell}{T} \tadvance^1 T_k}
		\qquad \qquad
		\infrule{k \in I}{\tbranchI{\ell}{T} \tadvance^1 T_k}
		\\[6mm]
		T \tadvance^0 T
		\qquad \qquad
		\infrule{T \tadvance^n T'' \quad\ T'' \tadvance^1 T'}{T \tadvance^{n+1} T'}
	\end{array}
	\]
\end{defi}
It is useful to distinguish output advance by writing
$\tout \beta T \tadvance_o T$,
extended to $T \tadvance_o^{n} T$ in the standard way.
Intuitively, $T \tadvance^n T'$ says that
$T'$ is reached in $n$ advancements from $T$.
Type advancement for recursive types is obtained by expansion.

Next, we use type advancement (\Definition{def:type-advancement})
and the session state information within a linear context to
define the linear context synchronisation relation over
linear contexts.
\begin{defi}[Linear Context Synchronisation]
	\label{def:synchr}
	We define the relation $\synchronise{\Delta}{ \Delta' }$ inductively as follows.
	\[
		\synchronise{\econtext}{\econtext}
		\andalso \qquad
		\tree {
			\synchronise{\Delta}{\Delta'} \andalso 
			T \tadvance^{m - n} T' \andalso 
			n \leq m
		}{
			\synchronise{\Delta, s: \ctype{n}{T}}{\Delta', s: \ctype{m}{T'}}
		}{}
		\andalso \qquad
		\tree {
			\synchronise{\Delta}{\Delta'} \andalso 
			T' \tadvance_{o}^{n - m} T \andalso 
			n \geq m
		}{
			\synchronise{\Delta, s: \ctype{n}{T}}{\Delta', s: \ctype{m}{T'}}
		}{}
	\]
\end{defi}
Linear context synchronisation is crucial for the correctness of the type system
and it used for achieving $s$-endpoint synchronisation.
The middle rule states that the type of an $s$-endpoint can use type
advancement to proceed its state within a linear context.
The right rule describes the opposite case where the type
of an $s$-endpoint uses output type advancement, to reconstruct
a previous session state.
The requirement for output type advancement
is aligned with the fact that a non-synchronised $s$-endpoint may
use reduction rule \Loss, that drops output session prefixes.
Finally the left rule is a basic case rule for the
linear context.

\begin{figure}
	\[
	\begin{array}{c}
		\tree[] {
			\Gamma; \Delta' \types N
			\andalso
			\synchronise{\Delta'}{\Delta}
		}{
			\Gamma; \Delta \types N
		}{\TSynch}
	\end{array}
	\]
	\caption{Runtime Typing Rule \label{fig:typing_runtime}}
\end{figure}

Figure~\ref{fig:typing_runtime} defines rule \TSynch.
Rule \TSynch, is used to provides static guarantees
of type duality between all non-synchronised sessions
endpoints that may arise during execution.
The rule uses linear context synchronisation,
$\synchronise{\Delta}{ \Delta'}$, to synchronise
non-synchronised $s$-endpoints.

Rule \TSynch types only networks that result at runtime,
where due to unreliability, session endpoints become
unsynchronised, and it is important to prove the properties
of progress (Theorem~\ref{thm:progress}) and recovery (Theorem~\ref{thm:session_recovery})
in Section~\ref{sec:soundness}.
The main use of the rule is to align non-synhronised
$s$-endpoints in a linear context in order to apply typing
rule $\TPar$, as well as enforcing the notion of duality
between $\aggr s$-endpoint and $s$-endpoint as in typing
rule $\TSRes$.
Inspecting rule \TSynch from the point of view of single
network node it seems that \TSynch adds a degree of
non-determinism during runtime typing.
However, at the network level, the state counter of the
$\aggr s$-endpoint can deterministically guide the computation of the
linear context synchronisation for the corresponding 
$s$-endpoints within a network, i.e., an implementation
of runtime checking will seek, during linear context synchornisation
to align the counter of the $s$-endpoints with the counter
of the $\aggr s$-endpoint. 

%
\begin{exa}[Typing Heartbeat at Runtime]
	For example we can use \TSynch to synchronise the
	type of network node $\net{\pin{s}{x} P_2 \pll \squeue{s}{0}{\equeue}}$
	in Example~(\ref{ex:typing_synch}).
	\[
		\tree[] {
			\heartbeat: \theartbeat; s: (0, \tin{\theartbeat} \tend) \types \net{\pin{s}{x} P_2 \pll \squeue{s}{0}{\equeue}}
			\qquad
			\synchronise{s: (0, \tin{\theartbeat} \tend)}{s: (1, \tend)}
		}{
			\heartbeat: \theartbeat; s: (1, \tend) \types \net{\pin{s}{x} P_2 \pll \squeue{s}{0}{\equeue}}
		}{\TSynch}
	\]
	This leads to typing judgement
	\[
		\heartbeat: \theartbeat; \aggr{s}: (1, \tend), s: (1, \tend) \types \Heartbeat'
	\]
\end{exa}
%
%
%

A network $N$, resp.~process $P$, is called well-typed,
whenever $\Gamma;\Delta \types N$, resp.~$\Gamma;\Delta \types P$,
for some contexts $\Gamma$ and $\Delta$. 

\subsection{Typing Derivation Example}
We present the typing derivation for the  $\Heartbeat_1$
network in Example~\ref{ex:reduction_heartbeat}.
The example demonstrates linear context synchronisation
which is a main notion of the typing system.
\begin{exa}[Typing derivation for the Heartbeat Protocol]
\label{ex:heartbeat_typing}
	Recall, network $\Heartbeat_1$ in Example~\ref{ex:reduction_heartbeat}:
	\[
		\begin{array}{rcl}
			\Heartbeat_1 &=&  
			\newnp{s}{
					\net{\pin{\aggr s}{x} \pin{\aggr s}{x} P_0 \pll \squeue{\aggr s}{0}{(0, \heartbeat_2), (1, \heartbeat_2), (0, \heartbeat_1)}}
				\\	&\npll&
					\net{\pout{s}{\heartbeat_1} P_1 \pll \squeue{s}{1}{\equeue}}
					\,\,\npll\,\,
					\net{P_2 \pll \squeue{s}{2}{\equeue}}
				}
		\end{array}
	\]
	Also, consider that the heartbeat message has type \theartbeat.
	We first type network node
	\[
		\net{\pin{\aggr s}{x} \pin{\aggr s}{x} P_0 \pll \squeue{\aggr s}{0}{(0, \heartbeat_2), (1, \heartbeat_2), (0, \heartbeat_1)}}
	\]
	We give the type derivation for buffer
	$\squeue{\aggr s}{0}{(0, \heartbeat_2), (1, \heartbeat_2), (0, \heartbeat_1)}$.
	\begin{eqnarray}
		\label{ex:heart_beat_typing:buffer_derivation}
		\btree {
			\btree {
				\Gamma; \aggr s: \ctype{0}{\tempty} \types \squeue{\aggr s}{0}{\equeue}\ \SEmp
				\andalso \andalso
				\equeue = \newbuffer{(0, \heartbeat_2), (0, \heartbeat_1)}{0}
				\\[1mm]
				\Gamma \types \gthr{(0, \heartbeat_2), (0, \heartbeat_1)}{0} : \theartbeat
				\\[1mm]
			}{
				\Gamma; \aggr s: \ctype{1}{\mtout{\theartbeat}} \types \squeue{\aggr s}{0}{(0, \heartbeat_2), (0, \heartbeat_1)}
			}{\LExp}
			\\[9mm]
			(0, \heartbeat_2), (0, \heartbeat_1) = \newbuffer{(0, \heartbeat_2), (1, \heartbeat_2), (0, \heartbeat_1)}{1}
			\\[1mm]
			\Gamma \types \gthr{(0, \heartbeat_2), (1, \heartbeat_2), (0, \heartbeat_1)}{1} : \theartbeat
			\\[1mm]
		}{
			\Gamma; \aggr s: \ctype{2}{\mtout{\theartbeat}. \mtout{\theartbeat}} \types \squeue{\aggr s}{0}{(0, \heartbeat_2), (1, \heartbeat_2), (0, \heartbeat_1)}
		}{\LExp}
	\end{eqnarray}
	The typing for the buffer constructs a derivation using typing rules:
	i) \SEmp to initially type the empty buffer;
	ii) \LExp to  buffer $\squeue{\aggr s}{0}{(0, \heartbeat_2), (0, \heartbeat_1)}$ that contains
	only the heartbeat messages send at session state $0$ of the protocol; and
	iii) \LExp again to type the entire buffer.
	The instances of the \LExp typing rule make use of operations $\newbuffer{\cdot}{c}$
	and $\gthr{\cdot}{c}$ to type the messages at each session state.
	The network node is then typed using an instance of rule \TNode:
	\begin{eqnarray}
		\label{ex:heart_beat_typing:node_0_derivation}
		\btree{
			\btree{
				\Gamma; \aggr s: \dual{T} \types P_0
			}{
				\Gamma; \aggr s: \tin{\theartbeat} \tin{\theartbeat} \dual{T} \types \pin{\aggr s}{x} \pin{\aggr s}{x} P_0
			}{\TRcv}
			\\[6mm]
			\Gamma; \aggr s: \ctype{2}{\mtout{\theartbeat}. \mtout{\theartbeat}} \types \squeue{\aggr s}{0}{(0, \heartbeat_2), (1, \heartbeat_2), (0, \heartbeat_1)} \quad \ref{ex:heart_beat_typing:buffer_derivation}
			\\[1mm]
			\aggr s: \tin{\theartbeat} \tin{\theartbeat} \dual{T} \sessionop \aggr s: \ctype{2}{\mtout{\theartbeat}. \mtout{\theartbeat}} = \aggr s: \ctype{2}{\dual{T}}
			\\[1mm]
		}{
			\Gamma; \aggr s: \ctype{2}{\dual{T}} \types \net{\pin{\aggr s}{x} \pin{\aggr s}{x} P_0 \pll \squeue{\aggr s}{0}{(0, \heartbeat_2), (1, \heartbeat_2), (0, \heartbeat_1)}}
		}{\TNode}
	\end{eqnarray}
	The rule uses the \sessionop operator to combine
	tuple $\ctype{2}{\mtout{\theartbeat}. \mtout{\theartbeat}}$,
	derived from typing derivation~\ref{ex:heart_beat_typing:buffer_derivation},
	and session type $\tin{\theartbeat} \tin{\theartbeat} \dual{T}$,
	derived from typing rule \TRcv. The result is
	tuple $\ctype{2}{\dual{T}}$, denoting that the $\aggr s$-endpoint
	is in state $2$ and has type $\dual{T}$.

	We continue by typing network node
	$\net{P_2 \pll \squeue{s}{2}{\equeue}}$ using an instance of typing rule \TNode:
	\begin{eqnarray}
		\label{ex:heart_beat_typing:node_2_derivation}
		\btree {
			\Gamma; s: \ctype{2}{\tempty} \types \squeue{s}{2}{\equeue} \  \SEmp
			\andalso \ 
			\Gamma; s: T \types P_2
			\andalso \ 
			s: T \sessionop s: \ctype{2}{\tempty} = s: \ctype{2}{T}
		}{
			\Gamma; s: \ctype{2}{T} \types \net{P_2 \pll \squeue{s}{2}{\equeue}}
		}{\TNode}
	\end{eqnarray}
	Similarly to the previous derivation,
	the \sessionop operator is used to combine
	tuple $\ctype{2}{\tempty}$,
	derived from rule \SEmp,
	and session type $T$,
	derived from typing judgement $\Gamma; s: T \types P_2$.
	The result is
	tuple $\ctype{2}{T}$,
	which means that the $s$-endpoint used by
	the network node 
	is in state $2$ and has type $T$.
	
	The network node $\net{\pout{s}{\heartbeat_1} P_1 \pll \squeue{s}{1}{\equeue}}$
	is also typed using an instance of rule \TNode:
	\begin{eqnarray}
		\label{ex:heart_beat_typing:node_1_derivation}
		\btree {
			\btree[\TSnd] {
				\Gamma; s: T \types P_1
				\andalso
				\Gamma \types \heartbeat_1: \theartbeat
			}{
				\Gamma; s: \tout{\theartbeat} T \types \pout{s}{\heartbeat_1} P_1
			}{}
			\ \ 
			\begin{array}{l}
				\Gamma; s: \ctype{1}{\tempty} \types \squeue{s}{1}{\equeue} \quad \SEmp
				\\
				s: \tout{\theartbeat} T \sessionop s: \ctype{1}{\tempty} = s: \ctype{1}{\tout{\theartbeat} T}
			\end{array}
			\\[9mm]
		}{
			\Gamma; s: \ctype{1}{\tout{\theartbeat} T} \types \net{\pout{s}{\heartbeat_1} P_1 \pll \squeue{s}{1}{\equeue}}
		}{\TNode}
	\end{eqnarray}
	The $s$-endpoint used by the network node 
	is in state $1$ and has type $\tout{\theartbeat} T$.
	However, it is not in the same state with the
	$\aggr s$-endpoint in derivation~\ref{ex:heart_beat_typing:node_0_derivation},
	therefore we apply rule
	\TSynch to achieve endpoint synchronisation.
	\begin{eqnarray}
		\label{ex:heart_beat_typing:synch_derivation}
		\btree {
			\Gamma; s: \ctype{1}{\tout{\theartbeat} T} \types \net{\pout{s}{\heartbeat_1} P_1 \pll \squeue{s}{1}{\equeue}}
			\  \ref{ex:heart_beat_typing:node_1_derivation}
			\andalso
			\tree{
				\tout{\theartbeat} T \tadvance^{2-1} T
			}{
				\synchronise{s: \ctype{1}{\tout{\theartbeat} T}}{s: \ctype{2}{T}}
			}{}
		}{
			\Gamma; s: \ctype{2}{T} \types \net{\pout{s}{\heartbeat_1} P_1 \pll \squeue{s}{1}{\equeue}}
		}{\TSynch}
	\end{eqnarray}
	The \TSynch rule uses linear context synchronisation, Definition~\ref{def:synchr},
	to synchronise type $\tout{\theartbeat} T$ from state $1$ to state $2$ and
	get $\ctype{2}{T}$.

	We then apply rule \TPar twice to get:
	\begin{eqnarray}
		\label{ex:heart_beat_typing:parallel_derivation}
		\btree {
			\Gamma; s: \ctype{2}{T} \types \net{\pout{s}{\heartbeat_1} P_1 \pll \squeue{s}{1}{\equeue}}
			\ \ref{ex:heart_beat_typing:synch_derivation}
			\andalso \andalso
			\Gamma; s: \ctype{2}{T} \types \net{P_2 \pll \squeue{s}{2}{\equeue}}
			\ \ref{ex:heart_beat_typing:node_2_derivation}
		}{
			\Gamma; s: \ctype{2}{T} \types \net{\pout{s}{\heartbeat_1} P_1 \pll \squeue{s}{1}{\equeue}} \npll \net{P_2 \pll \squeue{s}{2}{\equeue}}
		}{\TPar}
	\end{eqnarray}
	and also
	\begin{eqnarray}
		\label{ex:heart_beat_typing:network_derivation}
		\arraycolsep=1pt
		\btree {
			\begin{array}{cl}
				\Gamma; s: \ctype{2}{T} \types \net{\pout{s}{\heartbeat_1} P_1 \pll \squeue{s}{1}{\equeue}} \npll \net{P_2 \pll \squeue{s}{2}{\equeue}}
				& \quad \ref{ex:heart_beat_typing:parallel_derivation}
				\\[1mm]
				\Gamma; \aggr s: \ctype{2}{\dual{T}} \types \net{\pin{\aggr s}{x} \pin{\aggr s}{x} P_0 \pll \squeue{\aggr s}{0}{(0, \heartbeat_2), (1, \heartbeat_2), (0, \heartbeat_1)}}
				& \quad \ref{ex:heart_beat_typing:node_0_derivation}
			\end{array}
			\\[4mm]
		}{
			\Gamma; \aggr s: \ctype{2}{\dual{T}}, s: \ctype{2}{T} \types
			\begin{array}{rcl}
				&&\net{\pin{\aggr s}{x} \pin{\aggr s}{x} P_0 \pll \squeue{\aggr s}{0}{(0, \heartbeat_2), (1, \heartbeat_2), (0, \heartbeat_1)}}
				\\&\npll&
				\net{\pout{s}{\heartbeat_1} P_1 \pll \squeue{s}{1}{\equeue}} \npll \net{P_2 \pll \squeue{s}{2}{\equeue}}
			\end{array}
		}{\TPar}
	\end{eqnarray}
	Finally, we use rule \TSRes to restrict session $s$ and
	provide with the typing derivation
	for network $\Heartbeat_1$:
	\[
		\arraycolsep=1pt
		\btree {
			\Gamma; \aggr s: \ctype{2}{\dual{T}}, s: \ctype{2}{T} \types
			\begin{array}{rcl}
				&&\net{\pin{\aggr s}{x} \pin{\aggr s}{x} P_0 \pll \squeue{\aggr s}{0}{(0, \heartbeat_2), (1, \heartbeat_2), (0, \heartbeat_1)}}
				\\&\npll&
				\net{\pout{s}{\heartbeat_1} P_1 \pll \squeue{s}{1}{\equeue}} \npll \net{P_2 \pll \squeue{s}{2}{\equeue}}
			\end{array}
			\quad \ref{ex:heart_beat_typing:network_derivation}
			\\[4mm]
		}{
			\Gamma; \econtext \types \Heartbeat_1
		}{\TSRes}
	\]
	Both endpoints are in the same state, $2$,
	and have dual types, thus safe interaction
	endpoint interaction respects session duality
	and session $s$ can be restricted.
	\qed
\end{exa}


\section{Type Soundness, Type Safety, and Progress}
\label{sec:soundness}

In this Section we prove that the proposed type system is sound 
via a type preservation theorem and safe via a type safety theorem.
Before we proceed with the main results we deploy the necessary
technical machinery and auxiliary results.

%

We define the notion of the well-formed Linear Context.
\begin{defi}[Well-formed Linear Context]
	\label{def:well-formed}
	A context $\Delta$ is {\em well-formed} whenever
	$\aggr s: \ctype{c}{T} \in \Delta$ implies
	either
	\begin{itemize}
		\item	$s: \ctype{c}{\dual{T}} \in \Delta$; or
		\item	$s \notin \dom{\Delta}$.
	\end{itemize}
\end{defi}
Well-formed linear context $\Delta$ requires that whenever
the $\aggr s$-endpoint is present in $\Delta$,
then the corresponding $s$-endpoint, if present in $\Delta$,
needs to be syncrhonised with the $\aggr s$-endpoint and,
additionally,
have dual type with respect to the $\aggr s$-endpoint.
The case where a corresponding $s$-endpoint is not present in
$\Delta$ captures the case where no $s$-endpoints were
created or when all $s$-endpoints were dropped.




The next definition captures the interaction of processes
at the type level.
\begin{defi}[Linear Context Advancement]
	\label{def:lin_context_adv}
	Advancement relation, $\tadvance$, over linear contexts is defined as:
	\begin{itemize}
		\setlength\itemsep{0mm}
		\item	$\Delta, \aggr s: \ctype{c}{\tout{\beta} T_1}, s: \ctype{c}{\tin{\beta} T_2} \tadvance \Delta, \aggr s: \ctype{c+1}{T_1}, s: \ctype{c+1}{T_2}$.
		\item	$\Delta, \aggr s: \ctype{c}{\tsel{l_i: T_i}}, s: \ctype{c}{\tbranchOn{l_i: T_i'}} \tadvance \Delta, \aggr s: \ctype{c+1}{T_k}, s: \ctype{c+1}{T_k'}$.
		\item	$\Delta, s: \ctype{c}{\tout{\beta} T_1}, \aggr s: \ctype{c}{\tin{\beta} T_2} \tadvance \Delta, s: \ctype{c+1}{T_1}, \aggr s: \ctype{c+1}{T_2}$
		\item	$\Delta, \set{s: \ctype{c}{T_i}}_{i \in I} \tadvance \Delta$
	\end{itemize}
\end{defi}

The first two cases of linear context advancement define
session interaction at type level and align with broadcasting
and selection interactions.
The third case aligns with the session gather interaction. 
The fourth case 
indicates that  a linear context advances by reducing in size.
The definition aligns with the fact that processes
may drop $s$-endpoints due to rules \Recover, \BRecover, \True, and \False.

The next lemma shows that the well-formedness of $\Delta$
is preserved by linear context inclusion and linear context advancement.

\begin{lem}
	Let $\Delta$ be well-formed.
			If $\Delta \tadvance \Delta'$ then $\Delta'$ is well-formed.
\end{lem}
%
%

The next Lemma is a consequence of the typing system.
\begin{lem}
	\label{lem:send_prefix}
	Consider networks
	\begin{itemize}
		\item	If $N_1 \equiv \net{ \pout{\aggr{s}}{e} P \pll B \pll \squeue{\aggr s}{c}{\pol{h}}}$
			and $\Gamma; \Delta \types N_1$ with $\Delta$ well-formed, then $\pol{h} = \equeue$.
		\item	If $N_2 \equiv \net{ \pout{s}{e} P \pll B \pll \squeue{s}{c}{\pol{m}}}$
			and $\Gamma; \Delta \types N_2$ with $\Delta$ well-formed, then $\pol{m} = \equeue$.
		\item	If $N_3 \equiv \net{ \sel{\aggr{s}}{l} P \pll B \pll \squeue{s}{c}{\pol{m}}}$
			and $\Gamma; \Delta \types N_3$ with $\Delta$ well-formed, then $\pol{m} = \equeue$.
	\end{itemize}
\end{lem}

\begin{proof}
	The proof follows the requirement
	on rule \TNode and the fact that operations:
	i) $\tout{\beta} T \sessionop \mtout{\beta'}. M$;
	$\tout{\beta} T \sessionop \msel{\ell}. M$;
	$\tsel{\ell: T} \sessionop \mtout{\beta'}. M$; and
	$\tsel{\ell: T} \sessionop \msel{\ell}. M$
	are undefined.
\end{proof}

The above lemma states a standard property for
asynchronous session type systems (cf.~\cite{DBLP:journals/jfp/GayV10}), which
requires that in a session typed setting
whenever a network node has a send prefix,
the corresponding session buffer is necessarily empty.
This is because correct send/receive interaction will
consume any messages in the session buffer prior to the
send prefix.

We can now state and prove the Typing Preservation Theorem.

\begin{thm}[Typing Preservation]
	\label{thm:type-preservation}
	If	$\Gamma; \Delta \types N$ and $N \reduces N'$
	and $\Delta$ well-formed,
	then there exist well-formed $\Delta'$
	such that
		$\Delta \tadvance \Delta'$
	and
		$\Gamma;\Delta' \types N'$.
\end{thm}
\begin{proof}
	By induction on the depth of the derivation of $N \reduces N'$.
	For the full proof see Appendix~\ref{app:type_preservation} in page \pageref{app:type_preservation}.
\end{proof}

Type preservation theorem states that
a reduction maintains typing and well-formedness.

Towards the statement of a type safety theorem, we
proceed by defining the notion of the error network.
The class of error networks indicates all networks that
should not be typed with a well-formed linear context by the
proposed typing system.
We define the notion of an error network, cf.~\cite{YOSHIDA200773},
as the network that contains an invalid $s$-pair that cannot make a safe communication interaction.

\begin{defi}[Error Network]
\label{def:errornetwork}
	Let $s$-prefix be a network of the form
	\[
	\begin{array}{rclcrcl}
		\Nbrc^{c} &=& \net{\pout{\aggr{s}}{e} P_1 \pll \squeue{\aggr{s}}{c}{\equeue} \pll B}
		&\quad&
		\Ngth^c &=& \net{\pdin{\aggr{s}}{x}{e} P_2 \pll \squeue{\aggr{s}}{c}{\pol{h}} \pll B}
		\\
		\Nuni^c &=& \net{\pout{s}{e'} P_3 \pll \squeue{s}{c}{\equeue} \pll B}
		&&
		\Nrcv^c &=& \net{\pin{s}{x'} P_4 \pll \squeue{s}{c}{\pol{m}} \pll B}
		\\
		\Nsel^c &=& \net{\sel{\aggr{s}}{\ell} P_5 \pll \squeue{\aggr{s}}{c}{\equeue} \pll B}
		&&
		\Nbra^c &=& \net{\branchDef{s}{\ell_i: P_i}{R}_{i \in I} \pll \squeue{s}{c}{\pol{m}} \pll B}
	\end{array}
	\]
%
	An \emph{invalid} $s$-pair is one of the following
	parallel compositions of  $s$-prefixes:
	\[
	\begin{array}{rcl c  rcl c rcl c rcl cc}
		\Nbrc^{c} &\npll& \Nbrc^{c'}
		&\hspace{1em}&
		\Ngth^{c}&\npll&\Ngth^{c'}
		&\hspace{1em}&
		\Nsel^{c}&\npll&\Nsel^{c'}
		&\hspace{1em}
		\\
		\Nbrc^{c}&\npll&\Ngth^{c'}
		&&
		\Nbrc^{c}&\npll&\Nsel^{c'}
		&&
		\Nbrc^{c}&\npll&\Nuni^{c}
		&&
		\Nbrc^{c}&\npll&\Nbra^{c}
		\\
		\Ngth^{c}&\npll&\Nsel^{c'}
		&&
		\Ngth^{c}&\npll&\Nrcv^{c}
		&&
		\Ngth^{c}&\npll&\Nbra^{c}
		\\
		\Nsel^{c}&\npll&\Nrcv^{c}
		&&
		\Nsel^{c}&\npll&\Nuni^{c}
	\end{array}
	\]
	\noindent
	A network $N$ is called an \emph{error network} whenever there exists an
	invalid $s$-pair $M$ such that for some network $N'$ it holds that
	$
		N \equiv \newnp{\seq n}{N' \npll M}
	$.
\end{defi}
An invalid $s$-pair is formed by a parallel
composition of two $s$-prefixes. All other parallel compositions between
two $s$-prefixes are considered valid. An error network is any network
that composes in parallel at least one invalid $s$-pair.  
Equivalently, a valid network is a parallel composition of network terms of which
all $s$-prefixes that are in the same session state, consists of at most one
broadcast prefix (resp.~gather, selection) and many receive prefixes (resp. send,
branch).

An invalid $s$-pair is either a composition of a $\aggr s$-prefix and a $s$-prefix that cannot
safely interact, or a composition of two $\aggr s$-prefixes that violates linearity
conditions.
In the latter case the composition of two $\aggr s$-endpoint
is always an error regardless of the state they are in.
In the former case we require that the two $s$-prefixes are in the same state $c$,
because
for such $s$-redexes not in the same state, reduction semantics will never allow an
interaction, other than recovery (reduction rules \Recover, \BRecover, \Loss).
The next example presents instances of error networks.

\begin{exa}[Error Network]
	\label{ex:error_network}
	As a first example, consider an instance of error network $\Nbrc^c\npll\Nbra^c$
	\[
		\net{\pout{\aggr{s}}{e} P \pll \squeue{\aggr{s}}{c}{\equeue}}
		\,\,\npll\,\,
		\net{\branchDef{s}{\ell_i: P_i}{R}_{i \in I} \pll \squeue{s}{c}{\pol{m}}}
	\]
	The network cannot perform a safe interaction on channel
	$s$, because the $\aggr s$-endpoint does not select a
	continuation on the $s$-endpoint, therefore
	the above network cannot observe a reduction. 
	The network is not typable, because two session endpoints ($\aggr s$-endpoint and $s$-endpoint)
	do not have dual types.

	A more interesting example is given by an instance of the error network
	$N \npll \Nbrc^c \npll \Nbrc^c$
	\[
		\net{\pout{\aggr{s}}{e} \nil \pll \squeue{\aggr{s}}{c}{\equeue}}
		\,\,\npll\,\,
		\net{\pout{\aggr{s}}{e} \nil \pll \squeue{\aggr{s}}{c}{\equeue}}
		\,\,\npll\,\,
		\net{\pin{s}{x} \nil \pll \squeue{s}{c}{\equeue}}
	\]
	The network cannot perform a safe interaction on channel
	$s$, because the $s$-endpoint can interact with either
	of the $\aggr s$-endpoints.
	Network $N$ is not typable, because of the requirement on
	rule \TPar that the $\aggr s$-endpoint is linear. 
\end{exa}

We also do not consider as invalid $s$-pair the case where two $s$-endpoints are
composed, i.e.~redexes of the form $\Nrcv^c \npll \Nuni^c, \Nrcv^c \npll \Nbra^c$ and $\Nuni^c \npll \Nbra^c$.
Such redexes may be well-typed, due to the requirement for linear context synchronisation on the \TSynch rule.
However, 
composing them with a corresponding $\aggr s$-endpoint 
results in an non well-typed network.
%
Disallowing such $s$-redexes would require more complicated semantics
that record interaction sequences instead of session state.
The next example clarifies this last intuition.

\begin{exa}
	Consider for example an instance of network $\Nrcv^0 \npll \Nuni^0$:
	\[
		\net{\pin{s}{x} \nil \pll \squeue{s}{0}{\equeue}}
		\,\,\npll\,\,
		\net{\pout{s}{e} \nil \pll \squeue{s}{0}{\equeue}}
	\]
	The network may be well-typed, due to linear context synchronisation,
	with well-formed context $\Delta = s:\ctype{1}{s:\tend}$ using derivation:
	\begin{eqnarray}
		\label{ex:error_network1}
		\btree {
			\btree {
				\Gamma; s: \tin{\beta_1}{\tend} \types \pin{s}{x} \nil
				\andalso
				\Gamma; s: \ctype{0}{\tempty} \types \squeue{s}{0}{\equeue}
				\\
				s: \tin{\beta_1}{\tend} \sessionop s: \ctype{0}{\tempty} = s: \ctype{0}{\tin{\beta_1}{\tend}}
			}{
				\Gamma; \ctype{0}{\tin{\beta_1}{\tend}} \types \net{\pin{s}{x} \nil \pll \squeue{s}{0}{\equeue}}
			}{\TNode}
			\\[8mm]
			\synchronise{s: \ctype{0}{\tin{\beta_1}{\tend}}}{s:\ctype{1}{\tend}}
		}{
			\Gamma; s:\ctype{1}{s:\tend} \types \net{\pin{s}{x} \nil \pll \squeue{s}{0}{\equeue}}
		}{\TSynch}
	\end{eqnarray}
	\begin{eqnarray}
		\label{ex:error_network2}
		\btree {
			\btree {
				\Gamma; s: \tout{\beta_2}{\tend} \types \pout{s}{e} \nil
				\andalso
				\Gamma; s: \ctype{0}{\tempty} \types \squeue{s}{0}{\equeue}
				\\
				s: \tout{\beta_2}{\tend} \sessionop s: \ctype{0}{\tempty} = s: \ctype{0}{\tout{\beta_1}{\tend}}
			}{
				\Gamma; s: \ctype{0}{\tout{\beta_1}{\tend}} \types \net{\pout{s}{e} \nil \pll \squeue{s}{0}{\equeue}}
			}{\TNode}
			\\[8mm]
			\synchronise{s: \ctype{0}{\tout{\beta_2} \tend}}{s: \ctype{1}{\tend}}
		}{
			\Gamma; s: \ctype{1}{\tend} \types \net{\pout{s}{e} \nil \pll \squeue{s}{0}{\equeue}}
		}{\TSynch}
	\end{eqnarray}
	that allows the usage of rule \TPar to get
	\begin{eqnarray*}
		\btree{
			\begin{array}{ll}
				\Gamma; s: \ctype{1}{\tend} \types \net{\pin{s}{x} \nil \pll \squeue{s}{0}{\equeue}}
				& \ref{ex:error_network1}
				\\[1mm]
				\Gamma; s: \ctype{1}{\tend} \types \net{\pout{s}{e} \nil \pll \squeue{s}{0}{\equeue}}
				&\ref{ex:error_network2}
			\end{array}
			\\[4mm]
		}{
			\Gamma; s: \ctype{1}{\tend} \types
			\net{\pin{s}{x} \nil \pll \squeue{s}{0}{\equeue}}
			\,\,\npll\,\,
			\net{\pout{s}{e} \nil \pll \squeue{s}{0}{\equeue}}
		}{\TPar}
	\end{eqnarray*}
%
%
	However, attempting to add a network node that 
	implements the $\aggr s$-endpoint would result in an non well-typed network.
	For example, networks:
	\[
		\begin{array}{rcl}
			N_1 &=& \net{\pout{\aggr s}{e} \nil \pll \squeue{\aggr s}{0}{\equeue}}
			\,\,\npll\,\,
			\net{\pin{s}{x} \nil \pll \squeue{s}{0}{\equeue}}
			\,\,\npll\,\,
			\net{\pout{s}{e} \nil \pll \squeue{s}{0}{\equeue}}
			\\
			N_2 &=& \net{\pin{\aggr s}{x} \nil \pll \squeue{\aggr s}{0}{\equeue}}
			\,\,\npll\,\,
			\net{\pin{s}{x} \nil \pll \squeue{s}{0}{\equeue}}
			\,\,\npll\,\,
			\net{\pout{s}{e} \nil \pll \squeue{s}{0}{\equeue}}
		\end{array}
	\]
	are both non well-typed.
	In network $N_1$,
	the $\aggr s$-endpoint is typed
	as $\Gamma; \aggr s: \ctype{0}{\tout{\beta} \tend} \types \net{\pout{\aggr s}{e} \nil \pll \squeue{\aggr s}{0}{\equeue}}$
	and cannot be synchronised with the type of network node
	$\net{\pout{s}{e} \nil \pll \squeue{s}{0}{\equeue}}$.
	Similarly in network $N_2$,
	the $\aggr s$-endpoint is typed
	as $\Gamma; \aggr s: \ctype{0}{\tin{\beta} \tend} \types \net{\pin{\aggr s}{x} \nil \pll \squeue{\aggr s}{0}{\equeue}}$
	and cannot be synchronised with the type of network node
	$\net{\pin{s}{x} \nil \pll \squeue{s}{c}{\equeue}}$.
\end{exa}

The next theorem shows that our framework enjoys type safety.
\begin{thm}[Type Safety]
	\label{thm:type_safety}
	Let network $N$ such that
	$\Gamma; \Delta \types N$ for some $\Gamma$ and some well-formed $\Delta$.
	If there exists network $N'$ such that $N \reduces^* N'$,
	then network $N'$ is \emph{not} an error network.
\end{thm}
\begin{proof}
	From Theorem~\ref{thm:type-preservation} and the
	fact that an error network is not well-typed
	with a well-formed linear context.
\end{proof}
Type safety states that a network typed with
a well-formed context can always interact
safely, i.e.~never reduce to an error.

\subsection{Progress}

We prove that a network typed with a well-formed session
environment can ensure strong progress properties.
We formally introduce the notion of a deadlocked network.

\begin{defi}[Shared name prefixed process]
	A process is called {\em shared input prefixed process}
	if it has 
	the form $\accept{a}{x} P$.
	A process is called 
	{\em shared output prefixed process}
	if it has the form $\request{a}{\aggr x} P$.
\end{defi}

\begin{defi}[Deadlocked Network]
	A network $N$ is called {\em deadlocked} whenever
	\begin{eqnarray*}
		\textstyle
		N \equiv \newnp{\tilde{n}}{\Par{i}{I}{\net{\sum_{j \in J_i} P_j \pll B_i}}}
	\end{eqnarray*}
	and
			for all $i \in I$ we have that for all $j \in J_i$, $P_j$
			is prefixed with an input shared name.
\end{defi}
It is easy to show that whenever a network $N$ is a deadlocked
network, it holds that there exists no $N'$ such that $N \reduces N'$.
We can now formulate a basic progress result:
\begin{thm}[Progress]
	\label{thm:progress}
	Let $N$ be a network such that $\Gamma; \Delta \types N$
	for some $\Gamma$ and some well-formed $\Delta$.
	Then either
	\begin{itemize}
		\item	$N \equiv \Par{i}{j}{\net{\nil}}$; or
		\item	$N$ is a deadlocked network; or
		\item	$N \reduces N'$ for some $N'$.
	\end{itemize}
\end{thm}

\begin{proof}
	The proof is straightforward since all instances of a
	session prefixed processes in a network
	node can perform an interaction, e.g.~due to recovery semantics,
	whenever they are well-typed.
	Processes that are prefixed with an output on shared name 
	in a network can also perform an interaction,
	whenever they are well-typed.
\end{proof}

The progress result states that unless a deadlocked occurs
due to shared names, a network term is always able to progress.
Indeed, due to the recovery semantics sessions can progress
even if they are interleaved. If we combine the progress
result with the result for \ref{thm:type_safety} we can also
deduce that a network progresses in a session safe manner.

The progress result, however, does not indicate the nature of
progress within a session, nor it provides additional insight
on the nature of the recovery semantics in the progress within
a session. 
We state two results on session progress and on session recovery.
We first establish the notion of a simple network,
following the notion of the simple process in~\cite{DBLP:conf/popl/HondaYC08}.

\begin{defi}
	\label{def:simple_network}
	A network $N$ is called {\em simple} if
	it is well-typed with a type derivation
	where the linear environment in the premise and the
	conclusion for each prefix rule in Figure~\ref{fig:process_typing}
	has at most one element.	
\end{defi}

Intuitevely, a simple network is a network whose composing
processes do not implement more than one active session name
during their interaction, e.g.~session interaction does not
interleave. For example, network
$
	N \equiv \net{\accept{a}{x} \pout{s}{\heartbeat} \pout{x}{\heartbeat} \nil \pll \squeue{s}{0}{\ebuffer}}  
$
is not simple because session name $s$ and session variable $x$
(which will be substistuted with a fresh session at runtime) will be
both active at runtime. 

The next theorem states a form of session fidelity where
a session within a simple network is always able to progress
without observing a recovery interaction.
\begin{thm}[Session Progress]
	\label{thm:session_progress}
	Let $N$ be a simple network such that $\Gamma; \Delta \types N$
	for some $\Gamma$ and some well-formed $\Delta$.
	If
	\[
		\textstyle
		N \equiv
		\newnp{\tilde n}{\net{P \pll \squeue{\aggr s}{c}{\tilde{m}} \pll B} \npll \Par{i}{I}{\net{P_i \pll \squeue{s}{c}{\tilde{m}_i} \pll B_i}} \npll M}
	\]
	with i) $\aggr s \in \fs{P}$,
	ii) for all $i \in I$ it holds that $s \in \fs{P_i}$,
	and iii) $s \notin \fn{M}$,
	then there exists $N'$ such that $N \reduces^* N'$ with
	\begin{itemize}
		\item	$N' \equiv \newnp{\tilde n'}{\net{P' \pll \squeue{\aggr s}{c + 1}{\tilde{m}'} \pll B'} \npll \Par{j}{J}{\net{P_j \pll \squeue{s}{c + 1}{\tilde{m}_j} \pll B_j}} \npll M'}$;
		\item	$\reduces^*$ does not involve reduction rules \Recover, \BRecover, and \Loss;
		\item	$J \subseteq I$.
	\end{itemize}
\end{thm}

\begin{proof}
	The proof comes from the fact that
	$\Gamma; \Delta \types N$ and $\Delta$ well-formed.
	The typing judgement implies that sub-network:
	\[
		\textstyle
		N' \equiv
		\newnp{\tilde n}{\net{P \pll \squeue{\aggr s}{c}{\tilde{m}} \pll B} \npll \Par{i}{I}{\net{P_i \pll \squeue{s}{c}{\tilde{m}_i} \pll B_i}}}
	\]
	is also typed with a well-typed linear context and
	from the fact that it is a simple network we can show
	that it can perform the desired reduction using
	either:
	i) \Broadcast reduction possibly followed by several
	\Receive reductions;
	ii) or a \Select reduction possible followed
	by several \Branch reductions;
	iii) or several \Unicast reductions
	followed by a \Gather reduction. 
\end{proof}

The next theorem demonstrates the ability of $s$-endpoints
that are not synchronised with the corresponding $\aggr s$-endpoint
to recover, given that the entire network is a simple network.
\begin{thm}[Session Recovery]
	\label{thm:session_recovery}
	Let $N$ be a simple network 
	such that $\Gamma; \Delta \types N$ for some
	$\Gamma$ and some well-formed $\Delta$.
	If
	\[
		\textstyle
		N \equiv
		\newnp{\tilde n}{\net{P \pll \squeue{\aggr s}{c}{\tilde{m}} \pll B} \npll \Par{i}{I}{\net{P_i \pll \squeue{s}{c_i}{\tilde{m}_i} \pll B_i}} \npll M}
	\]
	and for all $i \in I, c_i < c$,
	then there exists $N'$ such that $N \reduces^* N'$ with
	\[
		N' \equiv
		\newnp{\tilde n}{\net{P \pll \squeue{\aggr s}{c}{\tilde{m}} \pll B} \npll \Par{j}{J}{\net{P_j' \pll \squeue{s}{c}{\tilde{m}_j'} \pll B_j'}} \npll \Par{k}{K}{\net{P_k' \pll B_k'}} \npll M}
	\]
	with $J, K \subseteq I$ and $J \cap K = \emptyset$.
\end{thm}

\begin{proof}
	The proof is done by double induction;
	the first is an induction on the size of $I$
	and the second is an induction on size $c - c_i$.
	The fact that $N$ is simple allows us to observe
	recovery rules \Recover, \BRecover and \Loss on
	the session prefixes of processes $P_i$
	that will increase the counter
	of the $s$-endpoint until it reaches the value of counter $c$.
	However, in the case of a \Cond process term
	a $s$-endpoint might be dropped resulting in
	a partition of the $I$ set into sets $J$ and $K$.
\end{proof}




\section{The Paxos Consensus Protocol}
\label{sec:paxos}

We present a session safe implementation of
the Paxos consensus protocol using the \UBSC.
Paxos~\cite{lamport1998paxos,lamport2001paxos}
is a protocol for reaching consensus in a network that operates
under conditions of unreliability.
It ensures that network agents can agree on a single value in the presence
of failures. Despite being the standard consensus algorithm, Paxos is
notoriously difficult to understand, and real-world implementations
have brought forth many problems that are not taken into account by
the theoretical model of Paxos~\cite{chandra2007paxos}.

The session type representation of Paxos is not used to prove the
correctness of the protocol itself, but to check whether the Paxos
protocol enjoys session types properties, i.e., type preservation, type safety, and progress,
under the particular session specification.
Session types can also help to identify
subtle interactions such as branching or dropping sessions.
Furthermore, a session type representation allows for the basic algorithm
to be easily extended while still providing formal guarantees.


\subsection{The Paxos Protocol}
We implement the most basic protocol of the Paxos family as described
in~\cite{lamport2001paxos}. The network agents act autonomously and propose
values for consensus to the other agents within the network.
If eventually a majority of agents run for long enough without failing,
consensus on one of the proposed values is guaranteed.
A correct implementation of the protocol ensures that:
\begin{itemize}
	\item	Only a value that has been proposed for consensus is chosen.
	\item	The agents within the network agree on a single value.
	\item	An agent is never informed that a value is chosen for consensus,
		unless it has been chosen for consensus.
\end{itemize}
The Paxos setting assumes asynchronous non-Byzantine communication that
operates under the following assumptions:
\begin{itemize}
	\item	Messages can take arbitrarily long to be delivered,
		can be duplicated, but are not delivered corrupted.
	\item	Agents operate at an arbitrary speed, may stop operating
		and may restart. However, it is assumed that
		agents maintain persistent storage that survives
		crashes.  
\end{itemize}
Following the Paxos setting assumptions, an implementation
of the Paxos protocol in \UBSC
is ensured to be correct under the assumptions A1-A6.
%
%
Specifically:
requirement S2 ensures asynchronous communication;
requirement A4 ensures non-Byzantine communication;
requirement A5 ensures non-corrupted delivered messages;
requirement A6 ensures that agents operate at an arbitrary speed;
and
requirement A3 ensures that messages are never duplicated,
which is subsumed by the Paxos requirement that messages might
be duplicated, i.e., the properties of the Paxos protocol are
also ensured in a setting where
messages are never duplicated.
Moreover, requirement A6 can express the case where an agent
has failed by allowing the agent to take an arbitrary long
time to perform an action. If the agent eventually interacts
then it is considered to have restarted maintaining, through
persistent storage, all the information it had prior to the fail.

Paxos agents implement three roles:
i) a \emph{proposer} agent proposes values towards the network for
reaching consensus;
ii) an \emph{acceptor} accepts a value from those proposed, whereas
a majority of acceptors accepting the same value implies consensus
and signifies protocol termination;
and
iii) a \emph{learner} discovers the chosen consensus value.
The implementation of the protocol may proceed over several rounds.
A successful round has two phases: \PreparePh and \AcceptPh.




The protocol ensures that  in the case where a consensus value $v$
has already been chosen among the majority of the network agents,
broadcasting a new proposal request with a higher proposal number
will result in choosing the already chosen consensus value $v$. 
Following this fact, we assume for simplicity that a learner has
the same implementation as a proposer.


\subsection{Implementation of the Paxos protocol in \UBSC}

\begin{figure}
	\[
		\arraycolsep=2pt
		\begin{array}{rcl}
			\PaxosType	&=&	\tout{\prep} \tin{\prom}
						\oplus \left\{
						\begin{array}{rl}
							\acceptLabel: & \tout{\messageTuple} \tend,
							\\
							\restart: & \tend
						\end{array}
						\right\}
			\\[6mm]

			\PaxosNode^{\id}_{r, v}	&=&		\net{\Paxos^{\id}_{r, v}}
			\\[2mm]
			\Paxos^{\id}_{r, v}		&=&		\Def{
						\\		&& \quad
								\begin{array}{rcl}
									\ProposerDef{x, y}
									&\defeq&
										\request{a}{s}
										\pout{\aggr s}{x}
										\pin{\aggr s}{\set{(r_i, v_i)}_{i \in I}}
									\\
									&&	\If\ |I| > \frac{M}{2}\ \Then
										\\ && \quad
										\sel{\aggr s}{\acceptLabel}
										\pout{\aggr s}{(x, v = \chooseVal{\set{(r_i, v_i)}_{i \in I}, \mathsf{id}})}
										\\ && \quad \PaxosVar{x, v}
									\\
									&&	\Else
										\\ && \quad
										\sel{\aggr s}{\restart} \PaxosVar{x, y}
									\\[2mm]
									\AcceptorDef{x, y}
									&\defeq&
										\accept{a}{s}
										\pin{s}{x'}
										\\ &&
										\If\ (x' > x)\ \Then
										\\ & & \quad \pout{s}{x, y} s \triangleright
										\left\{
											\begin{array}{rl}
												\acceptLabel:	& \pin{s}{x', y'} \PaxosVar{x', y'},
												\\
												\restart:	& \PaxosVar{x, y}
											\end{array}
										\right\}
										\\&& \Else
										\\&& \quad \PaxosVar{x, y}

									\\[4mm]
									\PaxosDef{x, y}
									&\defeq&
										\ProposerVar{x+1, y} \Sum (\AcceptorVar{x, y} \recover \PaxosDef{x, y})
								\end{array}
						\\		&& }{
						\\		&& \quad \PaxosVar{r, v}}
		\end{array}
	\]
	\caption{Implementation of the Paxos consensus protocol \label{fig:paxos}}
\end{figure}

Figure~\ref{fig:paxos} describes the implementation of the Paxos
protocol in our framework. The implementation assumes that
expressions contain finite sets of integer tuples.
We use notation $\set{(r_i, v_i)_i}_{i \in I}$ for such a
finite set.
Moreover, we assume that the aggregation operation is defined
as a union of sets of integer tuples.

The interaction for establishing a consensus value takes place within
a single session that involves the nodes of the network.
The communication behaviour of the Acceptor is described by session type
\PaxosType, whereas the communication behaviour of the proposer is described
by the dual type $\tdual{\PaxosType}$.
The \PaxosType session type provides with a type level description of the
protocol in \cite{lamport2001paxos}.
The fact that the type is enforced to the implementation, through the
type system, together with
the fact that the underlying assumptions of our calculus are covered by the
assumptions of the Paxos setting, provide supporting evidence for
proving the correctness of the implementation.

A Paxos agent is described by network node
$\PaxosNode^{\id}_{r, v} = \net{\Paxos^\id_{r, v}}$ where $r$
is the number of the current proposal number, $v$ is the
consensus value that corresponds to proposal number $r$,
and \id is a unique node identity number.
A Paxos agent non-deterministically behaves either as a proposer,
(definition \ProposerDef{x, y}) or as an acceptor
(definition $\AcceptorDef{x, y} \recover \PaxosDef{x, y}$).
The definition of the acceptor requires that during computation
an acceptor agent recovers when an input endpoint does not
progress, by dropping all active sessions and proceeding to 
process \PaxosVar{r, v}.


%
 
If a Paxos agent decides to act as a proposer, it does so by increasing
its current proposal number and proceeding to process $\ProposerVar{r+1, v}$.
It then requests a new session and enters the \PreparePh phase.
All the Paxos agents that accept a session request act as acceptors.
The proposer then broadcasts towards the network a \textit{prepare}
message request, type \prep, that contains the proposal number $r$.

All the acceptors that received the \textit{prepare} message check
whether the proposal number is greater than the one they currently
have. If not, they drop the session and restart the computation
proceeding to process $\PaxosVar{r, v}$. Otherwise, they reply with
a \textit{promise} message, type \prom, not to respond to a prepare
message with a lower round number.
The promise message contains the current proposal number and the
current consensus value of the acceptor. If this is the first time
the acceptor is involved in a consensus round and it has no
information of consensus value then the promise message will
contain empty values $(\emptyval, \emptyval)$.
Here assume that for all proposal numbers $r$ it holds that $r > \emptyval$.


After all the involved acceptors reply with a promise message,
the protocol enters the \AcceptPh phase.
The proposer gathers all the promises as a set of promise messages,
$\set{(r_i, v_i)}_{i \in I}$, and then checks whether the majority
of acceptors have replied using condition $|I| \leq \frac{M}{2}$,
with $M$ being the number of the nodes in the network.
Note that for clarity, in the Proposer agent we use the set notation
$\set{(r_i, v_i)}_{i \in I}$ in place of variable in an input process.

If the check fails the proposer sends a restart label to all the
acceptors, and restarts its own computation by proceeding to process
$\PaxosVar{r, v}$.
All acceptors that receive label \restart also restart their
computation by proceeding to process $\PaxosVar{r, v}$.
%

If the majority check is passed, the proposer selects a value to submit
to the acceptors by inspecting the promises received and choosing the
value corresponding to the highest proposal number received in a promise
message.
If no value is received then it chooses its \id value.
This is expressed by computation
\[
	\begin{array}{rcl}
		v_k &=& \chooseVal{\set{(r_i, v_i)}_{i \in I}, \id} \text{ when } \forall i \in I, r_k \geq r_i
		\qquad \text{and}
		\\
		\id &=& \chooseVal{\set{(r_i, v_i)}_{i \in I}, \id} \text{ when } \forall i \in I, r_i = \emptyval
	\end{array}
\]
%
%
The proposer then broadcasts an \textit{accept} message, 
which is an \acceptLabel label followed by a tuple that contains
the current proposal number and the chosen highest value. Finally,
it updates its own proposal number and consensus value 
and proceed to process $\PaxosVar{r, v}$.
%
Moreover, all acceptors that receive the \textit{accept} message
update their proposal number and their consensus value, and proceed
to state $\PaxosVar{r, v}$.

%

Network node $\PaxosNode_{n, v}$ can be typed using the following
typing judgement:
\begin{eqnarray*}
	a: \PaxosType; \econtext \types \PaxosNode^{\id}_{n, v}
\end{eqnarray*}
Shared channel $a$ uses type \PaxosType, thus all establish
sessions of the computation follow the behaviour of the
Paxos protocol as described by the \PaxosType session type.
Subsequently, a network that describes a set of nodes
that run the Paxos protocol is defined as:
\[
	N = \Par{i}{I}{\PaxosNode^{i}_{\emptyval, \emptyval}}
\]
Network $N$ is typed using typing judgement
$a: \PaxosType; \econtext \types N$.
Typing is possible due to the typing of network node $\PaxosNode^{\id}_{n, v}$
and multiple applications of rule \TPar.

The paxos process $\Paxos^\id_{r, v}$ is recursive, where each
recursion indicates a new propose round that implies the creation
of a new session of type \PaxosType.
Each iteration creates a new session by necessity,
since if, due to reliability, a session within a propose
round does not include the majority of Paxos acceptors
then it is necessary to terminate the session and reiterate.
This is reflected by the fact that
type \PaxosType is not recursive.

Our typing framework guaranties that network $N$,
which is typed with a well-formed (empty) linear context,
will never reduce to an error state
(Theorem~\ref{thm:type-preservation}). Moreover, the
fact that a Paxos network is a simple network (Definition~\ref{def:simple_network})
and the progress results (Theorems~\ref{thm:session_progress} and~\ref{thm:session_progress}) 
ensure that the interaction takes place within
a session as defined by the session type.
This lifts the burden from the programmer to check
for deadlocks and type mismatches, and leaves only
the burden for implementing correctly the algorithmic
logic. If the algorithmic logic is implemented 
correctly, then the fact that the assumptions
of the \UBSC are covered by the execution context of Paxos
ensures the properties of the Paxos protocol.

More complicated Paxos networks can be described.
For example, we can allow a Paxos agent that acts as an acceptor
to establish multiple sessions with different proposers
during an execution and explicitly drop the session
with the lowest proposal number:
\[
\begin{array}{rcl}
	\AcceptorDef{x, y}
	&\defeq&
		\accept{a}{s}
		\pin{s}{x'} \AccVar{s, x, x', y} 
	\\[2mm]
	\AccDef{w, x, x', y} &\defeq&
		\If\ (x' > x)\ \Then
		\\&& \quad \pout{w}{x, y}
		\\&& \qquad
		\left(
		\begin{array}{ll}
			&  \AccPhVar{w, x, y}
			\\ \Sum & \accept{a}{s'} \pin{s'}{x''}
			\\ & \If\ (x'' > x')
			\\ & \quad \AccVar{s', x, x'', y}
			\\ & \Else
			\\ & \quad \AccPhVar{w, x, y}
		\end{array}
		\right)
		\\&& \Else
		\\&& \quad \PaxosVar{x, y}

	\\[2mm]
	\AccPhDef{w, x, y} &\defeq&
		w \triangleright
			\left\{
				\begin{array}{rl}
					\acceptLabel:	& \pin{w}{x', y'} \PaxosVar{x', y'},
					\\
					\restart:	& \PaxosVar{x, y}
				\end{array}
			\right\}

\end{array}	
\]
The definition of the acceptor has the option, expressed
as non-deterministic choice, to establish a
new session, i.e.~enter a prepare phase with a different
proposer, while in the accept phase of another proposal.
The computation then compares the proposal numbers from the
two sessions. Following the result of the comparison, the
computation will explicitly drop the session with the lowest
proposal number and proceed as described by the type of the
session with the highest proposal number. The above process
is well-typed; note that the application of rule \TSum
ensures that the two branches of the non-deterministic operator
have the same session type.

The new acceptor definition results in a network that
is not simple following Definition~\ref{def:simple_network}, therefore
the conditions for Session Progress (Theorem~\ref{thm:session_progress})
and Session Recovery (Theorem~\ref{thm:session_recovery}) do not hold.
Nevertheless, the network will never reduce to an error network (Definition~\ref{def:errornetwork})
and it can be proven that it still enjoys the progress property via Theorem~\ref{thm:progress}.
Moreover we can show that individual sessions can still progress and recover.

\section{Conclusion and Future Work}
\label{sec:conclusion}

%



This paper is motivated by the need to to sufficiently define
a session type framework for systems such as ad-hoc and
wireless sensor networks. The basic characteristic of these
networks is a lossy stateless communication medium shared among
an arbitrary number of agents.

To this end, we have introduce the asynchronous unreliable broadcast
calculus accompanied with a corresponding session type framework.
The semantics of the calculus are inspired by the practice
of ad-hoc and wireless sensor network and develop the necessary
mechanisms to support safe session interaction and recovery.
Asynchrony is achieved with the use of message buffers.
The main session communication operations include asynchronous
broadcast operation and the asynchronous gather communication
pattern in the presence of link failure and message loss.
Message loss may lead to session endpoints that are not synchronised
with the overall protocol. In such a case the semantics propose
the appropriate mechanisms for autonomous session recovery.
The recovery semantics find sufficient justification from the 
practice of networks that operate in an unreliable setting.

Our session type systems ensures safe and sound communication interaction
that respects the session types principles. A type preservation
theorem ensures the soundness of the system, whereas a type safety
theorem ensures that the system will never reduce an undesired/error
state. Finally, a set of progress properties ensure that 
sessions may always progress, and, in the case of non-synchronisation, 
a session can always safely recover.

%

%

The syntax and semantics of the \UBSC are expressive enough to
describe non-trivial
communication protocols that operate in an unreliable setting.
%
We used the \UBSC to describe a basic implementation of the
Paxos consensus algorithm, the standard protocol for achieving
consensus in an asynchronous unreliable setting.
Our framework satisfies the underlying operation
assumptions of the Paxos protocol. The computation
for reaching consensus takes place within a single session, therefore 
a single session type can describe the interaction of the Paxos agents.
The session typing system ensures that a Paxos agent
interaction follows the Paxos session protocol.
%

Moreover, this work follows the syntax of
binary session types, where a rather intrigate typing
system enforces the static properties of soundness,
safety, and deadlock freedom to a calculus that supports
realistic interactions of broadcast, gather, and recovery.
However, it is not in the aim (nor in the static nature)
of the \UBSC session type system to guaranty more dynamic
properties that have to do with the algorithm correctness,
e.g., to prove that the Paxos protocol will
eventually reach consensus.


\subsection*{Future Work}

Our system is open for further extensions in the future.
A first direction is to capture more complicated patterns beyond the
unreliable broadcast communication and the gather communication pattern.
We would also like to investigate more elaborate autonomous
recovery mechanisms, e.g.~to explicitly define time-out
and interrupt routines.

Another interesting direction is the extension of
the system with locations and mobility, cf.~\cite{Hennessy2002}, as this 
would allow the use of session types in a highly dynamic and complex
setting.
Location and node mobility is an important feature which allows for new nodes to
be introduced, disabled, or migrate between locations within a network.
Locality and/or mobility might also imply semantics that take into account
broadcast/communication range. The feature of communication range further
strengthens the need for the existence of autonomous recovery mechanisms.
%

%

An important direction for future research is to develop a
more robust session type system, based on multiparty session types, where network
agents can interact with multiple roles inside a session.
The research direction on a multiparty session type framework in the context
of unreliable broadcast communication can follow principles from the
work on parametrised multiparty session types~\cite{DBLP:journals/soca/NgY15},
where a session type system is developed for describing interaction among an
arbitrary number of agents. 

\subsection*{Acknowledgements}

    Kouzapas and Gay were supported by the {UK EPSRC} project {EP/K034413/1}
    ``From Data Types to Session Types: A Basis for Concurrency and
    Distribution''. Voinea was supported by an EPSRC PhD studentship. This research was supported by a
    Short-Term Scientific Mission grant from COST Action IC1201
    (Behavioural Types for Reliable Large-Scale Software Systems).
    Voinea and Gay were supported by the UK EPSRC project "Session Types for Reliable Distributed Systems" (EP/T014628/1 and EP/T014512/1) and by BehAPI (Behavioural Application Program Interfaces), an EU H2020 RISE programme (Marie Skłodowska-Curie grant agreement No 778233).


	\bibliographystyle{alphaurl}
	\bibliography{references}

\newcommand{\etalchar}[1]{$^{#1}$}
\begin{thebibliography}{BDGK14}

\bibitem[ABB{\etalchar{+}}16]{BETTYWG3}
Davide Ancona, Viviana Bono, Mario Bravetti, Joana Campos, Giuseppe Castagna, Pierre-Malo Deniélou, Simon~J. Gay, Nils Gesbert, Elena Giachino, Raymond Hu, Einar~Broch Johnsen, Francisco Martins, Viviana Mascardi, Fabrizio Montesi, Rumyana Neykova, Nicholas Ng, Luca Padovani, Vasco~T. Vasconcelos, and Nobuko Yoshida.
\newblock Behavioral types in programming languages.
\newblock {\em Foundations and Trends® in Programming Languages}, 3:95--230, 2016.

\bibitem[APN17]{10.1007/978-3-319-60225-7_1}
Manuel Adameit, Kirstin Peters, and Uwe Nestmann.
\newblock Session types for link failures.
\newblock In Ahmed Bouajjani and Alexandra Silva, editors, {\em Formal Techniques for Distributed Objects, Components, and Systems}, pages 1--16, Cham, 2017. Springer International Publishing.

\bibitem[BD23]{magpi}
Matthew Alan~Le Brun and Ornela Dardha.
\newblock Mag{\(\pi\)}: Types for failure-prone communication.
\newblock In Thomas Wies, editor, {\em Programming Languages and Systems - 32nd European Symposium on Programming, {ESOP} 2023, Held as Part of the European Joint Conferences on Theory and Practice of Software, {ETAPS} 2023, Paris, France, April 22-27, 2023, Proceedings}, volume 13990 of {\em Lecture Notes in Computer Science}, pages 363--391. Springer, 2023.
\newblock \href {https://doi.org/10.1007/978-3-031-30044-8\_14} {\path{doi:10.1007/978-3-031-30044-8\_14}}.

\bibitem[BDGK14]{BDGK14}
Giovanni Bernardi, Ornela Dardha, Simon~J. Gay, and Dimitrios Kouzapas.
\newblock On duality relations for session types.
\newblock In Matteo Maffei and Emilio Tuosto, editors, {\em Trustworthy Global Computing}, pages 51--66, Berlin, Heidelberg, 2014. Springer Berlin Heidelberg.

\bibitem[BHJ{\etalchar{+}}11]{Borgstrom2011}
Johannes Borgstr{\"{o}}m, Shuqin Huang, Magnus Johansson, Palle Raabjerg, Bj{\"{o}}rn Victor, Johannes~{\AA}man Pohjola, and Joachim Parrow.
\newblock Broadcast psi-calculi with an application to wireless protocols.
\newblock In {\em {SEFM} 2011}, pages 74--89, 2011.

\bibitem[CDP12]{DBLP:journals/corr/abs-1203-0780}
Giuseppe Castagna, Mariangiola Dezani{-}Ciancaglini, and Luca Padovani.
\newblock On global types and multi-party session.
\newblock {\em Logical Methods in Computer Science}, 8(1), 2012.

\bibitem[CDYP16]{DBLP:journals/mscs/CoppoDYP16}
Mario Coppo, Mariangiola Dezani{-}Ciancaglini, Nobuko Yoshida, and Luca Padovani.
\newblock Global progress for dynamically interleaved multiparty sessions.
\newblock {\em Mathematical Structures in Computer Science}, 26(2):238--302, 2016.

\bibitem[CGR07]{chandra2007paxos}
Tushar~D Chandra, Robert Griesemer, and Joshua Redstone.
\newblock Paxos made live: an engineering perspective.
\newblock In {\em Proceedings of the twenty-sixth annual ACM symposium on Principles of distributed computing}, pages 398--407. ACM, 2007.

\bibitem[CGY16]{CGY2014}
Sara Capecchi, Elena Giachino, and Nobuko Yoshida.
\newblock Global escape in multiparty sessions.
\newblock {\em Mathematical Structures in Computer Science}, 26, 2 2016.

\bibitem[CHY08]{CarboneHY08}
Marco Carbone, Kohei Honda, and Nobuko Yoshida.
\newblock Structured interactional exceptions in session types.
\newblock In {\em CONCUR}, volume 5201 of {\em LNCS}, pages 402--417. Springer, 2008.

\bibitem[DOH07]{Dunkels2007}
Adam Dunkels, Fredrik \"{O}sterlind, and Zhitao He.
\newblock An adaptive communication architecture for wireless sensor networks.
\newblock In {\em Proceedings of Embedded Networked Sensor Systems}, SenSys '07, pages 335--349. ACM, 2007.

\bibitem[DYBH12]{DBLP:journals/corr/abs-1208-6483}
Pierre{-}Malo Deni{\'{e}}lou, Nobuko Yoshida, Andi Bejleri, and Raymond Hu.
\newblock Parameterised multiparty session types.
\newblock {\em Logical Methods in Computer Science}, 8(4), 2012.

\bibitem[GR17]{BETTY:tools}
Simon~J. Gay and Ant\'{o}nio Ravara, editors.
\newblock {\em Behavioural Types: from Theory to Tools}.
\newblock River Publishers, 2017.

\bibitem[GV10]{DBLP:journals/jfp/GayV10}
Simon~J. Gay and Vasco~Thudichum Vasconcelos.
\newblock Linear type theory for asynchronous session types.
\newblock {\em J. Funct. Program.}, 20(1):19--50, 2010.
\newblock \href {https://doi.org/10.1017/S0956796809990268} {\path{doi:10.1017/S0956796809990268}}.

\bibitem[HKP{\etalchar{+}}10]{DBLP:conf/ecoop/HuKPYH10}
Raymond Hu, Dimitrios Kouzapas, Olivier Pernet, Nobuko Yoshida, and Kohei Honda.
\newblock Type-safe eventful sessions in java.
\newblock In {\em {ECOOP}}, pages 329--353, 2010.

\bibitem[HR02]{Hennessy2002}
Matthew Hennessy and James Riely.
\newblock Resource access control in systems of mobile agents.
\newblock {\em Information and Computation}, 173(1):82 -- 120, 2002.

\bibitem[HVK98]{Honda1998}
Kohei Honda, Vasco~Thudichum Vasconcelos, and Makoto Kubo.
\newblock Language primitives and type discipline for structured communication-based programming.
\newblock In {\em ESOP}, pages 122--138. Springer-Verlag, 1998.

\bibitem[HYC08]{DBLP:conf/popl/HondaYC08}
Kohei Honda, Nobuko Yoshida, and Marco Carbone.
\newblock Multiparty asynchronous session types.
\newblock In {\em {POPL}}, pages 273--284, 2008.

\bibitem[KGG14]{Kouzapas2014}
Dimitrios Kouzapas, Ramunas Gutkovas, and Simon~J. Gay.
\newblock Session types for broadcasting.
\newblock In {\em {PLACES}}, pages 25--31, 2014.

\bibitem[Kou12]{dkphdthesis}
Dimitrios Kouzapas.
\newblock {\em A Study of Bisimulation Theory for Session Types}.
\newblock PhD thesis, Imperial College London, 2012.

\bibitem[KYHH16]{DBLP:journals/mscs/KouzapasYHH16}
Dimitrios Kouzapas, Nobuko Yoshida, Raymond Hu, and Kohei Honda.
\newblock On asynchronous eventful session semantics.
\newblock {\em Mathematical Structures in Computer Science}, 26(2):303--364, 2016.
\newblock \href {https://doi.org/10.1017/S096012951400019X} {\path{doi:10.1017/S096012951400019X}}.

\bibitem[L{\etalchar{+}}01]{lamport2001paxos}
Leslie Lamport et~al.
\newblock Paxos made simple.
\newblock {\em ACM Sigact News}, 32(4):18--25, 2001.

\bibitem[Lam98]{lamport1998paxos}
Leslie Lamport.
\newblock The part-time parliament.
\newblock {\em ACM Trans. Comput. Syst.}, 16(2):133--169, 1998.

\bibitem[NY15]{DBLP:journals/soca/NgY15}
Nicholas Ng and Nobuko Yoshida.
\newblock Pabble: parameterised scribble.
\newblock {\em Service Oriented Computing and Applications}, 9(3-4):269--284, 2015.
\newblock \href {https://doi.org/10.1007/s11761-014-0172-8} {\path{doi:10.1007/s11761-014-0172-8}}.

\bibitem[PNC23]{ftmpst}
Kirstin Peters, Uwe Nestmann, and Wagner Christoph.
\newblock Ftmpst: Fault-tolerant multiparty session types, 2023.
\newblock https://arxiv.org/pdf/2204.07728.pdf.

\bibitem[YV07]{YOSHIDA200773}
Nobuko Yoshida and Vasco~T. Vasconcelos.
\newblock Language primitives and type discipline for structured communication-based programming revisited: Two systems for higher-order session communication.
\newblock {\em Electronic Notes in Theoretical Computer Science}, 171(4):73 -- 93, 2007.

\end{thebibliography}

	\newpage
	\appendix


\section{Proofs}

\subsection{Congruence invariance proof}

\begin{lem}[Congruence Invariance]
	\label{lem:congruence-invariance}
	\begin{itemize}
		\item	If	$P \equiv P'$, then
				$\Gamma;\Delta \types P$
			if, and only if,
				$\Gamma;\Delta \types P'$.

		\item	If	$N \equiv N'$, then
				$\Gamma;\Delta \types N$
			if, and only if,
				$\Gamma;\Delta \types N'$.
	\end{itemize}
\end{lem}
\begin{proof}
    By induction on the structural congruence definition.
\end{proof}

\begin{lem}[Substitution]
	\label{lem:substitution}
	\begin{itemize}
		\item	Whenever
				$\Gamma \types e : \beta$
			and
				$\Gamma, x: \beta; \Delta \types P$,
			then
				$\Gamma; \Delta \types P\subst{e}{x}$.

		\item	Whenever
				$\Gamma; \Delta, x: T \types P$,
			then
				$\Gamma; \Delta, s: T \types P\subst{s}{x}$.
	\end{itemize}
\end{lem}
\begin{proof}
	By induction on the definition of processes.
\end{proof}

\subsection{Typing Preservation proof}
\label{app:type_preservation}

Now we are ready to prove the typing preservation property of our system.

\begin{proof}[Proof of Theorem~\ref{thm:type-preservation}]
    By induction on the depth of derivation of $N \reduces N'$.

	\begin{itemize}
		\item	Rule \Connect

			\noindent
			Assume
			\[
				\begin{array}{l}
					\tree{
						s \text{ fresh}
					}{
						\net{\request{a}{\aggr x} P \pll B} \,\pll\,
						\Par{i}{I}{\net{\accept{a}{x} Q_i \pll B_i}}
						\\
						\qquad
						\reduces 
						\newnp{s}{\net{P \subst{\aggr s}{\aggr x} \pll B \pll \squeue{\aggr s}{0}{\equeue}}
						\npll
						\Par{i}{I}{\net{Q_i \subst{s}{x} \pll B_i \pll \squeue{s}{0}{\equeue}}}}
					}{}
				\end{array}
			\]
			and
			\[
				\Gamma; \Delta \types \net{\request{a}{\aggr x} P \pll B} \npll \Par{i}{I}{\net{\accept{a}{x} Q_i \pll B_i}}
			\]
			with $\Delta$ well-formed.
			From the latter typing judgement we get the derivation
			\[
				\tree {
					\Gamma; \Delta_0, \Delta_1 \types \net{\request{a}{\aggr x} P \pll B}
					\\
					\Gamma; \Delta_0, \Delta_2 \types \Par{i}{I}{\net{\accept{a}{x} Q_i \pll B_i}}
					\\
					\text{$\Delta_0$ only $s$-endpoints}
					\andalso
					\Delta = \Delta_0, \Delta_1, \Delta_2 
				}{
					\Gamma; \Delta \types \net{\request{a}{\aggr x} P \pll B} \npll \Par{i}{I}{\net{\accept{a}{x} Q_i \pll B_i}}
				}{\TPar}
			\]
			We continue the above derivation for node
			$\net{\request{a}{\aggr x} P \pll B}$
			\[
				\arraycolsep=3pt
				\btree {
					\begin{array}{ll}
						\btree {
							\tree {
								\Gamma; \Delta_1'', \aggr x: \tdual{T} \types P
								\andalso
								\Gamma \types a: T
							}{
								\Gamma; \Delta_1'' \types \request{a}{\aggr x} P
							}{\TReq}
							\\
							\Gamma; \Theta \types B
							\andalso
							\Delta_1' = \Delta_1'' \sessionop \Theta
						}{
							\Gamma; \Delta_1' \types \net{\request{a}{s} P \pll B}
						}{\TNode}
					&
						\synchronise{\Delta_1'}{\Delta_0, \Delta_1}
					\end{array}
					\\[12mm]
				}{
					\Gamma; \Delta_0, \Delta_1 \types \net{\request{a}{s} P \pll B}				
				}{\TSynch}
			\]
			For every node $\net{\accept{a}{x} Q_i \pll B_i}$
			we can derive
			\[
				\btree {
					\begin{array}{ll}
						\btree {
							\tree {
								\Gamma; \Delta_i'', x: T \types Q_i
								\andalso
								\Gamma \types a: T
							}{
								\Gamma; \Delta_i'' \types \accept{a}{x} Q_i
							}{\TAcc}
							\\
							\Gamma; \Theta_i \types B_i
							\andalso
							\Delta_i' = \Delta_i'' \sessionop \Theta_i
						}{
							\Gamma; \Delta_i' \types \net{\accept{a}{x} Q_i \pll B_i}
						}{\TNode}
						&
						\synchronise{\Delta_i'}{\Delta_0, \Delta_i}
					\end{array}
					\\[12mm]
				}{
						\Gamma; \Delta_0, \Delta_i \types \net{\accept{a}{x} Q_i \pll B_i}
				}{\TSynch}
			\]
			and also with multiple applications of the \TPar rule we get
			\[
				\tree {
					\set{\Gamma; \Delta_0, \Delta_i \types \net{\accept{a}{x} Q_i \pll B_i}}_{i \in I}
					\andalso
					\Delta_2 = \bigcup_{i \in I} \Delta_i
				}{
					\Gamma; \Delta_0, \Delta_2 \types \Par{i}{I}{\net{\accept{a}{x} Q_i \pll B_i}}
				}{\TPar}
			\]
			We can then combine the information from the
			latter transition to derive the type for the
			result of the reduction
			$\newnp{s}{\net{P \subst{\aggr s}{\aggr x} \pll B \pll \squeue{\aggr s}{0}{\equeue}}
			\npll
			\Par{i}{I}{\net{Q_i \subst{s}{x} \pll B_i \pll \squeue{s}{0}{\equeue}}}}$.
			We first type network node
			$\net{P \subst{\aggr s}{\aggr x} \pll B \pll \squeue{\aggr s}{0}{\equeue}}$:
			\[
				\btree {
					\btree {
						\btree {
							\Gamma; \aggr s: \ctype{0}{\tempty} \types \squeue{\aggr s}{0}{\equeue}
							\quad \SEmp
							\\
							\Gamma; \Theta \types B \text{ (from former derivation)}
						}{
							\Gamma; \Theta, \aggr s: \ctype{0}{\tempty} \types B \pll \squeue{\aggr s}{\equeue}{\equeue}
						}{\BPar}
						\\[9mm]
						\begin{array}{l}
							\text{Substitution Lemma~\ref{lem:substitution} and former derivation}
							\\
							\Gamma; \Delta_1'', \aggr x: \tdual T \types P \text{ imply }
							\Gamma; \Delta_1'', \aggr s: \tdual T \types P \subst{\aggr s}{\aggr x}
						\end{array}
						\\[4mm]
						\Delta_1'', \aggr s: \tdual T \sessionop \Theta, \aggr s: \ctype{0}{\tempty} = \Delta_1', \aggr s: \ctype{0}{\tdual T}
						\\[1mm]
					}{
						\Gamma; \Delta_1', \aggr s: \ctype{0}{\tdual T} \types \net{P \subst{\aggr s}{\aggr x} \pll B \pll \squeue{\aggr s}{0}{\equeue}}
					}{\TNode}
					\\[22mm]
					\begin{array}{l}
						\text{Linear context synchronisation (Definition~\ref{def:synchr}) and former result }
						\\
						\synchronise{\Delta_1'}{\Delta_0, \Delta_1} \text{ imply }
						\synchronise{\Delta_1', s: 0, \aggr s: \tdual T}{\Delta_0, \Delta_1, \aggr s: \ctype{0}{\tdual T}}
					\end{array}
				}{
					\Gamma; \Delta_0, \Delta_1, \aggr s: \ctype{0}{\tdual T} \types \net{P \subst{\aggr s}{\aggr x} \pll B \pll \squeue{\aggr s}{0}{\equeue}}
				}{\TSynch}
			\]
			Similarly we type every network node
			$\net{Q_i \subst{s}{x} \pll B_i \pll \squeue{s}{0}{\equeue}}$
			\[
				\btree[] {
					\btree {
						\btree {
							\Gamma; s: \ctype{0}{\tempty} \types \squeue{s}{0}{\equeue}
							\quad \SEmp
							\\
							\Gamma; \Theta \types B_i \text{ (from former derivation)}
						}{
							\Gamma; \Theta, s: \ctype{0}{\tempty} \types B_i \pll \squeue{s}{\equeue}{\equeue}
						}{\BPar}
						\\[9mm]
						\begin{array}{l}
							\text{Substitution Lemma~\ref{lem:substitution} and former derivation}
							\\
							\Gamma; \Delta_i'', x: T \types Q_i \text{ imply }
							\Gamma; \Delta_i'', s: T \types Q_i \subst{s}{x}
						\end{array}
						\\[4mm]
						\Delta_i'', s: T \sessionop \Theta, s: \ctype{0}{\tempty} = \Delta_i' s: \ctype{0}{T}
						\\[1mm]
					}{
						\Gamma; \Delta_i', s: \ctype{0}{T} \types \net{Q_i \subst{s}{x} \pll B_i \pll \squeue{s}{0}{\equeue}}
					}{\TNode}
					\\[22mm]
					\begin{array}{l}
						\text{Linear context synchronisation (Definition~\ref{def:synchr}) and former result }
						\\
						\synchronise{\Delta_i'}{\Delta_0, \Delta_1} \text{ imply }
						\synchronise{\Delta_i', s: \ctype{0}{T}}{\Delta_0, \Delta_i, s: \ctype{0}{T}}
					\end{array}
				}{
					\Gamma; \Delta_0, \Delta_i, s: \ctype{0}{T} \types \net{Q_i \subst{s}{x} \pll B_i \pll \squeue{s}{0}{\equeue}}
				}{\TSynch}
			\]
			We can now have a multiple application of the \TPar rule to
			type network

			\noindent
			$\Par{i}{I}{\net{Q_i \subst{s}{x} \pll B_i \pll \squeue{s}{0}{\equeue}}}$
			\[
				\tree {
					\Gamma; \Delta_0, \Delta_i, s: \ctype{0}{T} \types \net{Q_i \subst{s}{x} \pll B_i \pll \squeue{s}{0}{\equeue}}
					\\
					\text{ from former result }
					\Delta_2 = \bigcup_{i \in I} \Delta_i				
				}{
					\Gamma; \Delta_0, \Delta_2, s: \ctype{0}{T} \types \Par{i}{I}{\net{Q_i \subst{s}{x} \pll B_i \pll \squeue{s}{0}{\equeue}}}
				}{\TPar}
			\]
			We then use typing rule \TPar followed by rule \TSRes to type the result of the reduction
			$\newnp{s}{\net{P \subst{\aggr s}{\aggr x} \pll B \pll \squeue{\aggr s}{0}{\equeue}}
			\npll
			\Par{i}{I}{\net{Q_i \subst{s}{x} \pll B_i \pll \squeue{s}{0}{\equeue}}}}$:
			\[
				\btree[\TSRes] {
					\btree[\TPar] {
						\Gamma; \Delta_0, \Delta_1, \aggr s: \ctype{0}{\tdual T} \types \net{P \subst{\aggr s}{\aggr x} \pll B \pll \squeue{\aggr s}{0}{\equeue}}
						\\
						\Gamma; \Delta_0, \Delta_2, s: \ctype{0}{T} \types \Par{i}{I}{\net{Q_i \subst{s}{x} \pll B_i \pll \squeue{s}{0}{\equeue}}}
					}{
							\Gamma; \Delta, \aggr s: \ctype{0}{\tdual T}, s: \ctype{0}{T} \types \net{P \subst{\aggr s}{\aggr x} \pll B \pll \squeue{\aggr s}{0}{\equeue}}
							\npll
							\Par{i}{I}{\net{Q_i \subst{s}{x} \pll B_i \pll \squeue{s}{0}{\equeue}}}
					}{}
					\\[12mm]
				}{
					\Gamma; \Delta \types \newnp{s}{\net{P \subst{\aggr s}{\aggr x} \pll B \pll \squeue{\aggr s}{0}{\equeue}}
					\npll
					\Par{i}{I}{\net{Q_i \subst{s}{x} \pll B_i \pll \squeue{s}{0}{\equeue}}}}
				}{}
			\]
			The last result together with the trivial result
			$\Delta \tadvance \Delta$ concludes the case.
%

		\item	Rule \Broadcast

			\noindent
			Assume
			\[
					\net{\pout{\aggr s}{e} P \pll B \pll \squeue{\aggr s}{c}{\equeue}} \npll
					\Par{i}{I}{\net{P_i \pll B_i \pll \squeue{s}{c}{\pol{m}_i}}}
					\reduces
					\net{P \pll B \pll \squeue{s}{c+1}{\equeue}} \npll
					\Par{i}{I}{\net{P_i \pll B_i \pll \squeue{s}{c+1}{\pol{m}_i \cat e} } }
			\]
			and
			\[
				\Gamma; \Delta \types	\net{\pout{\aggr s}{e} P \pll B \pll \squeue{\aggr s}{c}{\equeue}}
							\npll
							\Par{i}{I}{\net{P_i \pll B_i \pll \squeue{s}{c}{\pol{m}_i}}}
			\]
			with $\Delta$ well-formed.
			We proceed with the typing derivation of the latter typing judgement.
			\[
				\tree {
					\begin{array}{l}
						\text{well-formed $\Delta$ implies }
						\Delta = \Delta_0, \Delta_1, \aggr s: \ctype{c}{\tout{\beta} \tdual{T}}, \Delta_2, s: \ctype{c}{\tin{\beta} T}
						\\
						\text{with $\Delta_0$ only $s$-endpoints}
					\end{array}
					\\
					\Gamma; \Delta_0, \Delta_1, \aggr s: \ctype{c}{\tout{\beta} \tdual{T}} \types \net{\pout{\aggr s}{e} P \pll B \pll \squeue{\aggr s}{c}{\equeue}}
					\\
					\Gamma; \Delta_0, \Delta_2, s: \tin{\beta} T \types \Par{i}{I}{\net{P_i \pll B_i \pll \squeue{s}{c}{\pol{m}_i}}}
				}{
					\Gamma; \Delta \types	\net{\pout{\aggr s}{e} P \pll B \pll \squeue{\aggr s}{c}{\equeue}}
								\npll
								\Par{i}{I}{\net{P_i \pll B_i \pll \squeue{s}{c}{\pol{m}_i}}}
				}{\TPar}
			\]

			We continue the above derivation for network node
			$\net{\pout{\aggr s}{e} P \pll B \pll \squeue{\aggr s}{c}{\equeue}}$:
			\[
				\btree[] {
					\btree {
						\btree {
							\Gamma; \aggr s: \ctype{c}{\tempty} \types \squeue{\aggr s}{c}{\equeue}
							\quad \SEmp
							\andalso
							\Gamma; \Theta \types B
						}{
							\Gamma; \Theta, \aggr s: \ctype{c}{\tempty} \types B \pll \squeue{\aggr s}{c}{\equeue}
						}{\BPar}
						\\[6mm]
						\tree {
							\Gamma; \Delta_1'', \aggr s: \tdual{T} \types P
							\andalso
							\Gamma \types e: \beta
						}{
							\Gamma; \Delta_1'', \aggr s: \tout{\beta} \tdual{T} \types \pout{\aggr s}{e} P
						}{\TSnd}
						\\[4mm]
						\Delta_1', \aggr s: \ctype{c}{\tout{\beta} \tdual{T}} = \Delta_1'', \aggr s: \tout{\beta} \tdual{T} \sessionop \Theta, \aggr s: \ctype{c}{\tempty}
						\\[1mm]
					}{
						\Gamma; \Delta_1', \aggr s: \ctype{c}{\tout{\beta} \tdual{T}} \types \net{\pout{\aggr s}{e} P \pll B \pll \squeue{\aggr s}{c}{\equeue}}
					}{\TNode}
					\\[20mm]
					\synchronise{\Delta_1', \aggr s: \ctype{c}{\tout{\beta} \tdual{T}}}{\Delta_0, \Delta_1, \aggr s: \ctype{c}{\tout{\beta} \tdual{T}}}
					\\[1mm]
				}{
					\Gamma; \Delta_0, \Delta_1, \aggr s: \ctype{c}{\tout{\beta} \tdual{T}} \types \net{\pout{\aggr s}{e} P \pll B \pll \squeue{\aggr s}{c}{\equeue}}				
				}{\TSynch}
			\]
			We then derive the typing for network nodes
			$\net{P_i \pll B_i \pll \squeue{s}{c}{\pol{m}_i}}$
			\[
				\btree[] {
					\btree[\TNode] {
						\btree {
							\Gamma; s: \ctype{c}{M} \types \squeue{s}{c}{\pol{m}_i}
							\quad \SExp
							\andalso
							\Gamma; \Theta_i \types B_i
						}{
							\Gamma; \Theta_i, s: \ctype{c}{M} \types B_i \pll \squeue{s}{c}{\pol{m}_i}
						}{\BPar}
						\\[6mm]
						\Gamma; \Delta_i'', s: T' \types P_i
						\andalso
						\Delta_i', s: \ctype{c}{\tin{\beta} T} = \Delta_1'', s: T' \sessionop \Theta, s: \ctype{c}{M}
						\\[1mm]
					}{
						\Gamma; \Delta_i', s: \ctype{c}{\tin{\beta} T} \types \net{P_i \pll B_i \pll \squeue{s}{c}{\pol{m}_i}}
					}{}
					\\[16mm]
					\synchronise{\Delta_i', s: \ctype{c}{\tin{\beta} T}}{\Delta_0, \Delta_i, s: \ctype{c}{\tin{\beta} T}}
					\\[1mm]
				}{
					\Gamma; \Delta_0, \Delta_i, s: \ctype{c}{\tin{\beta} T} \types \net{P_i \pll B_i \pll \squeue{s}{c}{\pol{m}_i}}
				}{\TSynch}
			\]
			We can now have multiple applications of rule \TPar to get:
			\[
				\tree {
					\Gamma; \Delta_0, \Delta_i, s: \ctype{c}{\tin{\beta} T} \types \net{P_i \pll B_i \pll \squeue{s}{c}{\pol{m}_i}}
					\andalso
					\Delta_2 = \bigcup_{i \in I} \Delta_i
				}{
					\Gamma; \Delta_0, \Delta_2, s: \ctype{c}{\tin{\beta} T} \types \Par{i}{I}{\net{P_i \pll B_i \pll \squeue{s}{c}{\pol{m}_i}}}
				}{\TPar}
			\]
			Using the above information we can type the result
			of the reduction,
			$\net{P \pll B \pll \squeue{s}{c+1}{\equeue}} \npll
			\Par{i}{I}{\net{P_i \pll B_i \pll \squeue{s}{c+1}{\pol{m}_i \cat e} } }$.
			We begin with the typing derivation for network node
			$\net{P \pll B \pll \squeue{s}{c+1}{\equeue}}$:
			\[
				\btree[\TSynch] {
					\btree {
						\btree {
							\Gamma; \aggr s: \ctype{c+1}{\tempty} \types \squeue{\aggr s}{c+1}{\equeue}
							\quad \SEmp
							\andalso
							\Gamma; \Theta \types B
						}{
							\Gamma; \Theta, \aggr s: \ctype{c+1}{\tempty} \types B \pll \squeue{\aggr s}{c+1}{\equeue}
						}{\BPar}
						\\[5mm]
						\text{from former derivation }
						\Gamma; \Delta_1'', \aggr s: \tdual{T} \types P
						\\[1mm]
						\Delta_1', \aggr s: \ctype{c+1}{\tdual T} = \Delta_1'', \aggr s: \tdual{T} \sessionop \Theta, \aggr s: \ctype{c+1}{\tempty}
						\\[1mm]
					}{
						\Gamma; \Delta_1', \aggr s: \ctype{c+1}{\tdual T} \types \net{P \pll B \pll \squeue{\aggr s}{c}{\equeue}}
					}{\TNode}
					\\[18mm]
					\begin{array}{l}
						\text{Linear context synchronisation (Definition~\ref{def:synchr}) and former result}
						\\
						\synchronise{\Delta_1', \aggr s: \ctype{c}{\tout{\beta} \tdual T}}{\Delta_0, \Delta_1, \aggr s: \ctype{c}{\tout{\beta} \tdual T}}
						\text{ implies}
						\\
						\synchronise{\Delta_1', \aggr s: \ctype{c+1}{\tdual T}}{\Delta_0, \Delta_1, \aggr s: \ctype{c+1}{\tdual T}}
					\end{array}
					\\[7mm]
				}{
					\Gamma; \Delta_0, \Delta_1, \aggr s: \ctype{c+1}{\tdual T} \types \net{P \pll B \pll \squeue{\aggr s}{c+1}{\equeue}}
				}{}
			\]
			We derive the typing judgement for network node
			$\Par{i}{I}{\net{P_i \pll B_i \pll \squeue{s}{c+1}{\pol{m}_i \cat e} } }$:
			\[
				\btree[\TSynch] {
					\btree {
						\btree[\BPar] {
							\tree {
								\begin{array}{l}
									\text{from former judgement }
									\\
									\Gamma; s: \ctype{c}{M} \types \squeue{s}{c}{\pol{m}_i \cat e}
								\end{array}
							}{
								\Gamma; s: \ctype{c+1}{M. \mtout{\beta}} \types \squeue{s}{c+1}{\pol{m}_i \cat e}
							}{\SExp}
							\andalso
							\Gamma; \Theta_i \types B_i
						}{
							\Gamma; \Theta_i, s: \ctype{c+1}{M. \mtout{\beta}} \types B_i \pll \squeue{s}{c+1}{\pol{m}_i \cat e}
						}{}
						\\[10mm]
						\Gamma; \Delta_i'', s: T' \types P_i
						\\
						\begin{array}{l}
							\text{Operator $\sessionop$ (Definition~\ref{def:combine_context}) and former result}
							\\
							\Delta_i', s: \ctype{c}{\tin{\beta} T} = \Delta_1'', s: T' \sessionop \Theta, s: \ctype{c}{M}
							\text{ imply}
							\\
							\Delta_i', s: \ctype{c+1}{T} = \Delta_1'', s: T' \sessionop \Theta, s: \ctype{c+1}{M. \mtout{\beta}}
						\end{array}
					}{
						\Gamma; \Delta_i', s: \ctype{c+1}{T} \types \net{P_i \pll B_i \pll \squeue{s}{c+1}{\pol{m}_i \cat e}}
					}{\TNode}
					\\[20mm]
					\synchronise{\Delta_i', s: \ctype{c+1}{T}}{\Delta_0, \Delta_i, s: \ctype{c+1}{T}}
				}{
					\Gamma; \Delta_0, \Delta_i, s: \ctype{c+1}{T} \types \net{P_i \pll B_i \pll \squeue{s}{c+1}{\pol{m}_i \cat e}}
				}{}
			\]
			We can now have multiple applications of rule \TPar to get:
			\[
				\tree[] {
					\Gamma; \Delta_0, \Delta_i, s: \ctype{c+1}{T} \types \net{P_i \pll B_i \pll \squeue{s}{c+1}{\pol{m}_i \cat e}}
					\andalso
					\Delta_2 = \bigcup_{i \in I} \Delta_i
				}{
					\Gamma; \Delta_0, \Delta_2, s: \ctype{c+1}{T} \types \Par{i}{I}{\net{P_i \pll B_i \pll \squeue{s}{c+1}{\pol{m}_i \cat e}}}
				}{\TPar}
			\]
			The final rule needed to be applied is \TPar to get:
			\[
			\arraycolsep=2pt
			\tree[] {
				\Gamma; \Delta_0, \Delta_1, \aggr s: \ctype{c+1}{\tdual T} \types \net{P \pll B \pll \squeue{\aggr s}{c+1}{\equeue}}
				\\
				\Gamma; \Delta_0, \Delta_2, s: \ctype{c+1}{T} \types \Par{i}{I}{\net{P_i \pll B_i \pll \squeue{s}{c+1}{\pol{m}_i \cat e}}}
			}{
				\begin{array}{lcl}
					\Gamma; \Delta_0, \Delta_1, \aggr s: \ctype{c+1}{\tdual T}, \Delta_2, s: \ctype{c+1}{T} \types
					&&
					\net{P \pll B \pll \squeue{\aggr s}{c+1}{\equeue}}
					\\&\npll&
					\Par{i}{I}{\net{P_i \pll B_i \pll \squeue{s}{c+1}{\pol{m}_i \cat e}}}
				\end{array}
			}{\TPar}
			\]
			The case concludes because
			$\Delta = \Delta_0, \Delta_1, \aggr s: \ctype{c}{\tout{\beta} \tdual{T}}, \Delta_2, s: \tin{\beta} T$
			and
			$\Delta \tadvance \Delta_0, \Delta_1, \aggr s: \ctype{c+1}{\tdual T}, \Delta_2, s: \ctype{c+1}{T}$
			as required.


		\item	Rule \Unicast

			\noindent
			Assume
			\[
			\tree[\Unicast] {
				c_1 \geq c_2
			}{
				\net{\pout{s}{e} P_1 \pll B_1 \pll \squeue{s}{c_1}{\equeue}} \npll
				\net{P_2 \pll B_2 \pll \squeue{\aggr s}{c_2}{\pol{h}} }
				\\
				\qquad\qquad
				\reduces
				\net{P_1 \pll B_1 \pll \squeue{s}{c_1 + 1}{\equeue} } \npll
				\net{P_2 \pll \squeue{\aggr s}{c_2}{\pol{h} \cat (c_1, e)} }
			}{}
			\]
			and
			\[
				\Gamma; \Delta \types	\net{\pout{s}{e} P_1 \pll B_1 \pll \squeue{s}{c_1}{\equeue}} \npll
							\net{P_2 \pll B_2 \pll \squeue{\aggr s}{c_2}{\pol{h}} }
			\]
			We provide the derivation for the latter typing judgement.
			\[
				\btree {
					\begin{array}{l}
						\text{well-formed $\Delta$ implies }
						\Delta = \Delta_0, \Delta_1, s: \ctype{c_3}{T}, \Delta_2, \aggr s: \ctype{c_3}{\tdual T}
						\\
						\text{with $\Delta_0$ only $s$-endpoints and}
					\end{array}
					\\
					\Gamma; \Delta_0, \Delta_1, s: \ctype{c_3}{T} \types \net{\pout{s}{e} P_1 \pll B_1 \pll \squeue{s}{c_1}{\equeue}}
					\\
					\Gamma; \Delta_0, \Delta_0, \aggr s: \ctype{c_3}{\tdual T} \types \net{P_2 \pll B_2 \pll \squeue{\aggr s}{c_2}{\pol{h}} }
				}{
					\Gamma; \Delta \types	\net{\pout{s}{e} P_1 \pll B_1 \pll \squeue{s}{c_1}{\equeue}} \npll
								\net{P_2 \pll B_2 \pll \squeue{\aggr s}{c_2}{\pol{h}} }
				}{\TPar}
			\]
			We continue with the typing derivation for network node:
			$\net{\pout{s}{e} P_1 \pll B_1 \pll \squeue{s}{c_1}{\equeue}}$
			\[
				\btree[] {
					\btree {
						\btree {
							\Gamma; s: \ctype{c_1}{\tempty} \types \squeue{s}{c_1}{\equeue}
							\  \SEmp
							\andalso
							\Gamma; \Theta_1 \types B_1
						}{
							\Gamma; \Theta_1, s: \ctype{c_1}{\tempty} \types B_1 \pll \squeue{s}{c_1}{\equeue}
						}{\BPar}
						\\[6mm]
						\tree {
							\Gamma; \Delta_1'', s: T_1 \types P_1
							\andalso
							\Gamma \types e: \beta
						}{
							\Gamma; \Delta_1'', s: \tout{\beta} T_1 \types \pout{\aggr s}{e} P_1
						}{\TSnd}
						\\[4mm]
						\Delta_1', s: \ctype{c_1}{\tout{\beta} T_1} = \Delta_1'', s: \tout{\beta} T_1 \sessionop \Theta_1, s: \ctype{c_1}{\tempty}
						\\[1mm]
					}{
						\Gamma; \Delta_1', s: \ctype{c_1}{\tout{\beta} T_1} \types \net{\pout{s}{e} P_1 \pll B_1 \pll \squeue{s}{c_1}{\equeue}}
					}{\TNode}
					\\[20mm]
					\tree {
						\synchronise{\Delta_1'}{\Delta_0, \Delta_1}
						\andalso
						c_3 \geq c_1
						\andalso
						\tout{\beta} T_1 \tadvance^{c_3 - c_1} T
					}{
						\synchronise{\Delta_1', s: \ctype{c_1}{\tout{\beta} T_1}}{\Delta_0, \Delta_1, s: \ctype{c_3}{T}}
					}{}
					\\[4mm]
				}{
					\Gamma; \Delta_0, \Delta_1, s: \ctype{c_3}{T} \types \net{\pout{s}{e} P_1 \pll B_1 \pll \squeue{s}{c_1}{\equeue}}
				}{\TSynch}
			\]
			We also continue the derivation of network node
			$\net{P_2 \pll B_2 \pll \squeue{\aggr s}{c_2}{\pol{h}} }$:
			\[
				\btree[] {
					\btree[] {
						\btree {
							\tree{
								\dots
							}{
								\Gamma; \aggr s: \ctype{c_3}{M} \types \squeue{\aggr s}{c_2}{\pol{h}}
							}{\LExp}
							\\
							c_3 \geq c_1
							\andalso
							\Gamma; \Theta_2 \types B_2
						}{
							\Gamma; \Theta_2, \aggr s: \ctype{c_3}{M} \types B_2 \pll \squeue{\aggr s}{c_2}{\pol{h}}
						}{\BPar}
						\\[12mm]
						\Gamma; \Delta_2'', \aggr s: T_2 \types P_2
						\andalso
						\Delta_2', \aggr s: \ctype{c_3}{\tdual T} = \Delta_2'', \aggr s: T_2 \sessionop \Theta_2, \aggr s: \ctype{c_3}{M}
					}{
						\Delta_2', \aggr s: \ctype{c_3}{\tdual T} \types \net{P_2 \pll B_2 \pll \squeue{\aggr s}{c_2}{\pol{h}} }
					}{\TNode}
					\\[22mm]
					\synchronise{\Delta_2', \aggr s: \ctype{c_3}{\tdual T}}{\Delta_0, \Delta_2, \aggr s: \ctype{c_3}{\tdual T}}
					\\[1mm]
				}{
					\Delta_0, \Delta_2, \aggr s: \ctype{c_3}{\tdual T} \types \net{P_2 \pll B_2 \pll \squeue{\aggr s}{c_2}{\pol{h}} }
				}{\TSynch}
			\]
			Using the above information we can type the result
			of the reduction,
			$\net{P_1 \pll B_1 \pll \squeue{s}{c_1 + 1}{\equeue} } \npll
			\net{P_2 \pll \squeue{\aggr s}{c_2}{\pol{h} \cat (c_1, e)}}$.
			We begin with the typing derivation for network node
			$\net{P_1 \pll B_1 \pll \squeue{s}{c_1 + 1}{\equeue} }$:
			\[
				\btree[] {
					\btree[\TNode] {
						\btree[] {
							\Gamma; s: \ctype{c_1+1}{\tempty} \types \squeue{s}{c_1 + 1}{\equeue}
							\ \SEmp
							\andalso
							\Gamma; \Theta_1 \types B_1
						}{
							\Gamma; \Theta_1, s: \ctype{c_1 + 1}{\tempty} \types B_1 \pll \squeue{s}{c_1 + 1}{\equeue}
						}{\BPar}
						\\[6mm]
						\Gamma; \Delta_1'', s: T_1 \types P_1
						\andalso
						\Delta_1', s: \ctype{c_1+1}{T_1} = \Delta_1'', s: T \sessionop \Theta_1, s: \ctype{c_1 + 1}{\tempty}
						\\[1mm]
					}{
						\Gamma; \Delta_1', s: \ctype{ c_1 + 1}{T_1} \types \net{P_1 \pll B_1 \pll \squeue{s}{c_1 + 1}{\equeue}}
					}{}
					\\[17mm]
					\begin{array}{l}
						\text{results }
						\synchronise{\Delta_1', s: \ctype{c_1}{\tout{\beta} T_1}}{\Delta_0, \Delta_1, s: \ctype{c_3}{T}}
						\text{ and $c_3 \geq c_1$ imply}
						\\
						\tree {
							\synchronise{\Delta_1'}{\Delta_0, \Delta_1}
							\andalso
							c_4 = c_3 \vee c_4 = c_3 + 1
							\andalso
							T_1 \tadvance^{c_4 - c_1 + 1} T'
						}{
							\synchronise{\Delta_1', s: \ctype{c_1 + 1}{T_1}}{\Delta_0, \Delta_1, s: \ctype{c_4}{T'}}
						}{}
					\end{array}
					\\[7mm]
				}{
					\Gamma; \Delta_0, \Delta_1, s: \ctype{c_4}{T'} \types \net{P_1 \pll B_1 \pll \squeue{s}{c_1 + 1}{\equeue}}
				}{\TSynch}
			\]
			We continue with the typing derivation for network node
			$\net{P_2 \pll \squeue{\aggr s}{c_2}{\pol{h} \cat (c_1, e)}}$:
			\[
				\btree[] {
					\btree[\TNode] {
						\btree[\BPar] {
							\tree{
								\Gamma;\aggr s: \ctype{c_3}{M} \types \squeue{\aggr s}{c_2}{\pol{h}}
								\\
								c_4 = c_3 \vee c_4 = c_3 + 1
							}{
								\Gamma; \aggr s: \ctype{c_4}{M.\mtout{\beta}} \types \squeue{\aggr s}{c_2}{\pol{h} \cat (c_1, e)}
							}{\LExp}
							\andalso
							\Gamma; \Theta_2 \types B_2
						}{
							\Gamma; \Theta_2, \aggr s: \ctype{c_4}{M.\mtout{\beta}} \types B_2 \pll \squeue{\aggr s}{c_2}{\pol{h}\cat (c_1, e)}
						}{}
						\\[14mm]
						\Gamma; \Delta_2'', \aggr s: T_2 \types P_2
						\andalso
						\Delta_2', \aggr s: \ctype{c_4}{\tdual T'} = \Delta_2'', \aggr s: T_2 \sessionop \Theta_2, \aggr s: \ctype{c_4}{M. \mtout{\beta}}
						\\[1mm]
					}{
						\Delta_2', \aggr s: \ctype{c_4}{\tdual T'} \types \net{P_2 \pll B_2 \pll \squeue{\aggr s}{c_2}{\pol{h}} }
					}{}
					\\[26mm]
					\synchronise{\Delta_2', \aggr s: \ctype{c_4}{\tdual T'}}{\Delta_0, \Delta_2, \aggr s: \ctype{c_4}{\tdual T'}}
					\\[1mm]
				}{
					\Delta_0, \Delta_2, \aggr s: \ctype{c_4}{\tdual T'} \types \net{P_2 \pll B_2 \pll \squeue{\aggr s}{c_2}{\pol{h}} }
				}{\TSynch}
			\]
			We apply rule \TPar to get
			\[			
				\tree {
					\Gamma; \Delta_0, \Delta_1, s: \ctype{c_4}{T'} \types \net{P_1 \pll B_1 \pll \squeue{s}{c_1 + 1}{\equeue}}
					\\
					\Delta_0, \Delta_2, \aggr s: \ctype{c_4}{\tdual{T'}} \types \net{P_2 \pll B_2 \pll \squeue{\aggr s}{c_2}{\pol{h}} }
				}{
						\Gamma; \Delta_0, \Delta_1, s: \ctype{c_4}{T'}, \Delta_2, \aggr s: \ctype{c_4}{\tdual T'} \types
							\net{P_1 \pll B_1 \pll \squeue{s}{c_1 + 1}{\equeue}}
							\npll
							\net{P_2 \pll B_2 \pll \squeue{\aggr s}{c_2}{\pol{h}}}
				}{\TPar}
			\]
			From here there are two cases:
			If $c_4 = c_3$ then $T' = T$ and the result 
			follows from the fact that $\Delta_0, \Delta_1, s: \ctype{c_4}{T'}, \Delta_2, \aggr s: \ctype{c_4}{\tdual{T'}} = \Delta_0, \Delta_1, s: \ctype{c_3}{T}, \Delta_2, \aggr s: \ctype{c_3}{\tdual T}$.
			If $c_4 = c_3 + 1$ then $T \tadvance^{1} T'$
			and the results follows the fact that
			$\Delta_0, \Delta_1, s: \ctype{c_3}{T}, \Delta_2, \aggr s: \ctype{c_3}{\tdual T} \tadvance^{1} \Delta_0,s: \ctype{c_4}{T'}, \Delta_2, \aggr s: \ctype{c_4}{\tdual{T'}}$


		\item	Rule \Receive

			\noindent
			Assume
			\[
				\net{\pdin{s}{x}{e'} P \pll B \pll \squeue{s}{c}{e \cat \pol{m}}}
				\reduces
				\net{P \subst{e}{x} \pll B \pll \squeue{s}{c}{\pol{m}}}
			\]
			with
			\[
				\Gamma; \Delta \types \net{\pdin{s}{x}{e'} P \pll B \pll \squeue{s}{c}{e \cat \pol{m}}}
			\]
			and $\Delta$ well-formed.
			We give the derivation for the latter typing judgement
			\[
				\btree[] {
					\btree{
						\tree{
							\Gamma; s: \ctype{c}{M} \types B \pll \squeue{s}{c}{\pol{m}}
						}{
							\Gamma; s: \ctype{c}{\mtout{\beta}. M} \types B \pll \squeue{s}{c}{e \cat \pol{m}}
						}{\SExp}
						\andalso
						\Gamma; \Theta \types B
						\\[4mm]
					}{
						\Gamma; \Theta, s: \ctype{c}{\mtout{\beta}. M} \types B \pll \squeue{s}{c}{e \cat \pol{m}}
					}{\BPar}
					\\[10mm]
					\btree {
						\Gamma \types e': \beta
						\andalso
						\Gamma, x: \beta; \Delta', s: T' \types P
					}{
						\Gamma; \Delta', s: \tin{\beta} T' \types \pdin{s}{x}{e'} P
					}{\TRcv}
					\\[6mm]
					\begin{array}{rcl}
						\Delta = \Delta', s: \ctype{c}{T}	&=&	\Delta', s: \tin{\beta} T' \sessionop \Theta, \ctype{c}{\mtout{\beta}. M}
						\\					&=&	\Delta, s: T' \sessionop \Theta, s: \ctype{c}{M}
					\end{array}
					\\[4mm]
				}{
					\Gamma; \Delta \types \net{\pdin{s}{x}{e'} P \pll B \pll \squeue{s}{c}{e \cat \pol{m}}}
				}{\TNode}
			\]
			The above result allows us to produce the
			typing derivation for the result of the reduction,
			$\net{P \subst{e}{x} \pll B \pll \squeue{s}{c}{\pol{m}}}$:
			\[
				\btree {
					\btree{
						\Gamma; s: \ctype{c}{M} \types B \pll \squeue{s}{c}{\pol{m}}
						\andalso
						\Gamma; \Theta \types B
					}{
						\Gamma; \Theta, s: \ctype{c}{M} \types B \pll \squeue{s}{c}{\pol{m}}
					}{\BPar}
					\\[6mm]
					\Gamma; \Delta', s: T' \types P
					\andalso
					\Delta = \Delta', s: \ctype{c}{T} = \Delta, s: T' \sessionop \Theta, s: \ctype{c}{M}
					\\[1mm]
				}{
					\Gamma; \Delta \types \net{P \subst{e}{x} \pll B \pll \squeue{s}{c}{\pol{m}}}
				}{\TNode}
			\]
			as required.
		\item	Rule \Gather.

			\noindent
			Assume
			\[
				\tree {
					\pol{h}' = \newbuffer{\pol{h}}{c + 1}
					\andalso
					e = \gthr{\pol{h}}{c + 1}
				}{
					\net{\pin{\aggr s}{x} P \pll B \pll \squeue{\aggr s}{c}{ \pol{h} }}
					\reduces
					\net{P \subst{e}{x} \pll B \pll \squeue{\aggr s}{c + 1}{ \pol{h}' }}
				}{\Gather}
			\]
			with
			\[
				\Gamma; \Delta \types \net{\pin{\aggr s}{x} P \pll B \pll \squeue{\aggr s}{c}{ \pol{h} }}
			\]
			and $\Delta$ well-formed.
			We produce the derivation of the latter typing judgement:
			We first type the session buffer
			$\squeue{\aggr s}{c}{ \pol{h} }$:
			\[
				\begin{array}{ll}
					\left.
					\begin{array}{c}
					\btree {
						\btree[] {
							\Gamma; \aggr s: \ctype{c}{\tempty} \types \squeue{\aggr s}{c}{\equeue}
							\ \SEmp
							\\
							\equeue = \newbuffer{\pol{h}_1}{c + 1}
							\andalso
							\Gamma \types V(\pol{h}_1, c + 1) : \beta
						}{
							\Gamma; \aggr s: \ctype{c+1}{\mtout{\beta}} \types \squeue{\aggr s}{c}{\pol{h}_1}
						}{\LExp}
						\\
						\vdots
						\\
						\Gamma; \aggr s: \ctype{c'-1}{\mtout{\beta}. M} \types \squeue{\aggr s}{c}{ \pol{h}_{c' - c} }
						\\
						\pol{h} = \newbuffer{\pol{h}_{c' - c}}{c'}
						\andalso
						\Gamma \types V(\pol{h}, c') : \beta'
					}{
						\Gamma; \aggr s: \ctype{c'}{\mtout{\beta}. M. \mtout{\beta'}} \types \squeue{\aggr s}{c}{\pol{h}}
					}{\LExp}
					\end{array}
					\right\}
					&
					\begin{array}{l}
						c' - c 
						\\ \text{applications}
						\\ \text{of rule \LExp}
					\end{array}
				\end{array}
			\]
			We now produce the typing derivation for term
			$\net{\pin{\aggr s}{x} P \pll B \pll \squeue{\aggr s}{c}{ \pol{h} }}$:
			\[
				\btree{
					\btree{
						\Gamma; \aggr s: \ctype{c'}{\mtout{\beta}. M. \mtout{\beta'}} \types \squeue{\aggr s}{c}{\pol{h}}
						\andalso
						\Gamma; \Theta \types B
					}{
						\Gamma; \Theta, \aggr s: \ctype{c'}{\mtout{\beta}. M. \mtout{\beta'}} \types B \pll \squeue{\aggr s}{c}{ \pol{h} }
					}{\BPar}
					\\[6mm]
					\tree {
						\Gamma, x: \beta; \Delta', \aggr s: T' \types P
					}{
						\Gamma; \Delta', \aggr s: \tin{\beta} T' \types \pin{\aggr s}{x} P
					}{\TRcv}
					\\[4mm]
					\begin{array}{rcl}
						\Delta = \Delta', \aggr s: \ctype{c'}{T} 	&=&	\Delta', \aggr s: \tin{\beta} T' \sessionop \Theta, \aggr s: \ctype{c'}{\mtout{\beta}. M. \mtout{\beta'}}
						\\						&=&	\Delta', \aggr s: T' \sessionop \Theta, \aggr s: \ctype{c'}{M. \mtout{\beta'}}
					\end{array}
					\\[4mm]
				}{
					\Gamma; \Delta \types \net{\pin{\aggr s}{x} P \pll B \pll \squeue{\aggr s}{c}{ \pol{h} }}
				}{\TNode}
			\]
			We use the above information to produce the typing derivation
			for the result of the reduction.
			We first produce the typing derivation for session buffer
			$\squeue{\aggr s}{c + 1}{ \pol{h}' }$
			\[
				\begin{array}{ll}
					\left.
					\begin{array}{c}
					\btree {
						\btree[] {
							\Gamma; \aggr s: \ctype{c+1}{\tempty} \types \squeue{\aggr s}{c+1}{\equeue}
							\ \SEmp
							\\[4mm]
							\equeue = \newbuffer{\pol{h}_2'}{c + 2}
							\andalso
							\Gamma \types V(\pol{h}_2', c + 2) : \beta''
						}{
							\Gamma; \aggr s: \ctype{c+2}{\mtout{\beta''}} \types \squeue{\aggr s}{c+1}{\pol{h}_2'}
						}{\LExp}
						\\
						\vdots
						\\
						\Gamma; \aggr s \ctype{c'-1}{M} \types \squeue{\aggr s}{c + 1}{ \pol{h}_{c' - c + 1}' }
						\\
						\pol{h}' = \newbuffer{\pol{h}_{c' - c + 1}'}{c'}
						\andalso
						\Gamma \types V(\pol{h}', c') : \beta'
					}{
						\Gamma; \aggr s: \ctype{c'}{M. \mtout{\beta'}} \types \squeue{\aggr s}{c+1}{\pol{h}'}
					}{\LExp}
					\end{array}
					\right\}
					&
					\begin{array}{l}
						c' - c + 1
						\\ \text{applications}
						\\ \text{of rule \LExp}
					\end{array}
				\end{array}
			\]
			Finally we give the typing derivation for network node:
			$\net{P \subst{e}{x} \pll B \pll \squeue{\aggr s}{c + 1}{ \pol{h}' }}$
			\[
				\btree{
					\btree[\BPar] {
						\Gamma; \aggr s: \ctype{c'}{M. \mtout{\beta'}} \types \squeue{\aggr s}{c+1}{\pol{h}'}
						\andalso \andalso
						\Gamma; \Theta \types B
					}{
						\Gamma; \Theta, \aggr s: \ctype{c'}{M. \mtout{\beta'}} \types B \pll \squeue{\aggr s}{c'}{\pol{h}'}
					}{}
					\\[8mm]
					\begin{array}{l}
						\text{From Substitution Lemma~\ref{lem:substitution} }
						\\
						\Gamma, x: \beta; \Delta', \aggr s: T' \types P \text{ implies }
						\Gamma; \Delta', \aggr s: T' \types P \subst{e}{x}
					\end{array}
					\\[4mm]
					\Gamma; \econtext \types R
					\andalso
					\Delta = \Delta', \aggr s: \ctype{c'}{T} = \Delta', \aggr s: T' \sessionop \Theta, \aggr s: \ctype{c'}{M. \mtout{\beta'}} 
				}{
					\Gamma; \Delta \types \net{\pin{\aggr s}{x} P \pll B \pll \squeue{\aggr s}{c}{ \pol{h} }}
				}{\TNode}
			\]
			The result follows the fact that $\Delta \tadvance \Delta$.

		\item	Rule \Select. The proof is similar to the proof for rule \Broadcast.
		\item	Rule \Branch. The proof is similar to the proof for rule \Receive.

		\item	Rule \Loss

			\noindent
		        Assume
			\[
				\net{\pout{s}{e} P \pll B \pll \squeue{s}{c}{\equeue}}
				\reduces
				\net{P \pll B \pll \squeue{s}{c + 1}{\equeue}}
			\]
			with
			\[
				\Gamma; \Delta \types \net{\pout{s}{e} P \pll B \pll \squeue{s}{c}{\equeue}}
			\]
			and $\Delta$ well-formed.
			We provide the derivation of the latter typing judgement:
			\[
				\btree{
					\btree {
						\Gamma; s: \ctype{c}{\tempty} \types \squeue{s}{c}{\equeue}
						\ \SEmp
						\andalso
						\Gamma; \Theta \types B
					}{
						\Gamma; \Theta, s: \ctype{c}{\tempty} \types B \pll \squeue{s}{c}{\equeue}
					}{\BPar}
					\\[6mm]
					\btree {
						\Gamma; \Delta', s: T \types P
						\andalso
						\Gamma \types e: \beta
					}{
						\Gamma; \Delta', s: \tout{\beta} \types \pout{s}{e} P
					}{\TSnd}
					\\[6mm]
					\Delta = \Delta', s: \ctype{c}{\tout{\beta} T} = \Delta', s: \tout{\beta} T \sessionop \Theta, s: \ctype{c}{\tempty}
					\\[1mm]
				}{
					\Gamma; \Delta \types \net{\pout{s}{e} P \pll B \pll \squeue{s}{c}{\equeue}}
				}{\TNode}
			\]
			The above result allows us to produce the
			typing derivation for the result of the reduction,
			$\net{P \pll B \pll \squeue{s}{c + 1}{\equeue}}$
			\[
				\btree[\TSynch] {
					\btree {
						\btree[] {
							\Gamma; s: \ctype{c+1}{\tempty} \types \squeue{s}{c+1}{\equeue}
							\ \SEmp
							\andalso
							\Gamma; \Theta \types B
						}{
							\Gamma; \Theta, s: \ctype{c+1}{\tempty} \types B \pll \squeue{s}{c+1}{\equeue}
						}{\BPar}
						\\[6mm]
						\Gamma; \Delta', s: T \types P
						\andalso
						\Delta', s: \ctype{c+1}{T} = \Delta', s: T \sessionop \Theta, s: \ctype{c+1}{\tempty}
						\\[1mm]
					}{
						\Gamma; \Delta', s: \ctype{c+1}{T} \types \net{P \pll B \pll \squeue{s}{c + 1}{\equeue}}
					}{\TNode}
					\\[14mm]
					\tree{
						\synchronise{\Delta'}{\Delta'}
						\andalso
						\tout{\beta} T \tadvance_o^{c + 1 - c} T
					}{
						\synchronise{\Delta', s: \ctype{c+1}{T}}{\Delta', s: \ctype{c}{\tout{\beta} T}}
					}{}
					\andalso
					\Delta = \Delta', s: \ctype{c}{\tout{\beta} T}
					\\[4mm]
				}{
					\Gamma; \Delta \types \net{P \pll B \pll \squeue{s}{c + 1}{\equeue}}
				}{}
			\]
			as required.


		\item	Rule \Recover

			\noindent
			Assume
			\[
				\net{\pdin{s}{x}{e} P \pll B \pll \squeue{s}{c}{\equeue}}
				\reduces
				\net{P \subst{e}{x} \pll B \pll \squeue{s}{c+1}{\equeue}}
				\ \Recover
			\]
			with
			\[
				\Gamma; \Delta \types \net{\pdin{s}{x}{e} P \pll B \pll \squeue{s}{c}{\equeue}}
			\]
			and $\Delta$ well-formed.
			We give the derivation for the latter typing judgement
			\[
				\btree[] {
					\btree{
						\Gamma; s: \ctype{c}{\tempty} \types \squeue{s}{c}{\equeue}
						\ \SEmp
						\andalso
						\Gamma; \Theta \types B
					}{
						\Gamma; \Theta, s: \ctype{c}{\tempty} \types B \pll \squeue{s}{c}{\equeue}
					}{\BPar}
					\\[6mm]
					\btree {
						\Gamma \types e: \beta
						\andalso
						\Gamma, x: \beta; \Delta', s: T \types P
					}{
						\Gamma; \Delta', s: \tin{\beta} T \types \pdin{s}{x}{e} P
					}{\TRcv}
					\\[6mm]
					\Delta = \Delta', s: \ctype{c}{\tin{\beta} T}	=	\Delta', s: \tin{\beta} T \sessionop \Theta, \ctype{c}{\tempty}
					\\[1mm]
				}{
					\Gamma; \Delta \types \net{\pdin{s}{x}{e'} P \pll B \pll \squeue{s}{c}{e \cat \pol{m}}}
				}{\TNode}
			\]
			The above result allows us to produce the
			typing derivation for the result of the reduction,
			$\net{P \subst{e}{x} \pll B \pll \squeue{s}{c+1}{\equeue}}$
			\[
				\btree[\TSynch] {
					\btree {
						\btree{
							\Gamma; s: \ctype{c+1}{\tempty} \types \squeue{s}{c+1}{\tempty}
							\andalso
							\Gamma; \Theta \types B
						}{
							\Gamma; \Theta, s: \ctype{c+1}{\tempty} \types B \pll \squeue{s}{c}{\pol{m}}
						}{\BPar}
						\\[6mm]
						\Gamma; \Delta', s: T \types P \subst{e}{x}
						\andalso
						\Delta', s: \ctype{c+1}{T} = \Delta, s: T \sessionop \Theta, s: \ctype{c+1}{\tempty}
						\\[1mm]
					}{
						\Gamma; \Delta', s: \ctype{c+1}{T} \types \net{P \subst{e}{x} \pll B \pll \squeue{s}{c}{\pol{m}}}
					}{\TNode}
					\\[14mm]
					\synchronise{\Delta', s: \ctype{c+1}{T}}{\Delta', s: \ctype{c}{\tin{\beta} T}}
					\andalso
					\Delta = \Delta', s: \ctype{c}{\tin{\beta} T}
					\\[1mm]
				}{
					\Gamma; \Delta \types \net{P \subst{e}{x} \pll B \pll \squeue{s}{c}{\pol{m}}}
				}{}
			\]
			as required.

		\item	Rule \BRecover

			\noindent
			Assume
			\[
				\tree {
					\fs{B} = \fs{R}
				}{
					\net{\branchDef{s}{\ell_i: P_i}{R}_{i \in I} \pll B \pll B' \pll \squeue{s}{c}{\equeue}}
					\reduces
					\net{R \pll B}
				}{}
				\ \BRecover
			\]
			with
			\[
				\Gamma; \Delta \types \net{\branchDef{s}{\ell_i: P_i}{R}_{i \in I} \pll B \pll B' \pll \squeue{s}{c}{\equeue}}
			\]
			and $\Delta$ well-formed.
			The last result follows typing derivation:
			\[
				\btree {
					\btree{
						\Gamma; s: \ctype{c}{\tempty} \types \squeue{s}{c}{\equeue}
						\ \SEmp
						\\
						\Gamma; \Theta_1 \types B'
						\andalso
						\Gamma; \Theta_2 \types B
						\andalso
						\Theta = \Theta_1, \Theta_2
					}{
						\Gamma; \Theta, s: \ctype{c}{\tempty} \types B \pll B' \pll \squeue{s}{c}{\equeue}
					}{\BPar}
					\\[8mm]
					\btree {
						\forall i \in I, \Gamma; \Delta_1, \Delta_2, s: T_i \types P_i
						\\
						\text{only $s$-endpoints in } \Delta_1
						\andalso
						\Gamma; \Delta_2 \types R
						\andalso
						\Delta' = \Delta_1, \Delta_2
					}{
						\Gamma; \Delta', s: \tbranch{l_i}{T_i}_{i \in I} \types \branchDef{s}{\ell_i: P_i}{R}_{i \in I}
					}{\TBr}
					\\[8mm]
					\Delta = \Delta', s: \tbranch{l_i}{T_i}_{i \in I} \sessionop \Theta, s: \ctype{c}{\tempty}
					\\[1mm]
				}{
					\Gamma; \Delta \types \net{\branchDef{s}{\ell_i: P_i}{R}_{i \in I} \pll B \pll B' \pll \squeue{s}{c}{\equeue}}
				}{\TNode}
			\]
			We can now provide the typing derivation for the result of the
			reduction
			$\net{R \pll B}$
			\[
				\btree {
					\Gamma; \Delta_2 \types R
					\andalso
					\Gamma; \Theta_2 \types B
					\andalso
					\fs{R} = \fs{B} \text{ implies }
					\Delta'' = \Delta_2 \sessionop \Theta_2
				}{
					\Gamma; \Delta'' \types \net{R \pll B}
				}{\TNode}
			\]
			Because $\Delta_1$ includes only $s$-enpoints, we can use $\tadvance$ to drop them:
			\begin{eqnarray*}
				\Delta	&=&	\Delta_1, \Delta_2, s: \tbranch{l_i}{T_i}_{i \in I} \sessionop \Theta_1, \Theta_2, s: \ctype{c}{\tempty}
				\\	&=&	(\Delta_1 \sessionop \Theta_1), (\Delta_2 \sessionop \Theta_2), s: \ctype{c}{\tbranch{l_i}{T_i}_{i \in I}}
				\\	&\tadvance& (\Delta_2 \sessionop \Theta_2) 
				\\	&=& \Delta''
			\end{eqnarray*}
			as required

		\item	Rule \True

			\noindent
			Assume
			\[
				\tree {
					\condtrue \econd
					\andalso
					\fs{P_1} = \fs{B} 
				}{
					\net{\cond{\econd}{P_1}{P_2} \pll B \pll B'}
					\reduces
					\net{P_1 \pll B}
				}{\True}
			\]
			with
			\[
				\Gamma; \Delta \types \net{\cond{\econd}{P_1}{P_2} \pll B \pll B'}
			\]
			and
			$\Delta$ well-type.
			The derivation for the latter typing judgement is
			\[
				\btree[] {
					\btree {
						\Gamma; \Theta_1 \types B
						\andalso
						\Gamma; \Theta_2 \types B'
					}{
						\Gamma; \Theta_1, \Theta_2 \types B \pll B'
					}{\BPar}
					\\[6mm]
					\btree[\TCond] {
						\Gamma \types \econd: \bool
						\andalso
						\Gamma; \Delta_0, \Delta_1 \types P_1
						\andalso
						\Gamma; \Delta_0, \Delta_2 \types P_2
						\\
						i, j \in \set{1,2} \land i \not=j \implies \dom{\Delta_i} = \fs{P_i} - \fs{P_j}
					}{
						\Gamma; \Delta_0, \Delta_1, \Delta_2 \types \cond{\econd}{P_1}{P_2}
					}{}
					\\[10mm]
					\Gamma; \econtext \types R
					\andalso
					\Delta = \Delta_0, \Delta_1, \Delta_2 \sessionop \Theta_1, \Theta_2
				}{
					\Gamma; \Delta \types \net{\cond{\econd}{P_1}{P_2} \pll B \pll B'}
				}{\TNode}
			\]
			We use the above information to type the result of the reduction
			$\net{P_1 \pll B}$:
			\[
				\tree {
					\Gamma; \Delta_0, \Delta_1 \types P_1
					\andalso
					\Gamma; \econtext \types R
					\andalso
					\Gamma; \Theta_1 \types B
					\andalso
					\Delta' = \Delta_0, \Delta_1 \sessionop \Theta_1
				}{
					\Gamma; \Delta' \types \net{P_1 \pll B}
				}{\TNode}
			\]
			The reduction condition requires that
			$\fs{B} = \fs{P_1}$.
			From the last equation and the fact that
			$\fs{B} = \dom{\Theta_1}$ we have that 
			$\fs{P_1} = \dom{\Theta_1}$.
			Also $\fs{P_1} = \dom{\Delta_0, \Delta_1}$.
			From the latter two equalities
			we can conclude that
			$\Delta_0, \Delta_1 \sessionop \Theta_1$
			is defined and $\fs{P_1} = \dom{\Delta'}$.
			Moreover, $\fs{P_1} \cup \fs{P_2} = \Delta$.
			Therefore, $\Delta' \subseteq \Delta$
			as required.

		\item	Rule \False. The proof the same with the proof of case \True.


		\item	Rule \NDet.

			\noindent
			Assume
			\[
				\tree {
					\net{P_1 \pll B} \npll N
					\reduces
					\net{P' \pll B'} \npll N'
				}{
					\net{P_1 \Sum P_2 \pll B} \npll N
					\reduces
					\net{P' \pll B'} \npll N'
				}{\NonDet}
			\]
			with
			\[
				\Gamma; \Delta \types \net{P_1 \Sum P_2 \pll B} \npll N
			\]
			and $\Delta$ well-formed.

			We produce the derivation for the latter typing judgement.
			\[
			\btree {
				\btree {
					\btree{
						\btree[\TSum] {
							\Gamma; \Delta_1'' \types P_1
							\andalso
							\Gamma; \Delta_1'' \types P_2
						}{
							\Gamma; \Delta_1'' \types P_1 \Sum P_2
						}{}
						\\[3mm]
						\Gamma; \Theta \types B
						\andalso
						\Delta_1' = \Delta_1'' \sessionop \Theta
						\\[1mm]
					}{
						\Gamma; \Delta_1' \types \net{P_1 \Sum P_2 \pll B}
					}{\TNode}
					\andalso
					\synchronise{\Delta_1'}{\Delta_0, \Delta_1}
					\\[16mm]
				}{
					\Gamma; \Delta_0, \Delta_1 \types \net{P_1 \Sum P_2 \pll B}
				}{\TSynch}
				\\[20mm]
				\Gamma; \Delta_0, \Delta_2 \types N
				\andalso
				\Delta_0 \text{ only $s$-endpoints}
				\andalso
				\Delta = \Delta_0, \Delta_1, \Delta_2
				\\[1mm]
			}{
				\Gamma; \Delta \types \net{P_1 \Sum P_2 \pll B} \npll N
			}{\TPar}
			\]
			We can now produce the typing derivation for
			the left hand side network of the premise
			$\net{P_1 \pll B} \npll N$:
			\[
			\btree {
				\btree {
					\btree[\TNode] {
						\Gamma; \Delta_1'' \types P_1
						\andalso
						\Gamma; \Theta \types B
						\andalso
						\Delta_1' = \Delta_1'' \sessionop \Theta	
					}{
						\Gamma; \Delta_1' \types \net{P_1 \pll B}
					}{}
					\andalso
					\synchronise{\Delta_1'}{\Delta_0, \Delta_1}
					\\[8mm]
				}{
					\Gamma; \Delta_0, \Delta_1 \types \net{P_1 \pll B}
				}{\TSynch}
				\\[14mm]
				\Gamma; \Delta_0, \Delta_2 \types N
				\andalso
				\Delta_0 \text{ only $s$-endpoints}
				\andalso
				\Delta = \Delta_0, \Delta_1, \Delta_2
				\\[1mm]
			}{
				\Gamma; \Delta \types \net{P_1 \pll B} \npll N
			}{\TPar}
			\]
			The latter two derivations have the same
			typing contexts.
			The result is then immediate from the
			induction hypothesis;
			from the fact that type preservation holds on the
			premise of the reduction, we have that 
			\[
				\Gamma; \Delta' \types \net{P' \pll B'} \npll N'
			\]
			and
			$\Delta \tadvance \Delta'$.
			The latter results is what needed to
			conclude type preservation
			for the reduction.


		\item	Rule \RDef.

			\noindent
			Assume
			\[
				\tree {
					\net{P \pll B} \npll N
					\reduces
					\net{P' \pll B'} \npll N'
				}{
					\net{\Def{\set{\abs{D_i}{\tilde{x}_i} \defeq P_i}_{i \in I}}{P} \pll B} \npll N
					\reduces
					\net{\Def{\set{\abs{D_i}{\tilde{x}_i} \defeq P_i}_{i \in I}}{P'} \pll B' } \npll N'
				}{\RDef}
			\]
			with
			\[
				\Gamma; \Delta \types \net{\Def{\set{\abs{D_i}{\tilde{x}_i} \defeq P_i}_{i \in I}}{P} \pll B} \npll N
			\]
			and $\Delta$ well-formed.

			We produce the typing derivation of the latter judgement.
			We begin with the typing derivation for network node
			$\net{\Def{\set{\abs{D_i}{\tilde{x}_i} \defeq P_i}_{i \in I}}{P} \pll B}$:
			\[
				\btree[] {
					\btree[] {
						\tree[\TRec]{
							\forall i \in I, \abs{D_i}{\tilde{x}_i}: (\Gamma_i, \Delta_i) \in \Gamma \land \Gamma, \Gamma_i; \Delta_i \types P_i
							\andalso
							\Gamma; \Delta_1'' \types P
						}{
							\Gamma; \Delta_1'' \types \Def{\set{\abs{D_i}{\tilde{x}_i} \defeq P_i}_{i \in I}}{P}
						}{}
						\\
						\Gamma; \Theta \types B
						\andalso
						\Delta_1' = \Delta_1'' \sessionop \Theta
					}{
						\Gamma; \Delta_1' \types \net{\Def{\set{\abs{D_i}{\tilde{x}_i} \defeq P_i}_{i \in I}}{P} \pll B}
					}{\TNode}
					\\[14mm]
					\synchronise{\Delta_1'}{\Delta_0, \Delta_1}
					\\[1mm]
				}{
					\Gamma; \Delta_0, \Delta_1 \types \net{\Def{\set{\abs{D_i}{\tilde{x}_i} \defeq P_i}_{i \in I}}{P} \pll B}
				}{\TSynch}
			\]
			We type the entire network node
			$\net{\Def{\set{\abs{D_i}{\tilde{x}_i} \defeq P_i}_{i \in I}}{P} \pll B} \npll N$:
			\[	\tree {
					\Gamma; \Delta_0, \Delta_1 \types \net{\Def{\set{\abs{D_i}{\tilde{x}_i} \defeq P_i}_{i \in I}}{P} \pll B}
					\\
					\Gamma; \Delta_0, \Delta_2 \types N
					\andalso
					\Delta_0 \text{ only $s$-endpoint }
					\andalso
					\Delta = \Delta_0, \Delta_1, \Delta_2
				}{
					\Gamma; \Delta \types \net{\Def{\set{\abs{D_i}{\tilde{x}_i} \defeq P_i}_{i \in I}}{P} \pll B} \npll N
				}{\TPar}
			\]
			We can use the above information to produce the typing derivation
			of the network term $\net{P \pll B} \npll N$
			in the premise of the reduction:
			\[
				\btree{
					\btree[] {
						\btree[\TNode] {
							\Gamma; \Delta_1'' \types P
							\andalso
							\Gamma; \Theta \types B
							\andalso
							\Delta_1' = \Delta_1'' \sessionop \Theta
						}{
							\Gamma; \Delta_1' \types \net{P \pll B}
						}{}
						\andalso
						\synchronise{\Delta_1'}{\Delta_0, \Delta_1}
						\\[9mm]
					}{
						\Gamma; \Delta_0, \Delta_1 \types \net{P \pll B}
					}{\TSynch}
					\\[14mm]
					\Gamma; \Delta_0, \Delta_2 \types N
					\andalso
					\Delta_0 \text{ only $s$-endpoint }
					\andalso
					\Delta = \Delta_0, \Delta_1, \Delta_2
					\\[1mm]
				}{
					\Gamma; \Delta \types \net{\Def{P \pll B} \npll N}
				}{\TPar}
			\]
			The results then comes from the
			induction hypothesis;
			the latter typing derivation and the reduction
			on the premise of the case implies typing derivation
			$
				\Gamma; \Delta' \types \net{\Def{P' \pll B'} \npll N'}
			$
			and
			$\Delta \tadvance \Delta'$, which concludes the case.

		\item	Rule \TPar.

			\noindent
			Assume
			\[
				\tree {
					N \reduces N'
				}{
					N \npll M \reduces N' \npll M
				}{\RPar}
			\]
			with
			\[
				\Gamma; \Delta \types N \npll M
			\]
			and $\Delta$ well-formed.

			We provide with the typing derivation for
			the latter typing judgement.
			\[
				\tree{
					\Gamma; \Delta_0, \Delta_1 \types N
					\andalso
					\Gamma; \Delta_0, \Delta_2 \types M
					\\
					\dom{\Delta_0} \text{ only $s$-endpoints}
					\andalso
					\Delta = \Delta_0, \Delta_1, \Delta_2
				}{
					\Gamma; \Delta \types N \npll M
				}{\TPar}
			\]

			From the fact that $\Gamma; \Delta_0, \Delta_1 \types N$,
			the reduction, and the induction hypothesis we can
			conclude that $\Gamma; \Delta' \types N$ and
			$\Delta_0, \Delta_1 \tadvance \Delta'$.
			It remains to type the network term
			$N \npll M$ and show the Typing Preservation
			requirements for $\Delta'$. This can be analysed
			in several sub-cases.
			\begin{itemize}
				\item	$\Delta' =  \Delta_0', \Delta_1'$ with $\Delta_0' = \Delta_0, \set{s_i: \ctype{c_i}{T_i}}_{i \in I}$
					and $\Delta_1 = \Delta_1', \set{s_j: \ctype{c_j}{T_j}}_{j \in J}$. This implies derivation:
					\[
						\btree{
							\Gamma; \Delta_0', \Delta_1' \types N'
							\andalso
							\andalso
							\tree {
								\Delta_0 = \Delta_0', \Delta_0''
								\andalso
								\Gamma; \Delta_0, \Delta_2 \types M 
							}{
								\Gamma; \Delta_0', \Delta_0'', \Delta_2 \types M
							}{}
							\\[6mm]
							\begin{array}{l}
								\dom{\Delta_0} \text{ only $s$-endpoints and }
								\Delta_0' \subseteq \Delta_0 \text{ implies }
								\\
								\dom{\Delta_0'} \text{ only $s$-endpoints}
							\end{array}
						}{
							\Gamma; \Delta_0', \Delta_0'', \Delta_1', \Delta_2 \types N \npll M
						}{\TPar}
					\]
					as required, because
					$\Delta = \Delta_0, \Delta_1, \Delta_2 \tadvance \Delta_0', \Delta_0'', \Delta_1', \Delta_2 = \Delta_0, \Delta_1', \Delta_2$

				\item	$\Delta' =  \Delta_0, \Delta_1'$ and $\Delta_1 \tadvance \Delta_1'$
					implies:
					\[
						\tree{
							\Gamma; \Delta_0, \Delta_1' \types N'
							\andalso
							\Gamma; \Delta_0, \Delta_2 \types M
							\andalso
							\dom{\Delta_0} \text{ only $s$-endpoints}
						}{
							\Gamma; \Delta_0, \Delta_1', \Delta_2 \types N \npll M
						}{\TPar}
					\]
					as required, because
					$\Delta = \Delta_0, \Delta_1, \Delta_2 \tadvance \Delta_0, \Delta_1', \Delta_2$.

				\item	$\Delta' =  \Delta_0', \Delta_1'$ with $\Delta_0, \Delta_1 \tadvance \Delta_0', \Delta_1'$.
					The latter assumption, the fact that $\dom{\Delta_0}$ contains only $s$-endpoints,
					and Linear Context Advancement (Definition~\ref{def:lin_context_adv})
					imply that
					$\Delta_0 = \Delta_0'', s: \ctype{c}{T}$
					and
					$\Delta_1 = \Delta_1'', \aggr s: \ctype{c}{\tdual T}$
					and furthermore	
					$\Delta_0 = \Delta_0'', s: \ctype{c+1}{T'}$
					and
					$\Delta_1' = \Delta_1'', \aggr s: \ctype{c+1}{\tdual{T'}}$.
					Note that $T \tadvance T'$.
					These last equations allows to type the resulting
					network of the reduction:
					\[
						\btree{
							\btree[] {
								\btree{
									\Gamma; \Delta_0, \Delta_2 \types M
									\andalso
									\Delta_0 = \Delta_0'', s: \ctype{c}{T}
								}{
									\Gamma; \Delta_0'', s: \ctype{c}{T}, \Delta_2 \types M
								}{}
								\\[6mm]
								\btree {
									\synchronise{\Delta_0'', \Delta_2}{\Delta_0'', \Delta_2}
									\andalso
									T \tadvance^{c + 1 - c} T'
								}{
									\synchronise{\Delta_0'', s: \ctype{c}{T}, \Delta_2}{\Delta_0'', s: \ctype{c+1}{T'}, \Delta_2}
								}{}
								\\[6mm]
							}{
								\Gamma; \Delta_0'', s: \ctype{c+1}{T'}, \Delta_2 \types M
							}{\TSynch}
							\\[18mm]
							\Gamma; \Delta_0'', s: \ctype{c+1}{T'}, \Delta_1'', \aggr s: \ctype{c+1}{\tdual{T'}} \types N'
							\\
							\dom{\Delta_0'', s: \ctype{c+1}{T'}} \text{ only $s$-endpoints}
							\\[1mm]
						}{
							\Gamma; \Delta_0'', s: \ctype{c+1}{T'}, \Delta_1'', \aggr s: \ctype{c+1}{\tdual{T'}}, \Delta_2 \types N \npll M
						}{\TPar}
					\]
					as required, because
					$\Delta = \Delta_0, \Delta_1, \Delta_2 \tadvance \Delta_0', \Delta_1', \Delta_2 = \Delta_0', s: \ctype{c+1}{T'}, \Delta_1'', \aggr s: \ctype{c+1}{\tdual{T'}}, \Delta_2$.

				\item	$\Delta = \Delta_0', \Delta_1$ and $\Delta_0 \tadvance \Delta_0'$ cannot happen
					because $\dom{\Delta_0}$ has only $s$-endpoints.
			\end{itemize}

		\item	Rule \RCong. The result is immediate by applying Congruence
				Invariant (Lemma~\ref{lem:congruence-invariance})
				on the induction hypothesis.


		\item	Rule \RRes.

			\noindent
			Assume
			\[
				\tree {
					N \reduces N'
				}{
					\newn{n} N
					\reduces
					\newn{n} N'
				}{\RRes}
			\]
			with
			\[
				\Gamma; \Delta \types \newn{n} N
			\]
			and $\Delta$ well-formed.
			
			There are two sub-cases for the typing derivation of the latter
			judgement:
			\begin{itemize}
				\item	$n = a$. Restrict shared name. In this sub-case
					the typing derivation for network term $\newn{a} N$
					is as:
					
					\[
						\tree {
							\Gamma, a : T; \Delta \types N
						}{
							\Gamma; \Delta \types \newn{a} N
						}{\TCRes}
					\]
					From $\Gamma, a: T; \Delta \types N$,
					well-formed $\Delta$, and the induction
					hypothesis we conclude that
					$\Gamma, a : T; \Delta' \types N'$
					and either
					$\Delta' \subseteq \Delta$
					or
					$\Delta \tadvance \Delta'$.
					From the last typing judgement
					we can type the result of the
					reduction $\newn{a} N$:
					\[
						\tree {
							\Gamma, a: T; \Delta' \types N'
						}{
							\Gamma; \Delta' \types N'
						}{\TCRes}
					\]
					to conclude the sub-case.

				\item	$n = s$. Restrict session name. In this sub-case
					the typing derivation for network term $\newn{s} N$
					is as:

					\[
						\tree[] {
							\Gamma; \Delta, \Delta''', \aggr s: \ctype{c}{T} \types N
							\andalso
							\Delta''' = s: \ctype{c}{\dual{T}} \lor
							\Delta''' = \econtext
						}{
							\Gamma;\Delta \types \newn{s} N
						}{\TSRes}
					\]
					Typing context $\Delta, \aggr s: \ctype{c}{T}$
					is well-formed because $\Delta$ is well-formed,
					see well-formedmess Definition~\ref{def:well-formed}.

					From $\Gamma; \Delta, \aggr s: \ctype{c}{T}, s: \ctype{c}{\dual T} \types N$,
					well-formed $\Delta, \aggr s: \ctype{c}{T}, s: \ctype{c}{\dual T}$,
					and the induction hypothesis we conclude that
					$\Gamma; \Delta' \types N'$
					and
					$\Delta, \aggr s: \ctype{c}{T}, s: \ctype{c}{\dual T} \tadvance \Delta'$.
					All it remains is to conclude te case, is the
					typing derivation of the result of the reduction
					$\newn{s} N'$.

				If $\Delta, \aggr s: \ctype{c}{T}, s: \ctype{c}{\tdual T} \tadvance \Delta'$ then
				there are two sub-cases:
				\begin{itemize} 
					\item	$\Delta' = \Delta'', \aggr s: \ctype{c}{T}, s: \ctype{c}{\tdual T}$ with $\Delta \tadvance \Delta''$,
						which implies
						\[
							\tree {
								\Gamma; \Delta'', \aggr s: \ctype{c}{T}, s: \ctype{c}{\tdual T}  \types N'
							}{
								\Gamma; \Delta'' \types \newn{s} N'
							}{\TSRes}
						\]
						as required.
					\item	$\Delta' = \Delta, \aggr s: \ctype{c}{T}, s: \ctype{c}{\dual{T}}$ with $\aggr s: \ctype{c}{T}, s: \ctype{c}{\dual{T}} \tadvance \aggr s: \ctype{c+1}{T_1}, s: \ctype{c+1}{T_2}$,
						which implies
						\[
							\tree {
								\Gamma; \Delta, \aggr s: \ctype{c+1}{T'}, s: \ctype{c+1}{\tdual{T'}}  \types \newn{n} N'
								\\
								\text{From the definition of $\tadvance$ we get that } 
								T_1' = \tdual T_2'
							}{
								\Gamma; \Delta \types \newn{s} N'
							}{\TSRes}
						\]
						as required. \qedhere
				\end{itemize}
			\end{itemize}

	\end{itemize}

\end{proof}

\end{document}